%% file: main.tex
\documentclass[runningheads]{llncs}
\usepackage[T1]{fontenc}
\usepackage[utf8]{inputenc}
\usepackage[T1]{fontenc}
\usepackage{lipsum}        
\usepackage{amsmath,amssymb,amsfonts}
\usepackage{algorithmicx}
\usepackage{graphicx}
\usepackage{cite}
\usepackage{url}
\usepackage{booktabs}
\usepackage[hidelinks]{hyperref}
\usepackage{verbatim}
\usepackage{array}
\usepackage{booktabs}
\usepackage{subcaption}
\usepackage{algorithmicx}
\usepackage[noend]{algpseudocode}
\usepackage{xcolor}

\newcommand{\Vars}{\mathrm{Vars}}
\usepackage{algorithm}
\usepackage{algorithmicx}
\usepackage{algpseudocode}
\algtext*{EndIf}
\algtext*{EndFor}

\newcommand{\AlgLine}[2]{%
  \Statex\hspace*{-\algorithmicindent}%
  \textbf{#1:} #2%
}

\newcommand{\Input}[1]{\AlgLine{Input}{#1}}
\newcommand{\Output}[1]{\AlgLine{Output}{#1}}
\renewcommand{\Require}[1]{\AlgLine{Requires}{#1}}

\newcommand{\URG}{\mathsf{URG}}
\newcommand{\WRN}{\mathsf{WRN}}
\newcommand{\MOV}{\mathsf{MOV}}
\newcommand{\DCL}{\mathsf{DCL}}
\newcommand{\STP}{\mathsf{STP}}

\usepackage{float}

\newcommand{\qedbox}{\hfill$\blacksquare$}

\usepackage{amsmath,amssymb,mathtools,bm}
\usepackage{stmaryrd}   
\usepackage{physics}
\usepackage{bbm}

\usepackage{tikz}
\usetikzlibrary{arrows.meta,positioning,calc,decorations.pathreplacing}
\usetikzlibrary{tikzmark}

\newcommand{\Always}{\Box}
\newcommand{\Eventually}{\Diamond}
\newcommand{\Next}{\bigcirc}
\newcommand{\True}{\mathsf{true}}
\newcommand{\False}{\mathsf{false}}

\newenvironment{proofsketch}
  {\par\noindent\textit{Proof sketch.}\ }
  {\hfill$\blacksquare$\par}

\usepackage{enumitem}
\setlist[itemize]{noitemsep, topsep=0pt}
\setlist[enumerate]{noitemsep, topsep=0pt}

\usepackage{graphicx}
\begin{document}
%
\title{Computing GR(1) Assume--Guarantee Contracts}
%
%
\author{
Yi\u{g}it Narter
\and
Emily Zhao
\and
Inigo Incer
}
%
\authorrunning{Y. Narter et al.}
\institute{Electrical and Computer Engineering,\\
University of Michigan, Ann Arbor, MI 48109, USA\\
\email{\{\href{mailto:narter@umich.edu}{narter},\href{mailto:emzhao@umich.edu}{emzhao},\href{mailto:iir@umich.edu}{iir}\}@umich.edu}}
\maketitle              


\begin{abstract}
Assume--guarantee contracts provide a set of algebraic operations that enable the efficient compositional design and analysis of complex engineering systems. The explicit computation of these operations at component interfaces is currently supported only for some logics lacking a notion of time. This paper develops theory and algorithms to compute these operations when the contracts are expressed in the GR(1) modulo theory. The key challenge is the computation of two quantifier elimination problems for temporal formulas.
We develop context-aided variable elimination procedures on GR(1)-modulo formulas, establish their soundness, characterize subproblem-wise optimal solutions within an admissible class of formulas, and use the procedures to compute contract operations, including composition and quotient. We implement the resulting framework in Pacti, automating compositional reasoning over GR(1) contracts. Finally, we demonstrate the framework on an automated emergency-braking case study, where we compute system-level contracts and expose compatibility conditions arising from the interaction of component specifications. Experiments on synthetic families show that the cost is governed by the number of eliminated variables and liveness conjuncts, and a comparison with exact automata-based elimination shows complementary scaling.
\end{abstract}

\keywords{Assume--guarantee contracts \and GR(1) \and Linear temporal logic \and Compositional reasoning \and Contract-based design \and Quantifier elimination}

\section{Introduction}
Safety-critical systems increasingly rely on software that is built faster than our ability to verify it across all operating conditions.
A 2025 joint report by CISA, DARPA, OUSD (R\&E), and the NSA~\cite{cisa2025} refers to this gap between what we design and what we can verify as the ``software understanding gap.''
The technical call to action of this report states that ``all suitable techniques (including formal methods, AI, and others) must be leveraged to develop rigorous, reliable, rapid, and inexpensive capabilities'' to address such understanding.

A key principle for designing and analyzing complex systems is compositionality.
Over the last 20 years, assume--guarantee contracts have been proposed as a mathematical framework that supports tractable compositional design by formalizing system integrations across heterogeneous components~\cite{beneveniste2008, frankenstein}.
Contracts are interface specifications of components that encode information usually contained in the datasheet of a system; they define \emph{assumptions} a component makes on its context of operation together with \emph{guarantees} expected from that component whenever the assumptions hold~\cite{rouquette2026early}. They supply algebraic operations that enable system-level design and analysis \cite{incer2026some}. A recently introduced open-source platform for contract-based design, Pacti, supports the efficient computation of operations on contracts whose assumptions and guarantees are expressed in logical specification languages that are accessible to practitioners \cite{pacti}. 

This paper studies the computation of the algebra of contracts when the specifications are expressed in 
a fragment of Linear Temporal Logic (LTL)~\cite{ltl}, which can capture safety and liveness properties of discrete and hybrid event systems. This is desirable in domains such as computer architecture~\cite{armoni2002forspec}, automotive systems~\cite{annpureddy2011s}, and multi-agent robotics~\cite{alur2016compositional}. 
Pacti supports contracts over propositional logic and
linear real arithmetic, but not temporal specifications. Temporal contract tools such as OCRA~\cite{cimatti2013ocra}, AGREE~\cite{cofer2012compositional}, and CHASE~\cite{nuzzo2018chase} verify contract refinement, and algebraic
theories~\cite{beneveniste2008,incer2026some,hypercontracts} define contract operations semantically. \textbf{To the best of our knowledge,
it is not currently known how to compute the contract of an interconnection (composition) or of a missing
component (quotient) over interface variables when the assumptions and guarantees are temporal formulas.}
The difficulty lies in variable elimination.
It was shown in~\cite{pacti} that 
composition and quotient must eliminate variables outside the interface of the resulting contract. Concretely, we must compute $\forall\bar p.(\Gamma\Rightarrow\phi)$ or $\exists\bar p.(\Gamma\wedge\phi)$ for a formula $\phi$ and a context $\Gamma$.
Automata-based quantification~\cite{qptl,sistla1987complementation} eliminates existentially quantified
variables directly, but universal quantification needs complementation, which is
exponential in the worst case~\cite{sistla1987complementation}, and the result is an automaton whose language may not be expressible in LTL~\cite{wolper1983temporal}.
We restrict our attention to the GR(1) fragment of LTL, a common format for system
specifications with initial, invariant, and liveness clauses~\cite{gr1, maoz2015gr, ehlers2024fully}. The GR(1) format conveniently relates to contracts and offers useful properties for computing variable elimination. Exact elimination can leave this fragment, so our procedures trade exactness for soundness, and each operation again yields a contract in GR(1) form.

\noindent
\textbf{Contributions.}
We enable the computation of system-level design operations on contracts expressed in GR(1)-modulo theory. We provide procedures for the two GR(1) context-aided variable elimination problems behind composition and quotient (Section~\ref{sec:gr1-cave}), computed soundly by SMT-based quantifier elimination (QE). Our procedures attain closure within GR(1) and subproblem-wise optimality within an admissible class of GR(1) formulas. We implement the procedures in the Pacti~\cite{pacti} platform for contract-based design. We also provide an automated emergency-braking case study (Section~\ref{sec:gr1-contract-operations}) deriving system-level temporal contracts and compatibility conditions.
The appendices provide the complete soundness and optimality proofs (Appendices~\ref{app:soundness}-\ref{app:optimality}), the contract-operation algorithms (Appendix~\ref{app:gr1-contract-operations}), and additional computational experiments (Appendix~\ref{app:scalability}), including bounded-lasso soundness and precision checks and comparisons with automata-based elimination using Spot~\cite{duret2022spot}.

\begin{figure}[htb]
    \centering
    \includegraphics[width=\linewidth]{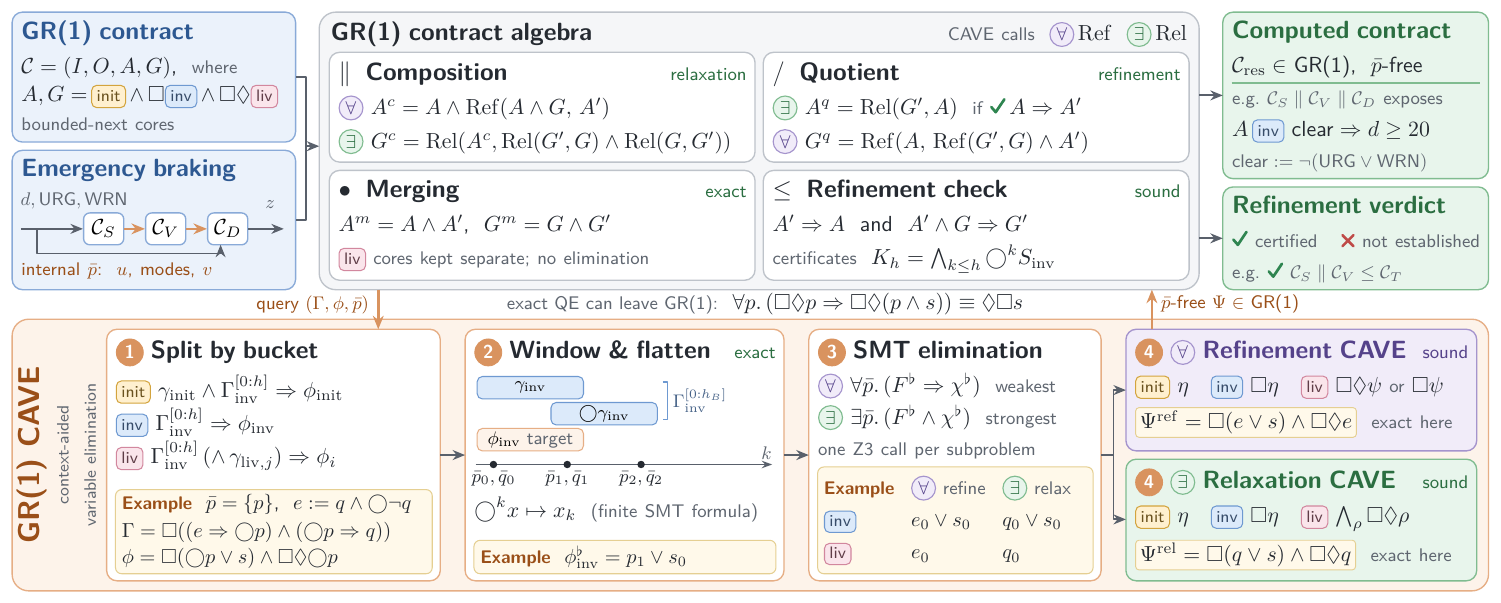}
    \caption{%
        Role of CAVE in the GR(1) contract-algebra workflow.
        Contract operations may produce formulas outside the
        GR(1) representation; CAVE constructs a GR(1) result through
        refinement or relaxation-based abstraction.
    }
    \label{fig:cave-overview}
\end{figure}

\section{Preliminaries}
\label{sec:preliminaries}

\subsection{Linear Temporal Logic}

Linear temporal logic (LTL) \cite{ltl} specifies properties of infinite traces over a set
of atomic propositions \(AP\). A trace is an infinite sequence
\(\sigma=s_0s_1s_2\cdots\), where each \(s_t\subseteq AP\) records the
propositions true at time \(t\). LTL formulas are generated by the syntax
$\varphi ::= \top \mid p \mid \neg\varphi \mid \varphi\land\varphi \mid
\bigcirc\varphi \mid \varphi\mathbin{\mathsf U}\varphi$.
Here, $p\in AP$ is an atomic proposition, and $\bigcirc$ and $\mathsf U$ denote
next and until. Their semantics are $(\sigma,t)\models\bigcirc\varphi$ iff
$(\sigma,t+1)\models\varphi$, and
$(\sigma,t)\models\varphi_1\mathbin{\mathsf U}\varphi_2$ iff there is $t'\ge t$
with $(\sigma,t')\models\varphi_2$ and $(\sigma,t'')\models\varphi_1$ for all
$t\le t''<t'$. The Boolean operators have their usual semantics. We use the derived
operators $\Diamond\varphi:=\top\mathbin{\mathsf U}\varphi$ (eventually) and
$\Box\varphi:=\neg\Diamond\neg\varphi$ (always), so $(\sigma,t)\models\Diamond\varphi$
iff $(\sigma,t')\models\varphi$ for some $t'\ge t$, and $(\sigma,t)\models\Box\varphi$
iff $(\sigma,t')\models\varphi$ for every $t'\ge t$. We write
\(\sigma\models\varphi\) for \((\sigma,0)\models\varphi\). For specifications modulo a background theory $\mathcal T$, atomic formulas
may also be predicates over typed variables. A trace assigns each such variable
a value in its domain at every time position, with atomic predicates evaluated
under the corresponding theory interpretation.

\subsection{The GR(1) Fragment}
\label{sec:gr1-fragment}

Generalized Reactivity (1), or GR(1), is a structured fragment of LTL that
supports efficient symbolic reasoning for reactive systems
\cite{gr1,maoz2015gr,ehlers2024fully,ehlers2016slugs,majumdar2019environmentally}. A standard GR(1) specification has the
form
\[
  \left(
    \varphi_{\mathrm{init}}
    \wedge
    \Always\varphi_{\mathrm{inv}}
    \wedge
    \bigwedge_i\Always\Eventually \varphi_{\mathrm{liv},i}
  \right)
  \Rightarrow
  \left(
    \psi_{\mathrm{init}}
    \wedge
    \Always\psi_{\mathrm{inv}}
    \wedge
    \bigwedge_i\Always\Eventually \psi_{\mathrm{liv},i}
  \right),
\]
where the left-hand side of the implication contains initial, invariant, and liveness requirements made on the operational environment, and the right-hand side contains the corresponding system requirements.

Throughout this work, the sub-formulas
$\varphi_{\mathrm{init}},\varphi_{\mathrm{inv}},\varphi_{\mathrm{liv},i},\psi_{\mathrm{init}},\psi_{\mathrm{inv}},$ and $\psi_{\mathrm{liv},i}$
are built from Boolean connectives and from propositions or predicates over typed
variables in a background theory, such as linear real arithmetic
\cite{becchi2025abstraction}, using finitely many applications of \(\Next\) and no
other temporal operators. We call such a sub-formula a \emph{bounded-next core}, or
simply a \emph{core}, of the GR(1)-modulo formula. A term $\Next^k x$ can be mapped to an indexed variable $x_k$
of the same type, allowing each bounded-next core to be reduced to a finite
propositional or SMT formula. This representation later enables algorithms for
computing GR(1) contract operations.

\subsection{Assume--Guarantee Contracts}
\label{sec:assume-guarantee-contracts}
Assume--guarantee contracts are interface specifications of the form \(\mathcal C=(A,G)\), where \(A\) states assumptions the component makes on its input variables and
\(G\) states guarantees the component promises to deliver on its outputs whenever the assumptions hold. A valid implementation of a contract is a model that satisfies \(A\Rightarrow G\). 
A valid environment for the contract is a model that satisfies $A$.

\paragraph{Contract Operations.}
Algebraic operations on contracts relate the contracts of interconnected
components with the system-level contract. We use the contract algebra underlying Pacti~\cite{pacti}. At the semantic
level, the exact operations are defined independently of the
specification theory used to express assumptions and guarantees. For contracts
$\mathcal C=(A,G)$ and $\mathcal C'=(A',G')$, \emph{refinement} is defined by
$\mathcal C\leq\mathcal C'$ iff
$(A'\Rightarrow A)\land
((G\lor\neg A)\Rightarrow(G'\lor\neg A'))$;
it captures substitutability, with $\mathcal C$ satisfying at least the
requirements of $\mathcal C'$.
\emph{Composition} combines component contracts into a contract for their
interconnection and is defined by
$\mathcal C\parallel\mathcal C'
=((A\land A')\lor(A\land\neg G)\lor(A'\land\neg G'),
(G\lor\neg A)\land(G'\lor\neg A'))$.
\emph{Quotient} derives the largest contract for a missing component that,
when composed with an existing component, refines a target contract:
$\mathcal C/\mathcal C'
=(A\land(\neg A'\lor G'),
(A'\land G)\lor\neg A\lor(A'\land\neg G'))$.
It satisfies the residual property
$\mathcal C'\parallel\mathcal X\leq\mathcal C
\Leftrightarrow
\mathcal X\leq\mathcal C/\mathcal C'$.
Finally, \emph{merging} combines requirements that must hold simultaneously
and is defined by
$\mathcal C\bullet\mathcal C'
=(A\land A',
(G\land G')\lor\neg A\lor\neg A')$.

\paragraph{Context Aided Variable Elimination.} We distinguish the exact expressions for operations from their computational
realizations, which, as in Pacti~\cite{pacti}, require specification-theory-specific procedures.
Consider the composition operation, which derives a system-level contract based on the assumptions and guarantees of its constituent components. During contract interconnection, some variables that appear in the assumptions and guarantees become \emph{irrelevant}; they help describe component interfacing, but should not appear in the contract returned to a user. In composition, these are the internal variables that connect the components, and in quotient, they are the variables external to the missing component. Computing contract operations requires eliminating irrelevant variables while preserving the information supplied by the formulas they appear in. Generally, this elimination cannot be performed by simply deleting the variables or treating them in isolation, because their admissible values are constrained by the rest of the interconnected
specification. As pointed out in \cite{pacti}, computing operations on contracts reduces to solving two such \emph{context-aided variable elimination (CAVE)} problems for the background theory that the assumptions and guarantees are expressed in. 

\begin{enumerate}
    \item \emph{CAVE by refinement}: Given a formula \(\phi(x,y)\), a context \(\Gamma(x,y,z)\), and an irrelevant variable \(y\), find a \(y\)-free formula \(\Psi(x,z)\) satisfying \(\Gamma\wedge\Psi\models\phi\). The exact logical solution is \(\Psi\colon \forall y.(\Gamma\Rightarrow\phi)\).
    \item \emph{CAVE by relaxation}: Given a formula \(\phi(x,y)\), a context \(\Gamma(x,y,z)\), and an irrelevant variable \(y\), find a \(y\)-free consequence \(\Psi(x,z)\) satisfying \(\Gamma\wedge\phi\models\Psi\). The exact solution is \(\Psi\colon \exists y.(\Gamma\wedge\phi)\). 
\end{enumerate} 
These are not auxiliary operations; CAVE procedures are the primitives required to compute user-interpretable contract
operations by eliminating variables inside the desired interface. Figure~\ref{fig:cave-overview} illustrates the contract-algebraic workflow. 


\subsection{GR(1) Specifications as Contracts}

We model GR(1) specifications using assume--guarantee contracts. In general, a contract may
contain multiple assumptions and guarantees. Since these are interpreted
conjunctively, they can be collected without loss of generality into
\(A:=\bigwedge_k A_k\) and \(G:=\bigwedge_\ell G_\ell\). To match the structure of GR(1) specifications, we consider assumptions and guarantees with forms
\(A=A_{\mathrm{init}}\wedge A_{\mathrm{inv}}\wedge A_{\mathrm{liv}}\) and
\(G=G_{\mathrm{init}}\wedge G_{\mathrm{inv}}\wedge G_{\mathrm{liv}}\), where
\(A_{\mathrm{init}}=\varphi_{\mathrm{init}}\),
\(A_{\mathrm{inv}}=\Always\varphi_{\mathrm{inv}}\), and
\(A_{\mathrm{liv}}=\bigwedge_i\Always\Eventually\varphi_{\mathrm{liv},i}\), and similarly
\(G_{\mathrm{init}}=\psi_{\mathrm{init}}\),
\(G_{\mathrm{inv}}=\Always\psi_{\mathrm{inv}}\), and
\(G_{\mathrm{liv}}=\bigwedge_i\Always\Eventually\psi_{\mathrm{liv},i}\).
Thus, a GR(1) specification can be translated directly into a contract
\(\mathcal C=(A,G)\). The sub-formulas
\(\varphi_{\mathrm{init}}\), \(\varphi_{\mathrm{inv}}\), \(\varphi_{\mathrm{liv},i}\),
\(\psi_{\mathrm{init}}\), \(\psi_{\mathrm{inv}}\), and \(\psi_{\mathrm{liv},i}\)
are the corresponding initial, invariant, and liveness bounded-next cores. 

The preceding notions provide the logical and contract-theoretic ingredients needed for our formulation. Our objective henceforth is to derive CAVE
procedures for bounded-next GR(1) formulas, so that the contract-operation
algorithms can be instantiated for these GR(1) contracts. These procedures are introduced in Section~\ref{sec:gr1-cave}. 
Next, we set up an example that we will use to illustrate our procedures throughout the rest of the work.

\subsection{Running Example: Automated Emergency-Braking}
\label{sec:aeb-running-example}

We use an automated emergency-braking (AEB) scenario as a running example,
adapted from the vehicle-in-the-loop automated-driving case study of
Zhang et al.~\cite{zhang2025vehicle}, which combines an adaptive-cruise
controller with a safety filter that modifies the commanded longitudinal
velocity according to traffic conditions. We encode it as a bounded-next GR(1)
abstraction with a safety supervisor $\mathcal C_S=(A_S,G_S)$, a
vehicle-response component $\mathcal C_V=(A_V,G_V)$, and a distance-to-obstacle
prediction component $\mathcal C_D=(A_D,G_D)$ that we add. The goal is to
guarantee a clearance between the vehicle and a detected obstacle at every time
step, and the contract operations must derive what this requires of the
components and of the environment.

Here, $d\in\mathbb R$ is the measured obstacle distance, $u\in\mathbb R$ the
longitudinal command produced by the supervisor, $v\in\mathbb R$ the vehicle
speed, and $z\in\mathbb R$ the predicted clearance. The Booleans $\URG$ and
$\WRN$ denote urgent and warning traffic conditions, and the supervisor has three
Boolean modes, $\MOV$ (\emph{moving}), $\DCL$ (\emph{decelerating}), and $\STP$
(\emph{stopped}). The architecture is shown in Figure~\ref{fig:aeb-running}.

\begin{figure}[t]
\centering
  \resizebox{\linewidth}{!}{\input{aeb_architecture}}
\caption{%
Running AEB example, adapted from the automated-driving architecture
of Zhang et al.~\cite{zhang2025vehicle} and augmented with our
predicted-clearance component $\mathcal C_D$. Composition eliminates
$\{u,\MOV,\DCL,\STP,v\}$ along the three-component chain, leaving $z$
externally visible.
}
\label{fig:aeb-running}
\end{figure}

Table~\ref{tab:aeb-contracts} summarizes $\mathcal C_S$, $\mathcal C_V$, and $\mathcal C_D$.
The supervisor assumes an environment that starts in the clear region
($20\le d\le30$), whose conditions follow the obstacle distance ($\URG$ for
$d\le10$, $\WRN$ for $10<d\le15$, neither beyond), and that returns infinitely
often to the clear condition $\neg\URG\land\neg\WRN$, which strengthens the separate recurrence of $\neg\URG$ and $\neg\WRN$ assumed by Zhang et al.~\cite{zhang2025vehicle}, since the latter allows $\URG$ and $\WRN$ to alternate forever and prevent motion. It guarantees exactly one
active mode ($\operatorname{OneHot}(\MOV,\DCL,\STP)$), starts moving with command
$u=d/10$, commands stopping under $\URG$, deceleration under $\WRN$ unless
already stopped and motion otherwise, constrains $u$ by the obstacle distance
and the commanded mode, and guarantees motion infinitely often.

The vehicle assumes commands that are admissible for the commanded mode and
motion commanded infinitely often, which the supervisor liveness guarantee
provides. It guarantees nonnegative speed and a tracking error of $\frac12$
around the commanded value at the next step, with no liveness guarantee. Zhang et
al.~\cite{zhang2025vehicle} supply the $10\,\mathrm m$ and $15\,\mathrm m$
hazard thresholds; the remaining values are illustrative, with
$\epsilon=0.5\,\mathrm{m/s}$, $\Delta t=1\,\mathrm s$,
$d_s=1\,\mathrm m$, and $v_{\mathrm{prog}}=3/2\,\mathrm{m/s}$.
The contract $\mathcal C_D$ (Table~\ref{tab:aeb-contracts}) predicts the clearance
under constant speed over the one-second horizon, $z_{t+1}=d_t-\Delta t\,v_{t+1}$,
and carries no recurrent state.

\begin{table}[t]
\centering
\caption{GR(1) contracts $\mathcal C_S$, $\mathcal C_V$, and $\mathcal C_D$ for the AEB running example.}
\label{tab:aeb-contracts}
\scriptsize
\setlength{\tabcolsep}{5pt}
\renewcommand{\arraystretch}{1.2}

\resizebox{\columnwidth}{!}{%
\begin{tikzpicture}
\node[inner sep=0pt] (aebtable) {%
\begin{tabular}{
    @{}
    >{\raggedright\arraybackslash}m{0.49\textwidth}
    >{\raggedright\arraybackslash}m{0.49\textwidth}
    @{}
}
\toprule

\multicolumn{1}{c}{
    \textbf{Safety supervisor $\mathcal C_S=(A_S,G_S)$}
}
&
\multicolumn{1}{c}{
    \textbf{Vehicle response $\mathcal C_V=(A_V,G_V)$}
}
\tabularnewline

\midrule

\begin{tabular}{@{}l@{}}
\textbf{Inputs:}\quad $d,\URG,\WRN$\\
\textbf{Outputs:}\quad $u,\MOV,\DCL,\STP$
\end{tabular}
&
\begin{tabular}{@{}l@{}}
\textbf{Inputs:}\quad $u,\MOV,\DCL,\STP$\\
\textbf{Output:}\quad $v$
\end{tabular}
\tabularnewline

\midrule

\multicolumn{1}{c}{
    \begin{tabular}{@{}c@{}}
    \textbf{Assumptions}\\[-0.3mm]
    $A_S
    =A_{S,\mathrm{init}}
    \land A_{S,\mathrm{inv}}
    \land A_{S,\mathrm{liv}}$
    \end{tabular}
}
&
\multicolumn{1}{c}{
    \begin{tabular}{@{}c@{}}
    \textbf{Assumptions}\\[-0.3mm]
    $A_V
    =A_{V,\mathrm{init}}
    \land A_{V,\mathrm{inv}}
    \land A_{V,\mathrm{liv}}$
    \end{tabular}
}
\tabularnewline

\begin{tabular}{@{}l@{}}
$A_{S,\mathrm{init}}$:\quad
$20\le d\le30\land\neg\URG\land\neg\WRN$
\\[1.5mm]
$A_{S,\mathrm{inv}}$:\quad
$\displaystyle
\Always\!\left(
\begin{aligned}
&\,\,\,0\le d\le30\\
&\land(d\le10\Rightarrow\URG\land\neg\WRN)\\
&\land(10<d\le15\Rightarrow\WRN\land\neg\URG)\\
&\land(d>15\Rightarrow\neg\URG\land\neg\WRN)
\end{aligned}
\right)$
\\[7mm]
$A_{S,\mathrm{liv}}$:\quad
$\Always\Eventually(\neg\URG\land\neg\WRN)$
\end{tabular}
&
\begin{tabular}{@{}l@{}}
$A_{V,\mathrm{init}}$:\quad $2\le u\le3$
\\[1.5mm]
$A_{V,\mathrm{inv}}$:\quad
$\displaystyle
\Always\!\left(
\begin{aligned}
&\,\,\,0\le u\le6\\
&\land(\Next\STP\Rightarrow u=0)\\
&\land(\Next\DCL\Rightarrow u\le3)\\
&\land(\Next\MOV\Rightarrow u\ge2)
\end{aligned}
\right)$
\\[7mm]
$A_{V,\mathrm{liv}}$:\quad $\Always\Eventually\MOV$
\end{tabular}
\tabularnewline

\midrule

\multicolumn{1}{c}{
    \begin{tabular}{@{}c@{}}
    \textbf{Guarantees}\\[-0.3mm]
    $G_S
    =G_{S,\mathrm{init}}
    \land G_{S,\mathrm{inv}}
    \land G_{S,\mathrm{liv}}$
    \end{tabular}
}
&
\multicolumn{1}{c}{
    \begin{tabular}{@{}c@{}}
    \textbf{Guarantees}\\[-0.3mm]
    $G_V
    =G_{V,\mathrm{init}}
    \land G_{V,\mathrm{inv}}
    \land G_{V,\mathrm{liv}}$
    \end{tabular}
}
\tabularnewline

\begin{tabular}{@{}l@{}}
$G_{S,\mathrm{init}}$:\quad
$\MOV\land\neg\DCL\land\neg\STP\land u=d/10$
\\[1.5mm]
$G_{S,\mathrm{inv}}$:\quad
$\displaystyle
\Always\!\left(
\begin{aligned}
&\,\,\,\operatorname{OneHot}(\MOV,\DCL,\STP)\\
&\land(\URG\land\neg\WRN\Rightarrow\Next\STP)\\
&\land\bigl(
\WRN\land\neg\URG\land(\MOV\lor\DCL)\\
&\hspace{10mm}\Rightarrow\Next\DCL
\bigr)\\
&\land\bigl(
\WRN\land\neg\URG\land\STP\\
&\hspace{10mm}\Rightarrow\Next\STP
\bigr)\\
&\land(\neg(\URG\lor\WRN)\Rightarrow\Next\MOV)\\
&\land 0\le u\le6\land u\le d/5\\
&\land(\Next\STP\Rightarrow u=0)\\
&\land(\Next\DCL\Rightarrow u\le3)\\
&\land(\Next\MOV\Rightarrow u\ge d/10)
\end{aligned}
\right)$
\\[20mm]
$G_{S,\mathrm{liv}}$:\quad $\Always\Eventually\MOV$
\end{tabular}
&
\begin{tabular}{@{}l@{}}
$G_{V,\mathrm{init}}$:\quad $v=u$
\\[1.5mm]
$G_{V,\mathrm{inv}}$:\quad
$\displaystyle
\Always\!\left(
v\ge0
\land
u-\frac12\le\Next v\le u+\frac12
\right)$
\\[1.5mm]
$G_{V,\mathrm{liv}}$:\quad ---
\\[2mm]
\multicolumn{1}{@{}l@{}}{\hspace*{-\tabcolsep}\rule{\dimexpr\linewidth+\tabcolsep\relax}{1.5pt}}
\\[0.5mm]
\multicolumn{1}{c}{\textbf{Clearance prediction $\mathcal C_D=(A_D,G_D)$}}
\\[0.8mm]
\multicolumn{1}{c}{\textbf{Inputs:}\quad $d,v$\qquad\textbf{Output:}\quad $z$}
\\[2.5mm]
\multicolumn{1}{c}{$A_D=G_{D,\mathrm{init}}=G_{D,\mathrm{liv}}=\top$}
\\[1.0mm]
\multicolumn{1}{c}{$G_{D,\mathrm{inv}}$:\quad $\Always(\Next z=d-\Delta t\,\Next v)$}
\end{tabular}
\tabularnewline

\bottomrule
\end{tabular}%
};

\draw[line width=0.4pt]
    ([yshift=-5.3mm]aebtable.north)
    --
    (aebtable.south);

\end{tikzpicture}}
\end{table}

For quotient, we use a safety target $\mathcal C_T=(A_T,G_T)$ with inputs
$\{d,\allowbreak\URG,\allowbreak\WRN\}$ and output $\{v\}$. Its assumptions extend those of
$\mathcal C_S$ with $d\ge3/2$ and with $d\ge20$ whenever neither condition is
active, namely $A_{T,\mathrm{init}}:=A_{S,\mathrm{init}}$,
$A_{T,\mathrm{inv}}
:=
A_{S,\mathrm{inv}}
\land
\Always(d\ge3/2)
\land
\Always(\neg(\URG\lor\WRN)\Rightarrow d\ge20),$
and $A_{T,\mathrm{liv}}:=A_{S,\mathrm{liv}}$.
Its guarantees are $G_{T,\mathrm{init}}:=\top$,
$G_{T,\mathrm{inv}}:=\Always(\Next v\le(d-d_s)/\Delta t)$, a one-step safety
bound, and $G_{T,\mathrm{liv}}:=\Always\Eventually(
\neg(\URG\lor\WRN)\land\Next v\ge v_{\mathrm{prog}})$, which synchronizes
forward progress with recurring clear opportunities. Since the vehicle tracks
the command $u$ within $\pm\epsilon$, these guarantees translate into the
following sufficient conditions on $u$: $u\le\frac{d-d_s}{\Delta t}-\epsilon$ and
$\Always\Eventually\!\left(
\neg(\URG\lor\WRN)\land u\ge v_{\mathrm{prog}}+\epsilon\right)$. The quotient
$\mathcal C_T/\mathcal C_V$ is the supervisor-side component that meets
$\mathcal C_T$ with the fixed vehicle.
Composing $\mathcal C_S$ and $\mathcal C_V$ eliminates the internal interface
$\{u,\MOV,\DCL,\STP\}$, so
\(
\mathcal C_{\mathrm{comp}}
:=
\mathcal C_S\parallel\mathcal C_V
\)
has inputs $\{d,\URG,\WRN\}$ and output $\{v\}$. Composing it with
$\mathcal C_D$ additionally eliminates $v$, yielding the same inputs and output
$\{z\}$. Section~\ref{sec:gr1-contract-operations} computes both operations on
this example with the CAVE procedures, where refinement CAVE generates external
requirements and relaxation CAVE propagates externally visible guarantees.

\section{Related Work}
\label{sec:related-work}

\paragraph{Quantification in Temporal Logic.}
Propositional quantification in modal and temporal logic has been studied from
early foundations by Carnap~\cite{Carnap_1946} to the modern account of
Fritz~\cite{Fritz_2024}. QPTL~\cite{qptl} extends LTL with quantification over atomic
propositions, and Bolotov~\cite{6004288} develops a deductive verification
method for QPTL. Automata-based treatments eliminate existentially quantified
propositions, but universal quantification requires complementing a B\"uchi
automaton~\cite{sistla1987complementation}. These works give semantics and decision
procedures for quantified formulas, not GR(1) initial, invariant, and liveness
results; CAVE instead eliminates irrelevant variables by refinement and
relaxation to obtain quantifier-free GR(1) results. We compared CAVE for refinement with an exact automata-based pipeline built on Spot~\cite{duret2022spot} (Appendix~\ref{app:eval-automata}). The two have complementary scaling, with Spot faster for few liveness conjuncts or many irrelevant variables and CAVE otherwise. CAVE returns a GR(1) formula with initial, invariant, and liveness parts, whereas the exact result is a B\"uchi automaton.

\paragraph{Contract Theories and Tools for Temporal Logic.}
Algebraic contract theories~\cite{beneveniste2008,incer2026some} and
hypercontracts~\cite{hypercontracts} define refinement, composition, and
quotient semantically, and Pacti~\cite{pacti} computes these operations for
propositional and linear-arithmetic specifications. For temporal logic,
OCRA~\cite{cimatti2013ocra} checks that the contracts of a decomposition refine
the contract of their parent by generating temporal-logic proof obligations, and
its tightening extension~\cite{cimatti2017tightening} searches for weaker parent
assumptions and subcomponent guarantees by parameter synthesis.
AGREE~\cite{cofer2012compositional} uses assume--guarantee reasoning to compositionally verify architectural models against component contracts,
CHASE~\cite{nuzzo2018chase} formalizes requirements in a pattern-based language to check them for inconsistencies, and CROME~\cite{mallozzi2020crome} states the composition and conjunction of LTL goal contracts syntactically with logical connectives, passing the resulting formulas to a model checker. These tools verify or tune refinements supplied by the designer, or state contract operations syntactically; to our knowledge, none eliminates irrelevant variables when computing the contract of an interconnection or of a missing component for temporal contracts.

\paragraph{GR(1) Specifications: Decomposition, Hiding, and Synthesis.}
Filippidis and Murray~\cite{Filippidis2017hiding,filippidis2018layering}
use BDD-based quantification to hide variables while decomposing
system-level GR(1) specifications into realizable component contracts; in
contrast, our component contracts are given, and elimination is induced by
composition and quotient. Graebener et al.~\cite{Graebener2022} apply strong-merge 
to GR(1) unit-test contracts, which, unlike composition and quotient,
does not require CAVE. Zhang et al.~\cite{zhang2025vehicle}
synthesize an automated-driving safety filter, from which our AEB example is
adapted, and Katis et al.~\cite{katis2018validity} synthesize implementations
from assume--guarantee contracts; we treat GR(1) specifications as
assume--guarantee contracts~\cite{nuzzo2014contract} and compute contract operations while remaining
within GR(1).

\paragraph{Temporal-Logic Contracts for Compositional Control.}
Kim et al.~\cite{Kim2017smallgain} use parametric LTL assume--guarantee
contracts for compositional reasoning; Saoud et al.~\cite{saoud2021} and Sharf et al.~\cite{sharf2021assume} develop assume--guarantee contracts for compositional reasoning about continuous-time and discrete-time dynamical systems, respectively. Related
STL approaches coordinate multi-agent controllers through local
contracts~\cite{liu2025controller} and synthesize compatible local contracts
for decentralized control~\cite{Ghasemi2022}. These methods target verification or controller synthesis for dynamical systems, not the computation of contract operations such as composition and quotient on temporal specifications.

\section{GR(1) Context-Aided Variable Elimination}
\label{sec:gr1-cave}

It has been shown that computing the algebraic operations of contracts requires
solving two QE problems~\cite{pacti}, introduced in
Section~\ref{sec:assume-guarantee-contracts}. We now fix their notation for GR(1).
Suppose that a formula $\phi(\bar p,\bar q)$, called the \emph{target}, contains
irrelevant variables $\bar p$ that must be eliminated, and that we may use a
predicate $\Gamma(\bar p,\bar q)$, called the \emph{context}, which collects the
surrounding contract conditions that constrain these variables. The relevant
variables $\bar q$ remain in the result. In the notation of
Section~\ref{sec:assume-guarantee-contracts}, $\bar p$ plays the role of $y$, and
$\bar q$ the role of $x$ and $z$. Each variable is typed; some may be Boolean
propositions, and others may range over a background theory such as linear real
arithmetic. A formula is \emph{$\bar p$-free} if it is semantically independent of
every variable in $\bar p$, so the result of an elimination is a formula over
$\bar q$.
The two problems ask for a $\bar p$-free formula $\Psi^{ref}(\bar q)$ with
$\Gamma\wedge\Psi^{ref}\models\phi$ (refinement), and for a $\bar p$-free formula
$\Psi^{rel}(\bar q)$ with $\Gamma\wedge\phi\models\Psi^{rel}$ (relaxation). Their
exact solutions are
$\Psi^{ref}_{exact}(\bar q):=
\forall\bar p.\,(\Gamma(\bar p,\bar q)\Rightarrow\phi(\bar p,\bar q))$
for refinement and
$\Psi^{rel}_{exact}(\bar q):=
\exists\bar p.\,(\Gamma(\bar p,\bar q)\wedge\phi(\bar p,\bar q))$
for relaxation. We call $\forall\bar p.\,F$ and $\exists\bar p.\,F$ the
\emph{universal} and \emph{existential elimination} of $\bar p$ from $F$, where
$\forall\bar p.\,F$ is the weakest $\bar p$-free formula that implies $F$, and
$\exists\bar p.\,F$ is the strongest $\bar p$-free formula implied by $F$.

In the GR(1) setting, the context and the target have the form
$\Gamma(\bar p,\bar q):=
\gamma_{\mathrm{init}}\wedge
\Always\gamma_{\mathrm{inv}}\wedge
\bigwedge_{j\in J}\Always\Eventually\gamma_{\mathrm{liv},j}$
and
$\phi(\bar p,\bar q):=
\phi_{\mathrm{init}}\wedge
\Always\phi_{\mathrm{inv}}\wedge
\bigwedge_{i\in I}\Always\Eventually\phi_i$.
Exact temporal QE need not preserve the GR(1) form.
For example,
$\forall p.\left(
\Always\Eventually p
\Rightarrow
\Always\Eventually(p\wedge s)
\right)$
$\equiv\Eventually\Always s$.
Indeed, if \(s\) eventually holds forever, every recurrent valuation of \(p\)
eventually overlaps \(s\) infinitely often. Conversely, if
\(\Eventually\Always s\) fails, then \(\neg s\) occurs infinitely often; choosing
\(p:=\neg s\) makes \(\Always\Eventually p\) true while
\(\Always\Eventually(p\wedge s)\) is false. Thus exact elimination may produce
a formula outside GR(1), motivating sound algorithms guaranteed to remain within GR(1).


\subsection{Bounded-Next CAVE}
\label{sec:finite-window-cave}

Recall from Section~\ref{sec:gr1-fragment} that a bounded-next core is a sub-formula
built from Boolean connectives and typed predicates using a finite number of
$\Next$'s and no other temporal operator.
For such a core $\chi$, its
\emph{next-depth} $d(\chi)$ is the maximum number of $\Next$ operators on the path
from the root of $\chi$ to any occurrence of a variable.
When $\chi$ contains no $\Next$, $d(\chi)=0$.
We \emph{flatten} a bounded-next core by substituting
$\Next^k p$ with a fresh variable $p_k$ of the same sort as $p$ for each $p\in\bar p$ and
$\Next^k q$ with a fresh variable $q_k$ for each $q\in\bar q$. The resulting formula contains a finite number of variables and is denoted by $\chi^\flat$.

We assume cores
\(\gamma_{\mathrm{init}},\gamma_{\mathrm{inv}},\gamma_{\mathrm{liv},j},
\phi_{\mathrm{init}},\phi_{\mathrm{inv}},\phi_i\)
are bounded-next. For a finite set $B$ of cores involved in a variable elimination problem, define
$h_B:=\max_{\chi\in B}d(\chi)$ and
$D_B:=h_B+d(\gamma_{\mathrm{inv}})$, where
$\gamma_{\mathrm{inv}}$ is the invariant core of the context $\Gamma$, and let
\(
\Gamma_{\mathrm{inv}}^{[0:h_B]}
:=
\bigwedge_{k=0}^{h_B}\Next^k\gamma_{\mathrm{inv}}
\)
be the conjunction of the shifted invariants $\bigcirc^k\gamma_{\mathrm{inv}}$ (the invariant core shifted by $k$ steps) at the positions read by the cores in $B$, which we call the \textit{shifted invariant context}. 
Since $\Gamma$ contains $\Box\gamma_{\mathrm{inv}}$, the core $\gamma_{\mathrm{inv}}$
holds at every time position, so each $\bigcirc^{k}\gamma_{\mathrm{inv}}$ holds at
every position as well (Lemma~\ref{lem:shifted-invariant}). These instances are therefore
consequences of the context, and assuming them while eliminating $\bar p$ is sound.
We call the indexed positions $0,\ldots,D_B$ the \emph{finite window}
of the elimination problem. It contains all time positions needed to
evaluate these cores together with the shifted invariant context.
Consequently, after flattening, these formulas can be represented
over
$\bar p_0,\ldots,\bar p_{D_B}$ and
$\bar q_0,\ldots,\bar q_{D_B}$.
Throughout, $\forall\bar p$ and $\exists\bar p$ in a finite-window
elimination denote quantification over all copies
$\bar p_0,\ldots,\bar p_{D_B}$ associated with that window. We use $F_B(\bar p,\bar q)$ to denote the propositional formula obtained by flattening $\Gamma_{\mathrm{inv}}^{[0:h_B]}$. We show that evaluating these temporal cores over the finite window is equivalent to evaluating their flattened finite-variable representations. 

\begin{lemma}[Flattening Preserves Semantics]
\label{lem:flattening-correctness}
Let $\chi(\bar p,\bar q)$ be a bounded-next core with $d(\chi)\le D$.
Then, for every trace $\sigma$ and time $t$,
$(\sigma,t)\models\chi$ if and only if
$v_{\sigma,t,D}\models\chi^\flat$, where $v_{\sigma,t,D}$ assigns each
indexed variable $x_k$ the typed value $\sigma_{t+k}(x)$, for every
$x\in\bar p\cup\bar q$ and $0\le k\le D$.
\end{lemma}

\begin{proof}
By structural induction on $\chi$. The atomic and Boolean cases are immediate.
For $\chi=\Next\eta$, the semantics of $\Next$ give
$(\sigma,t)\models\Next\eta$ if and only if
$(\sigma,t+1)\models\eta$. By the induction hypothesis, this holds if and
only if $v_{\sigma,t+1,D-1}\models\eta^\flat$. Relative to time $t$, each
$x_k$ in this valuation corresponds to $x_{k+1}$, exactly as prescribed by
flattening $\Next\eta$. Hence,
$v_{\sigma,t,D}\models(\Next\eta)^\flat$, completing the induction.
\qedbox
\end{proof}

Hence, flattening cores preserves semantics over the finite-window. The restriction to bounded-next formulas reduces each CAVE problem to QE over a finite SMT formula. Constructions with
unbounded next-depth, such as infinite conjunctions of progressively longer
time windows, do not admit this finite representation and generally require automata- or fixed-point-based reasoning.

Our procedures treat the cores within each conjunct of the GR(1) formula separately; namely,
the initial core, the invariant core, and each of the liveness cores. 
We call each core solved over its finite window of invariant instances a \emph{finite subproblem}. A \emph{certificate} for a subproblem is a $\bar p$-free core $\eta$ that solves it;
the procedures then lift $\eta$ to a conjunct of the GR(1) result by adding the
temporal operators shown in the first column of the tables below. Each subproblem may utilize only the cores of the context that hold at the time positions it addresses. The initial subproblem concerns time~$0$, where $\gamma_{\mathrm{init}}$ and the shifted invariant context $\Gamma_{\mathrm{inv}}^{[0:h_B]}$ hold; the invariant subproblem pertains to every time position, where only the shifted invariant context is guaranteed to hold; and a liveness subproblem concerns infinitely many positions, so it may additionally use one fixed liveness core from the context. We call a certificate \emph{admissible} when it solves its
subproblem in this way. Keeping the subproblems separate preserves the conjunctive
GR(1) structure, and by Lemma~\ref{lem:flattening-correctness}, each one can be flattened into an equivalent
finite-window problem. Results that depend on several subproblems at once, on preceding invariant instances, on combinations or time-varying choices of
recurring liveness cores, or on unbounded next-depth need joint or unbounded
temporal reasoning and are outside the class. The restriction therefore yields a uniform construction for arbitrary bounded-next GR(1) formulas; formula-specific constructions outside the class may yield stronger results.

\begin{definition}[Admissible GR(1) CAVE Construction]
\label{def:admissible-gr1-cave}
A certificate $\eta$ is \emph{admissible} for a finite subproblem if it is
bounded-next, $\bar p$-free, and satisfies the \emph{obligation} on the finite window listed below for 
its corresponding subproblem. The refinement obligations are
\[
\begin{array}{c|c}
\text{lifted conjunct} & \text{finite window obligation}\\ \hline
\eta_{\mathrm{init}}
&
\eta_{\mathrm{init}}\models
\gamma_{\mathrm{init}}\wedge
\Gamma_{\mathrm{inv}}^{[0:h_B]}
\Rightarrow\phi_{\mathrm{init}},
\quad B=\{\gamma_{\mathrm{init}},\phi_{\mathrm{init}}\}
\\
\Always\eta_{\mathrm{inv}}
&
\eta_{\mathrm{inv}}\models
\Gamma_{\mathrm{inv}}^{[0:h_B]}
\Rightarrow\phi_{\mathrm{inv}},
\quad B=\{\phi_{\mathrm{inv}}\}
\\
\Always\Eventually\eta_{i,0}
&
\eta_{i,0}\models
\Gamma_{\mathrm{inv}}^{[0:h_{i,0}]}
\Rightarrow\phi_i,
\quad h_{i,0}=d(\phi_i)
\\
\Always\eta_{i,j}
&
\eta_{i,j}\models
\Gamma_{\mathrm{inv}}^{[0:h_{i,j}]}
\wedge\gamma_{\mathrm{liv},j}
\Rightarrow\phi_i,
\quad
\begin{gathered}
j\in J\ \text{fixed},\\[-1mm]
h_{i,j}=\max\{d(\phi_i),d(\gamma_{\mathrm{liv},j})\}.
\end{gathered}
\end{array}
\]
The relaxation obligations are
\[
\begin{array}{c|c}
\text{lifted conjunct} & \text{finite window obligation}\\ \hline
\eta_{\mathrm{init}}
&
\gamma_{\mathrm{init}}\wedge
\Gamma_{\mathrm{inv}}^{[0:h_B]}\wedge\phi_{\mathrm{init}}
\models\eta_{\mathrm{init}},
\quad B=\{\gamma_{\mathrm{init}},\phi_{\mathrm{init}}\}
\\
\Always\eta_{\mathrm{inv}}
&
\Gamma_{\mathrm{inv}}^{[0:h_B]}\wedge\phi_{\mathrm{inv}}
\models\eta_{\mathrm{inv}},
\quad B=\{\phi_{\mathrm{inv}}\}
\\
\Always\Eventually\eta_j^\gamma
&
\Gamma_{\mathrm{inv}}^{[0:d(\gamma_{\mathrm{liv},j})]}
\wedge\gamma_{\mathrm{liv},j}
\models\eta_j^\gamma,
\quad j\in J
\\
\Always\Eventually\eta_i^\phi
&
\Gamma_{\mathrm{inv}}^{[0:d(\phi_i)]}\wedge\phi_i
\models\eta_i^\phi,
\quad i\in I .
\end{array}
\]
A complete admissible relaxation, in which each certificate is admissible for its own
finite subproblem independently of the others, has the decomposed form
$\eta_{\mathrm{init}}\wedge\Always\eta_{\mathrm{inv}}
\wedge\bigwedge_{j\in J}\Always\Eventually\eta_j^\gamma
\wedge\bigwedge_{i\in I}\Always\Eventually\eta_i^\phi \in \mathsf{GR(1)}$. Dependencies between subproblems, multiple or
time-varying refinement supports, preceding invariant time-windows, and unbounded
next-depth are not admissible.
\end{definition}

The finite-window restriction can exclude GR(1)-expressible consequences that
require earlier invariant information. For example, for
$\Gamma=p\wedge\Always(\Next p\leftrightarrow q)$ and
$\phi=\Always(p\leftrightarrow s)$, eliminating $p$ gives the exact refinement
$s\wedge\Always(\Next s\leftrightarrow q)$, whereas the admissible invariant
subproblem yields $\bot$ because its finite window does not include the
preceding invariant instance. Such dependencies require a broader elimination
construction.

\subsection{GR(1)-Compatible Refinement and Relaxation}
\label{sec:gr1-cave-bounds}

\subsubsection{Refinement.}

Algorithm~\ref{alg:gr1-refinement} constructs the refinement bound by solving the initial, invariant, and liveness subproblems of Definition~\ref{def:admissible-gr1-cave}. For every liveness target \(\phi_i\) it builds one candidate per
\emph{branch}: the \emph{synthetic branch} (index $0$) and one \emph{support
branch} for each context-liveness core $\gamma_{\mathrm{liv},j}$, $j\in J$. 

\begin{algorithm}[t]
\caption{GR(1) Refinement CAVE}
\label{alg:gr1-refinement}
\begin{algorithmic}[1]
\Input{context
$\Gamma=\gamma_{\mathrm{init}}\wedge\Always\gamma_{\mathrm{inv}}
\wedge\bigwedge_{j\in J}\Always\Eventually\gamma_{\mathrm{liv},j}$,
target
$\phi=\phi_{\mathrm{init}}\wedge\Always\phi_{\mathrm{inv}}
\wedge\bigwedge_{i\in I}\Always\Eventually\phi_i$,
eliminated variables $\bar p$}
\Output{$\bar p$-free GR(1) refinement $\Psi^{ref}$}

\State $B\gets\{\gamma_{\mathrm{init}},\phi_{\mathrm{init}}\}$;
$\psi^{ref}_{\mathrm{init}}\gets
\forall\bar p.\,
\bigl(\gamma_{\mathrm{init}}^{\flat}\wedge
(\Gamma_{\mathrm{inv}}^{[0:h_B]})^\flat
\Rightarrow\phi_{\mathrm{init}}^{\flat}\bigr)$

\State $B\gets\{\phi_{\mathrm{inv}}\}$;
$\psi^{ref}_{\mathrm{inv}}\gets
\forall\bar p.\,
\bigl((\Gamma_{\mathrm{inv}}^{[0:h_B]})^\flat
\Rightarrow\phi_{\mathrm{inv}}^{\flat}\bigr)$

\For{$i\in I$}
    \State $(j_i,\psi^{ref}_{i,j_i})\gets
    \Call{SelectSupport}{\phi_i,\gamma_{\mathrm{inv}},
    \{\gamma_{\mathrm{liv},j}\}_{j\in J},\bar p}$
\EndFor
\State $I_0\gets\{i\in I:j_i=0\}$;
$I_{\mathrm{sup}}\gets I\setminus I_0$
\State $\Psi^{ref}\gets
\psi^{ref}_{\mathrm{init}}\wedge\Always\psi^{ref}_{\mathrm{inv}}
\wedge\bigwedge_{i\in I_0}\Always\Eventually\psi^{ref}_{i,0}
\wedge\bigwedge_{i\in I_{\mathrm{sup}}}\Always\psi^{ref}_{i,j_i}$
\State \Return $\Psi^{ref}$

\Statex
\Procedure{SelectSupport}
{$\phi_i,\gamma_{\mathrm{inv}},
\{\gamma_{\mathrm{liv},j}\}_{j\in J},\bar p$}
    \State $h_{i,0}\gets d(\phi_i)$;
    $\psi^{ref}_{i,0}\gets
    \forall\bar p.\bigl((\Gamma_{\mathrm{inv}}^{[0:h_{i,0}]})^\flat
    \Rightarrow\phi_i^\flat\bigr)$
    \State $U_i\gets\{0\}$
    \For{$j\in J$}
        \State $h_{i,j}\gets\max\{d(\phi_i),d(\gamma_{\mathrm{liv},j})\}$
        \State $\psi^{ref}_{i,j}\gets
        \forall\bar p.\bigl((\Gamma_{\mathrm{inv}}^{[0:h_{i,j}]})^\flat
        \wedge\gamma_{\mathrm{liv},j}^\flat
        \Rightarrow\phi_i^\flat\bigr)$;
        $U_i\gets U_i\cup\{j\}$
    \EndFor
    \State $j_i\gets\operatorname{Select}(U_i)$
    \Comment{valid core first, then non-false core, then fixed order}
    \State \Return $(j_i,\psi^{ref}_{i,j_i})$
\EndProcedure
\end{algorithmic}
\end{algorithm}


The synthetic branch eliminates $\bar p$ using the invariant context alone,
assuming no context-liveness core (hence synthetic: $0\notin J$); its
certificate (Definition~\ref{def:admissible-gr1-cave}) is lifted as the
recurrence condition $\Always\Eventually\psi^{ref}_{i,0}$. A support branch
$j\in J$ additionally assumes $\gamma_{\mathrm{liv},j}$ and is lifted as
$\Always\psi^{ref}_{i,j}$, because two independent recurrence conditions need
not synchronize. If $\phi_i$ is $\bar p$-free, we need not retain
$\Always\Eventually\phi_i$ directly. Since $\phi_i\Rightarrow\psi^{ref}_{i,0}$,
we have
$\Always\Eventually\phi_i\Rightarrow\Always\Eventually\psi^{ref}_{i,0}$, so the
synthetic branch, when selected, is at least as weak as direct retention.
After constructing every candidate, $\operatorname{Select}$ prefers one whose
finite core is established valid by the SMT solver, then one not established
unsatisfiable, and otherwise the first in a fixed synthetic-first order
(support syntax depth, then SMT-LIB core). It keeps the chosen candidate's
temporal operator ($\Always\Eventually$ for the synthetic branch, $\Always$ for
a support branch). Every candidate is individually sound
(Lemmas~\ref{lem:synthetic-liveness-refinement}
and~\ref{lem:fixed-support-liveness-refinement}), so $\operatorname{Select}$
affects only the strength of the result; no optimality is claimed across
branches with different temporal operators ($\Always$ versus
$\Always\Eventually$). Appendices~\ref{app:soundness}
and~\ref{app:optimality} give the proofs.

\paragraph{Illustrative example: refinement.}
Let $p$ be irrelevant and $q,s$ relevant, with context
$\Gamma=\Always\gamma_{\mathrm{inv}}$,
$\gamma_{\mathrm{inv}}:=((q\land\Next\neg q)\Rightarrow\Next p)\land(\Next p\Rightarrow q)$,
which forces $\Next p$ at each falling edge of $q$ and allows it only after $q$,
and target $\phi=\Always(\Next p\lor s)\land\Always\Eventually\Next p$. Over the window
$\{0,1,2\}$, the invariant and liveness subproblems evaluate to
$(q_0\land\neg q_1)\lor s_0$ and $q_0\land\neg q_1$, so refinement CAVE returns
$\forall p.\,{(\Always\gamma_{\mathrm{inv}}\Rightarrow\phi)}
\equiv\Always{\bigl((q\land\Next\neg q)\lor s\bigr)}\land\Always\Eventually{(q\land\Next\neg q)}$,
which is exact, since the worst case $p_{t+1}=q_t\land\neg q_{t+1}$ is a legal trace. Without
$\Gamma$, the elimination would be $\bot$.


\subsubsection{Relaxation.}
Algorithm~\ref{alg:gr1-relaxation} constructs the relaxation bound by treating
the initial, invariant, and liveness subproblems separately. Initial and
invariant elimination use the corresponding finite windows directly, while
the liveness subroutine derives consequences over the remaining relevant
variables from the context and target liveness cores, removing redundant
ones. For liveness relaxation, every context and target liveness core can
independently yield a consequence over the remaining relevant variables. The subroutine collects these
consequences in $\mathcal P$, discards those equivalent to $\top$, retains one
representative of each logical-equivalence class, and removes any core
$\rho_b$ strictly implied by another core $\rho_a$. Since
$\rho_a\Rightarrow\rho_b$ implies
$\Always\Eventually\rho_a\Rightarrow\Always\Eventually\rho_b$, these
reductions preserve the conjunction of liveness requirements exactly.

\begin{algorithm}[t]
\caption{GR(1) Relaxation CAVE}
\label{alg:gr1-relaxation}
\begin{algorithmic}[1]
\Input{ context
$\Gamma=\gamma_{\mathrm{init}}\wedge\Always\gamma_{\mathrm{inv}}
\wedge\bigwedge_{j\in J}\Always\Eventually\gamma_{\mathrm{liv},j}$,
target
$\phi=\phi_{\mathrm{init}}\wedge\Always\phi_{\mathrm{inv}}
\wedge\bigwedge_{i\in I}\Always\Eventually\phi_i$,
eliminated variables $\bar p$}
\Output{$\bar p$-free GR(1) relaxation $\Psi^{rel}$}

\State $B\gets\{\gamma_{\mathrm{init}},\phi_{\mathrm{init}}\}$;
$\psi^{rel}_{\mathrm{init}}\gets
\exists\bar p.\,
\bigl[
\gamma_{\mathrm{init}}^{\flat}\wedge
(\Gamma_{\mathrm{inv}}^{[0:h_B]})^\flat\wedge
\phi_{\mathrm{init}}^{\flat}
\bigr]$

\State $B\gets\{\phi_{\mathrm{inv}}\}$;
$\psi^{rel}_{\mathrm{inv}}\gets
\exists\bar p.\,
\bigl[
(\Gamma_{\mathrm{inv}}^{[0:h_B]})^\flat
\wedge\phi_{\mathrm{inv}}^{\flat}
\bigr]$

\State $\mathcal P^\star\gets
\Call{SelectLiveness}
{\{\phi_i\}_{i\in I},\gamma_{\mathrm{inv}},
 \{\gamma_{\mathrm{liv},j}\}_{j\in J},\bar p}$

\State $\Psi^{rel}\gets
\psi^{rel}_{\mathrm{init}}
\wedge\Always\psi^{rel}_{\mathrm{inv}}
\wedge\bigwedge_{\rho\in\mathcal P^\star}
\Always\Eventually\rho$

\State \Return $\Psi^{rel}$

\Statex
\Procedure{SelectLiveness}
{$\{\phi_i\}_{i\in I},\gamma_{\mathrm{inv}},
  \{\gamma_{\mathrm{liv},j}\}_{j\in J},\bar p$}

    \State $\mathcal P\gets\emptyset$

    \For{$j\in J$}
        \State $h_j\gets d(\gamma_{\mathrm{liv},j})$;
        $D_j\gets h_j+d(\gamma_{\mathrm{inv}})$
        \State $\psi^\gamma_j\gets
        \exists\bar p.\,
        \bigl[
        (\Gamma_{\mathrm{inv}}^{[0:h_j]})^\flat
        \wedge\gamma_{\mathrm{liv},j}^\flat
        \bigr]$
        \If{$\psi^\gamma_j\not\equiv\top$}
            \State $\mathcal P\gets
            \mathcal P\cup\{\psi^\gamma_j\}$
        \EndIf
    \EndFor

    \For{$i\in I$}
        \State $h_i\gets d(\phi_i)$;
        $D_i\gets h_i+d(\gamma_{\mathrm{inv}})$
        \State $\psi^\phi_i\gets
        \exists\bar p.\,
        \bigl[
        (\Gamma_{\mathrm{inv}}^{[0:h_i]})^\flat
        \wedge\phi_i^\flat
        \bigr]$
        \If{$\psi^\phi_i\not\equiv\top$}
            \State $\mathcal P\gets
            \mathcal P\cup\{\psi^\phi_i\}$
        \EndIf
    \EndFor

    \State $\mathcal P\gets\mathcal P/{\equiv}$;
    $\mathcal P^\star\gets\mathcal P$ \Comment{remove logically equivalent candidates from $\mathcal P$}

    \ForAll{$\rho_b\in\mathcal P$}
        \ForAll{$\rho_a\in\mathcal P\setminus\{\rho_b\}$}
            \If{$\mathrm{SMT\_UNSAT}
            (\rho_a\wedge\neg\rho_b)$
            \textbf{and}
            $\mathrm{SMT\_SAT}
            (\rho_b\wedge\neg\rho_a)$}
                \State $\mathcal P^\star\gets
                \mathcal P^\star\setminus\{\rho_b\}$
                \State \textbf{break}
            \EndIf
        \EndFor
    \EndFor

    \State \Return $\mathcal P^\star$
\EndProcedure
\end{algorithmic}
\end{algorithm}

The invariant condition $\Always\psi^{rel}_{\mathrm{inv}}$ is sound but need
not be exact, since the admissible construction performs elimination on a finite invariant
window rather than unrestricted infinite-trace elimination.
We therefore retain this finite GR(1) construction. Each existential elimination
is the strongest $\bar p$-free consequence of its fixed finite subproblem, and
\textsc{SelectLiveness} preserves the conjunction of the liveness
requirements obtained by elimination exactly. The resulting complete relaxation is strongest within
the decomposed finite-subproblem class of Definition~\ref{def:admissible-gr1-cave};
the proofs are in Appendices~\ref{app:soundness} and~\ref{app:optimality}.

\paragraph{Illustrative example: relaxation.}
For the same $\Gamma$ and $\phi$, the window subproblems evaluate to $q_0\lor s_0$
and $q_0$, so relaxation CAVE returns
$\exists p.\,(\Always\gamma_{\mathrm{inv}}\land\phi)
\equiv\Always{(q\lor s)}\land\Always\Eventually q$,
again exact, since the best case $p_{t+1}=q_t$ is a legal trace meeting $\phi$ whenever
the result holds.

\looseness=-1 \subsubsection{Complexity.}
Let $n_p:=|\bar p|$, $n:=|\bar p|+|\bar q|$, and $D_B=h_B+d(\gamma_{\mathrm{inv}})$ be the CAVE-window depth. Flattening produces $O(n(D_B+1))$ indexed variables, with formula size linear in the $h_B+1$ invariant copies before simplification. The dominant operation is finite propositional/SMT QE. For each refinement liveness target, Algorithm~\ref{alg:gr1-refinement} performs one QE call for the synthetic branch and $|J|$ for the support branches. Relaxation performs one QE call for each context and target liveness core, followed by at most quadratic finite SMT implication checks in \textsc{SelectLiveness}. Thus CAVE avoids general temporal QE by exploiting bounded-next GR(1) structure, while its symbolic cost depends strongly on the modulo theory and representation. 
Appendix~\ref{app:scalability} reports measurements on synthetic families, including a comparison baseline that unrolls the specifications over a finite horizon and ignores liveness (Appendix~\ref{app:eval-unrolling}), chains of compositions (Appendix~\ref{app:eval-chains}), and ablations (Appendix~\ref{app:eval-ablations}). The cost is governed by $n_p(D_B+1)$ (Appendix~\ref{app:eval-scaling}). On the Boolean substitution family (Appendix~\ref{app:eval-families}), refinement grows from $0.06\,\mathrm{s}$ to $27\,\mathrm{s}$ when $n_p$ goes from $2$ to $6$ at $D_B=1$, and $n_p=8$ exceeds $600\,\mathrm{s}$, whereas on linear real arithmetic the runtime stays below $0.8\,\mathrm{s}$ for $n_p\le 8$. The runtime grows linearly with the number of liveness cores (for refinement from $0.03\,\mathrm{s}$ to $0.22\,\mathrm{s}$ when $n_I$ goes from $1$ to $8$). These are observations over the tested ranges, not asymptotic bounds.

\subsubsection{Empirical comparisons.}
We compared the refinement construction with exact elimination through B\"uchi automata computed with Spot~\cite{duret2022spot}, by translating $\Gamma\wedge\neg\varphi$, existentially eliminating the irrelevant variables, and complementing, on a Boolean family designed to stress complementation. 
Spot~\cite{duret2022spot} is faster on small instances, but its cost grows steeply with the number of liveness conjuncts, as expected from complementation. Spot needs $0.003\,\mathrm{s}$ with one target liveness conjunct, $1.1\,\mathrm{s}$ with four, $62\,\mathrm{s}$ with five, and exceeds $300\,\mathrm{s}$ with six, whereas CAVE needs between $0.05$ and $0.36\,\mathrm{s}$ for one to eight. On the number of eliminated variables, the order reverses ($7.8\,\mathrm{s}$ for Spot against $137\,\mathrm{s}$ for CAVE, on seven variables). The automaton is exact, whereas the CAVE result was sound but not necessarily equivalent (Appendix~\ref{app:eval-automata}). Against an explicit-state oracle on bounded-lasso words ($8000$ random instances per direction), no unsound result was found in $16000$ classifications, and the refinement and relaxation results were exact on $35.3\%$ and $78.6\%$ respectively on nontrivial problems (where result and exact answer are not the same constant) (Appendix~\ref{app:eval-precision}).

\subsection{Soundness and Optimality}
\label{sec:cave-optimality}

We now show that the constructions above are sound solutions to the CAVE by refinement and relaxation problems.

\begin{theorem}[Soundness of GR(1) CAVE]
\label{thm:cave-soundness}
Let $\Psi^{\mathrm{ref}}(\bar q)$ and $\Psi^{\mathrm{rel}}(\bar q)$ be the $\bar p$-free GR(1) formulas returned by Algorithms~\ref{alg:gr1-refinement} and~\ref{alg:gr1-relaxation} for context $\Gamma$, target $\phi$, and eliminated variables $\bar p$. Then
\[
  \Psi^{ref}(\bar q)
  \Rightarrow
  \forall\bar p.
  \bigl(
    \Gamma(\bar p,\bar q)
    \Rightarrow
    \phi(\bar p,\bar q)
  \bigr) \quad \text{and} \quad  \exists\bar p.
  \bigl(
    \Gamma(\bar p,\bar q)
    \wedge
    \phi(\bar p,\bar q)
  \bigr)
  \Rightarrow
  \Psi^{rel}(\bar q).
\]
\end{theorem}

\begin{proofsketch}
Exact finite-window flattening and universal elimination establish the refinement
INIT/INV obligations. For LIV, recurrence of the synthetic core after elimination is
sufficient under the invariant context, while a genuine recurring support uses
a persistent implication after elimination so that both hold at the same support
positions. Existential elimination gives the relaxation consequences, and
\textsc{SelectLiveness} removes only terms redundant in their recurrent
conjunction. Appendix~\ref{app:soundness} gives the trace-level proofs.
\end{proofsketch}


Next, we compare CAVE with all admissible solutions of the finite subproblems.

\begin{definition}[Finite-Subproblem Optimality]
\label{def:gr1-optimality}
Within a fixed admissible refinement subproblem or certificate branch as in Definition~\ref{def:admissible-gr1-cave},
$\eta^\star$ is \emph{weakest} iff $\eta\Rightarrow\eta^\star$ for every admissible certificate $\eta$. Similarly, within a fixed admissible relaxation subproblem as in Definition~\ref{def:admissible-gr1-cave}, $\rho^\star$ is
\emph{strongest} iff $\rho^\star\Rightarrow\rho$ for every admissible consequence $\rho$. These orders are subproblem or branch specific and do not imply
complete-refinement optimality.
\end{definition}

\begin{theorem}[Finite-Subproblem Optimality of CAVE]
\label{thm:finite-window-gr1-optimality}
Under Definition~\ref{def:gr1-optimality}, for refinement,
$\psi^{ref}_{\mathrm{init}}$ and $\Always\psi^{ref}_{\mathrm{inv}}$ are weakest
in their corresponding finite subproblems; for each liveness target $i$,
$\Always\Eventually\psi^{ref}_{i,0}$ and, for every fixed $j\in J$,
$\Always\psi^{ref}_{i,j}$ are weakest in their respective admissible branches.

For relaxation, $\psi^{rel}_{\mathrm{init}}$,
$\Always\psi^{rel}_{\mathrm{inv}}$,
$\Always\Eventually\psi^\gamma_j$ for every $j\in J$, and
$\Always\Eventually\psi^\phi_i$ for every $i\in I$ are strongest in their
respective finite subproblems. Moreover, the complete relaxation is strongest within the decomposed finite-subproblem class of
Definition~\ref{def:admissible-gr1-cave}.
\end{theorem}

\begin{proofsketch}
Flattening converts admissible temporal subproblem conditions to finite
indexed implications. Universal elimination then gives the weakest $\bar p$-free
sufficient refinement conditions; existential elimination dually gives the
strongest $\bar p$-free relaxation consequences.
\textsc{SelectLiveness} preserves the recurrent conjunction, so strongest
relaxation consequences compose across the subproblems. Refinement branch
choices need not compose optimally. See Appendix~\ref{app:optimality}.
\end{proofsketch}

The theorem concerns the admissible finite GR(1) constructions of
Definition~\ref{def:admissible-gr1-cave}. Admissibility is a logical condition
on a certificate, namely sufficiency for its window obligation, and is stated
without reference to Algorithm~\ref{alg:gr1-refinement}; the theorem therefore
says that CAVE loses nothing relative to any certificate satisfying it, and not
merely relative to its own output.
The admissible class is not an arbitrary restriction. It is what the GR(1)
decomposition prescribes, independently of any algorithm, namely solving each
initial, invariant, and liveness obligation separately at the position where it
is checked, under the invariant context from that position onward. Within each
such subproblem the finite-window representation is exact
(Lemma~\ref{lem:flattening-correctness}), so CAVE returns the weakest refinement
and strongest relaxation attainable for that subproblem.
Any conservatism therefore comes from what the class deliberately leaves out, that is,
dependencies between distinct obligations, as in
the example following Definition~\ref{def:admissible-gr1-cave}, and
consistency conditions of unbounded depth. For example, over real-valued $u$,
every finite shifted-window conjunction formed
from $u\geq0\wedge\Next u\leq u-1$ is satisfiable, whereas
$\Always(u\geq0\wedge\Next u\leq u-1)$ is unsatisfiable, so an unrestricted exact
invariant relaxation may require increasingly deep consistency conditions.


\section{GR(1) Contract Operations in Pacti}
\label{sec:gr1-contract-operations}

We instantiate the contract algebra of Section~\ref{sec:assume-guarantee-contracts} for GR(1)-modulo theory using the existing Pacti contract operation algorithms ~\cite{pacti}, with an open-source implementation available at \url{https://github.com/pacti-org/pacti}.
The contract algebra itself is unchanged; the refinement and
relaxation CAVE procedures of Algorithms~\ref{alg:gr1-refinement}
and~\ref{alg:gr1-relaxation} provide the GR(1)-specific variable-elimination
primitives that composition, quotient, and related operations require. As in~\cite{pacti}, the resulting GR(1) composition is a
sound relaxation of the exact algebraic composition, while the quotient is a
refinement of the exact quotient; merging is exact after lifting to a compatible
union interface. The algorithmic realizations of these operations and their soundness results
are given in Appendix~\ref{app:gr1-contract-operations}.

Internally, contracts maintain separate \textsc{Init}, \textsc{Inv}, and
\textsc{Liv} assumptions and guarantees. Operations construct
the corresponding assume--guarantee expressions and invoke CAVE, using
finite-window flattening and symbolic QE, whenever irrelevant
variables must be eliminated; the generated GR(1) contracts are retained for
all subsequent operations. 

The implementation also provides a sound refinement verifier for bounded-next
GR(1) contracts (Appendix~\ref{app:gr1-refinement-check}). It checks
the structured \textsc{Init}, \textsc{Inv}, and \textsc{Liv} obligations using
direct implications and finite certificates formed from shifted source
invariants, independently of the finite CAVE windows used to construct operation
results; an exact procedure would instead require a separate infinite-trace temporal
reasoning backend, beyond the finite SMT-based CAVE framework considered here.
By Proposition~\ref{prop:app-refinement-soundness} in Appendix~\ref{app:gr1-refinement-check}, every successful check
establishes semantic contract refinement, whereas failure means only that the
implemented procedure did not establish refinement.

\subsection{Running Example: Composition and Quotient}
\label{sec:aeb-operations}

We apply the GR(1) contract operations to the AEB architecture of
Section~\ref{sec:aeb-running-example}. Composition returns a sound
relaxation of the exact composition, and quotient a sound refinement of the
exact quotient. We compute
$\mathcal C_{\mathrm{comp}}:=\mathcal C_S\parallel\mathcal C_V$ and then
$\mathcal C_{SVD}:=\mathcal C_{\mathrm{comp}}\parallel\mathcal C_D$, which
together eliminate $\{u,\MOV,\DCL,\STP,v\}$ and leave inputs
$\{d,\URG,\WRN\}$ and output $\{z\}$. Separately, we compute
$\mathcal Q:=\mathcal C_T/\mathcal C_V$, with inputs $\{d,\URG,\WRN\}$ and
outputs $\{u,\MOV,\DCL,\STP\}$, resulting in what a supervisor component must
satisfy to meet $\mathcal C_T$ when composed with the vehicle.

Let $\mathsf{clear}:=\neg(\URG\lor\WRN)$ and abbreviate the environmental
invariant core by
$E_b := b\le d\le30\land(d\le10\Rightarrow\URG\land\neg\WRN)
\land(10<d\le15\Rightarrow\WRN\land\neg\URG)
\land(d>15\Rightarrow\mathsf{clear})
\land(\mathsf{clear}\Rightarrow d\ge20)$. Table~\ref{tab:aeb-operation-results} shows the contracts generated
by Pacti, in compact forms verified equivalent to the generated contracts by mutual refinement. For readability, both initial assumptions state
$20\le d\le30$ and the composition's initial guarantee shows the full
predicted-clearance interval.

\begin{table}[!t]
\centering
\caption{Equivalent compact forms of the computed AEB contracts.}
\label{tab:aeb-operation-results}
\scriptsize
\setlength{\tabcolsep}{5pt}
\renewcommand{\arraystretch}{1.2}

\resizebox{\columnwidth}{!}{%
\begin{tikzpicture}
\node[inner sep=0pt] (aebopstable) {%
\begin{tabular}{
    @{}
    >{\raggedright\arraybackslash}m{0.49\textwidth}
    >{\raggedright\arraybackslash}m{0.49\textwidth}
    @{}
}
\toprule

\multicolumn{1}{c}{
    \textbf{Composition $\mathcal C_{SVD}$}
}
&
\multicolumn{1}{c}{
    \textbf{Quotient $\mathcal Q=\mathcal C_T/\mathcal C_V$}
}
\tabularnewline

\midrule

\begin{tabular}{@{}l@{}}
\textbf{Inputs:}\quad $d,\URG,\WRN$\\
\textbf{Output:}\quad $z$
\end{tabular}
&
\begin{tabular}{@{}l@{}}
\textbf{Inputs:}\quad $d,\URG,\WRN$\\
\textbf{Outputs:}\quad $u,\MOV,\DCL,\STP$
\end{tabular}
\tabularnewline

\midrule

\multicolumn{1}{c}{
    \begin{tabular}{@{}c@{}}
    \textbf{Assumptions}\\[-0.3mm]
    $A_{SVD}
    =A_{SVD,\mathrm{init}}
    \land A_{SVD,\mathrm{inv}}
    \land A_{SVD,\mathrm{liv}}$
    \end{tabular}
}
&
\multicolumn{1}{c}{
    \begin{tabular}{@{}c@{}}
    \textbf{Assumptions}\\[-0.3mm]
    $A_Q
    =A_{Q,\mathrm{init}}
    \land A_{Q,\mathrm{inv}}
    \land A_{Q,\mathrm{liv}}$
    \end{tabular}
}
\tabularnewline

\begin{tabular}{@{}l@{}}
$A_{SVD,\mathrm{init}}$:\quad
$20\le d\le30\land\mathsf{clear}$
\\[1.5mm]
$A_{SVD,\mathrm{inv}}$:\quad $\Always(E_0)$
\\[1.5mm]
$A_{SVD,\mathrm{liv}}$:\quad
$\Always\Eventually(\mathsf{clear})$
\end{tabular}
&
\begin{tabular}{@{}l@{}}
$A_{Q,\mathrm{init}}$:\quad
$20\le d\le30\land\mathsf{clear}$
\\[1.5mm]
$A_{Q,\mathrm{inv}}$:\quad $\Always(E_{3/2})$
\\[1.5mm]
$A_{Q,\mathrm{liv}}$:\quad
$\Always\Eventually(\mathsf{clear})$
\end{tabular}
\tabularnewline

\midrule

\multicolumn{1}{c}{
    \begin{tabular}{@{}c@{}}
    \textbf{Guarantees}\\[-0.3mm]
    $G_{SVD}
    =G_{SVD,\mathrm{init}}
    \land G_{SVD,\mathrm{inv}}
    \land G_{SVD,\mathrm{liv}}$
    \end{tabular}
}
&
\multicolumn{1}{c}{
    \begin{tabular}{@{}c@{}}
    \textbf{Guarantees}\\[-0.3mm]
    $G_Q
    =G_{Q,\mathrm{init}}
    \land G_{Q,\mathrm{inv}}
    \land G_{Q,\mathrm{liv}}$
    \end{tabular}
}
\tabularnewline

\begin{tabular}{@{}l@{}}
$G_{SVD,\mathrm{init}}$:\quad
$9d/10-1/2\le\Next z\le9d/10+1/2$
\\[1.5mm]
$G_{SVD,\mathrm{inv}}$:\quad
$\displaystyle
\Always\!\left(
\begin{aligned}
&\,\,\,4d/5-1/2\le\Next z\le d\\
&\land\bigl(\URG\land\neg\WRN\\
&\hspace{6mm}\Rightarrow\Next z\ge d-1/2\bigr)\\
&\land\bigl(\mathsf{clear}\\
&\hspace{6mm}\Rightarrow\Next z\le9d/10+1/2\bigr)
\end{aligned}
\right)$
\\[1.5mm]
$G_{SVD,\mathrm{liv}}$:\quad $\mathsf{true}$
\end{tabular}
&
\begin{tabular}{@{}l@{}}
$G_{Q,\mathrm{init}}$:\quad $2\le u\le3$
\\[1.5mm]
$G_{Q,\mathrm{inv}}$:\quad
$\displaystyle
\Always\!\left(
\begin{aligned}
&\,\,\,0\le u\le6\\
&\land u\le d-3/2\\
&\land(\Next\STP\Rightarrow u=0)\\
&\land(\Next\DCL\Rightarrow u\le3)\\
&\land(\Next\MOV\Rightarrow u\ge2)
\end{aligned}
\right)$
\\[9mm]
$G_{Q,\mathrm{liv}}$:\quad
$\displaystyle
\begin{aligned}
\,\,\,\Always\Eventually(\mathsf{clear}\land u\ge2) \land\Always\Eventually\MOV
\end{aligned}$
\end{tabular}
\tabularnewline

\bottomrule
\end{tabular}%
};

\draw[line width=0.4pt]
    ([yshift=-5.3mm]aebopstable.north)
    --
    (aebopstable.south);

\end{tikzpicture}}
\end{table}

Composition exposes a compatibility requirement that neither component states.
The condition $\mathsf{clear}\Rightarrow d\ge20$ is introduced by
$\mathcal C_S\parallel\mathcal C_V$ and reused in $A_T$. It is the weakest
condition under which every clear-mode supervisor command meets the vehicle
requirement: for $15<d<20$, the supervisor may issue commands
$u\in[d/10,d/5]$ below $2$, violating $\Next\MOV\Rightarrow u\ge2$.
Composing with $\mathcal C_D$ preserves these assumptions and exposes the
one-step predicted residual clearance through
$\Next z=d-\Delta t\,\Next v$.
For $\Delta t=1$, the computed relaxed composition yields region-wise
bounds on this clearance in every environment mode, namely
$\URG\land\neg\WRN\Rightarrow d-\tfrac12\le\Next z\le d$,
$\WRN\land\neg\URG\Rightarrow\tfrac{4d}{5}-\tfrac12\le\Next z\le d$, and
$\mathsf{clear}\Rightarrow\tfrac{4d}{5}-\tfrac12\le\Next z\le\tfrac{9d}{10}+\tfrac12$.
These bounds describe the computed relaxation, not exact physical dynamics (Figure~\ref{fig:aeb-operations-regions}(a)).

The quotient generates concrete requirements for the missing
supervisor (Figure~\ref{fig:aeb-operations-regions}(b)). It retains the vehicle's command-admissibility
requirements and adds one-step safety and recurrent progress in clear
conditions: $u\le\frac{d-d_s}{\Delta t}-\epsilon$ and $\Always\Eventually
  (\mathsf{clear}\land u\ge v_{\mathrm{prog}}+\epsilon)$.
For the chosen parameters, these become $u\le d-3/2$ and
$\Always\Eventually(\mathsf{clear}\land u\ge2)$. The command bound leaves
room for the vehicle's tracking error while meeting the target's
braking-distance bound, and the progress condition ensures forward motion
whenever the environment clears. Because the quotient retains $u\ge0$,
the eliminations simplify exactly to these forms. The requirement $\Always\Eventually\MOV$ stays separate,
since the two recurrence conditions need not hold at the same positions.
The restriction $d\ge3/2$ is the safety-feasibility condition already
included in $A_T$, not generated by quotient. Pacti verifies both
$\mathcal C_S\le\mathcal Q$, so the existing supervisor already meets the
computed requirements over the operating domain, and $\mathcal C_{\mathrm{rec}}\le\mathcal C_T$, where
$\mathcal C_{\mathrm{rec}}:=\mathcal C_V\parallel\mathcal Q$. Hence, any supervisor refining $\mathcal Q$ meets the target when composed with the vehicle.

\begin{figure}[t]
\centering
\begin{minipage}[t]{0.49\linewidth}
  \centering
  \includegraphics[width=\linewidth]{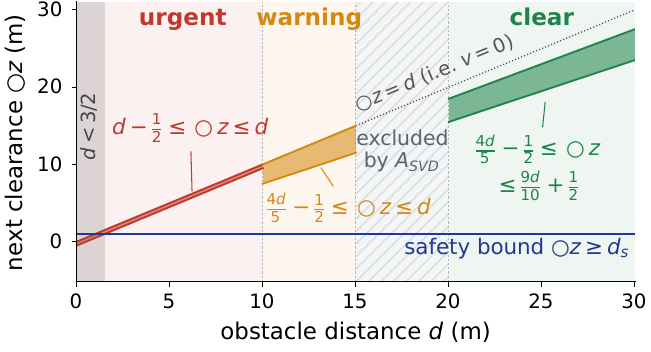}\\[-0.5ex]
  {\small (a)}
\end{minipage}\hfill
\begin{minipage}[t]{0.49\linewidth}
  \centering
  \includegraphics[width=\linewidth]{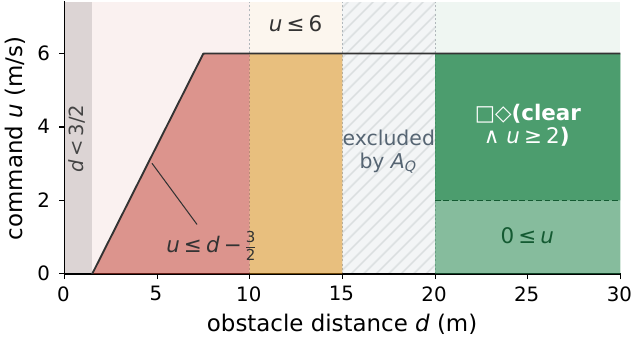}\\[-0.5ex]
  {\small (b)}
\end{minipage}
\caption{Region-wise view of the computed AEB contracts over the obstacle
distance $d$ (hatched: excluded by the computed assumption; gray:
$d<\tfrac32$, outside $A_T$).
(a)~Relaxed composition $\mathcal C_{SVD}$: bounds on the clearance
$\bigcirc z$ per mode, above the safety bound $\bigcirc z\ge d_s$ for
$d\ge\tfrac32$.
(b)~Quotient $\mathcal Q=\mathcal C_T/\mathcal C_V$: admissible commands
$0\le u\le\min\{6,\,d-\tfrac32\}$ and recurrent progress $u\ge2$ when clear.}
\label{fig:aeb-operations-regions}
\end{figure}

Finally, the refinement checker establishes
$\mathcal C_{\mathrm{comp}}\le\mathcal C_T$: under the target assumptions,
the supervisor and vehicle together satisfy the target's safety bound and
recurrent progress. $\mathcal C_{\mathrm{comp}}$ relaxes the exact
composition, so the latter also refines $\mathcal C_T$.
The recurring support
$\Always\Eventually(\mathsf{clear})$ comes from the target
assumption $A_{T,\mathrm{liv}}$; combined with the actual invariant context,
whose clear-region consequences include $d\ge20$ and
$\Next v\ge d/10-1/2$, it yields
$\Always\Eventually(\mathsf{clear}\land\Next v\ge3/2)$.
The certificate succeeds at horizon $h=0$. Support alone fails, invariants
alone do not establish recurring clear positions, but their combination
succeeds by exercising the invariant-assisted recurrence branch of
Algorithm~\ref{alg:app-gr1-refinement-check} in Appendix~\ref{app:gr1-refinement-check}.
All results are computed automatically in Pacti. Computing $\mathcal C_{SVD}$ and $\mathcal Q$ takes a median of $0.93\,\mathrm{s}$ over $10$ runs, of which $0.48\,\mathrm{s}$ is spent in CAVE. Detailed run-time results are reported in Appendix~\ref{app:scalability-aeb}.

\section{Conclusion and Future Work}
\label{sec:conclusion}

We developed a bounded-next GR(1) assume--guarantee contract framework with
sound refinement and relaxation CAVE procedures and weakest/strongest
finite-subproblem guarantees within the admissible class. Their Pacti
implementation supports composition, quotient, merging, and refinement
analysis through finite symbolic computation; the AEB study illustrates the
resulting system-level requirements and compatibility conditions. Experiments indicate that the cost grows steeply with the number of eliminated variables but only linearly with the number of liveness conjuncts. Future work
includes more complete temporal refinement verification, less conservative
CAVE branch selection, extensions beyond bounded-next GR(1), and evaluation on
larger compositional systems.

\paragraph{Acknowledgements.} We acknowledge the support of the Automotive Research Center (ARC), Cooperative Agreement W56HZV-24-2-0001 U.S. Army DEVCOM GVSC.

\paragraph{AI usage statement.} Generative AI tools were used to assist with rewriting and condensing text, improving clarity and grammar, reviewing mathematical arguments and proofs for potential errors or gaps, and supporting the development and debugging of the Pacti implementation. All original scientific contributions, including the theoretical framework, mathematical formulations, algorithms, and proofs, were developed by the authors. The authors retain full responsibility for the correctness and integrity of the mathematical results, software implementation, and manuscript. 



\paragraph{Data availability statement.} The implementation of the GR(1) contract operations is available as part of the open-source Pacti framework at \url{https://github.com/pacti-org/pacti}. The accompanying examples, tests, benchmark scripts, and experimental measurements are provided in the associated research artifact, where available. A versioned archival release will be provided to support long-term reproducibility.



\bibliographystyle{splncs04}
\bibliography{references}

\newpage
\appendix
\section*{Appendices}

\addcontentsline{toc}{section}{Appendices}
\input{appendices}

\end{document}

%% file: aeb_architecture.tex
%
\ifdefined\pgfversion\else
  \errmessage{aeb_architecture.tex is a figure fragment: load TikZ in the
    preamble and \string\input\space it inside a figure}%
\expandafter \fi
\expandafter\ifx\csname tikz@library@arrows.meta@loaded\endcsname\relax\usetikzlibrary{arrows.meta}\fi%
\expandafter\ifx\csname tikz@library@decorations.pathreplacing@loaded\endcsname\relax\usetikzlibrary{decorations.pathreplacing}\fi%
\begingroup
\begin{tikzpicture}[x=1cm, y=1cm, line join=round]
\definecolor{aeInk}{HTML}{252A31}\definecolor{aeMuted}{HTML}{5F6670}%
\definecolor{aeArrow}{HTML}{59606B}%
\definecolor{aeGrayFill}{HTML}{F5F6F8}\definecolor{aeGrayLine}{HTML}{BEC3CA}%
\definecolor{aeBlueFill}{HTML}{EBF2FB}\definecolor{aeBlueLine}{HTML}{89AAD8}%
\definecolor{aeBlueText}{HTML}{2D5891}%
\definecolor{aeMintFill}{HTML}{E8F5EC}\definecolor{aeMintLine}{HTML}{80BC90}%
\definecolor{aeMintText}{HTML}{2B6E41}%
\definecolor{aeLavFill}{HTML}{F1ECF9}\definecolor{aeLavLine}{HTML}{A594CD}%
\definecolor{aeLavText}{HTML}{55428A}%
\definecolor{aePeachFill}{HTML}{FDF3EA}\definecolor{aePeachLine}{HTML}{E6AE85}%
\definecolor{aePeachText}{HTML}{9A4F18}\definecolor{aePeachArrow}{HTML}{D9935F}%
\definecolor{aeInitFill}{HTML}{FFF0CF}\definecolor{aeInitLine}{HTML}{D3A43E}%
\definecolor{aeLivFill}{HTML}{FAE5EB}\definecolor{aeLivLine}{HTML}{D2879E}%
\definecolor{aeOK}{HTML}{2E8550}\definecolor{aeBad}{HTML}{BF4B48}%
\colorlet{aeUrgFill}{aeLivFill}   \colorlet{aeUrgLine}{aeBad}
\colorlet{aeWrnFill}{aeInitFill}  \colorlet{aeWrnLine}{aeInitLine}
\colorlet{aeClrFill}{aeMintFill}  \colorlet{aeClrLine}{aeOK}
\ifdefined\aefont\else\def\aefont{\sffamily}\fi
\def\aeHead{\fontsize{12}{14}\selectfont\bfseries}%
\def\aeTitle{\fontsize{9.4}{11.4}\selectfont\bfseries}%
\def\aeBody{\fontsize{8.2}{10.3}\selectfont}%
\def\aeNote{\fontsize{7}{8.6}\selectfont}%
\def\aeVar{\fontsize{10}{12}\selectfont}
\def\aeTiny{\fontsize{5.6}{6.4}\selectfont}%
\def\aesf#1{\mathsf{#1}}%
\def\aeN{\mathord{\bigcirc}}%
\tikzset{
  aebox/.style    ={rounded corners=5pt, line width=0.65pt},
  aeregion/.style ={aebox, rounded corners=8pt, draw=aeGrayLine, fill=aeGrayFill},
  aecard/.style   ={aebox, fill=#1Fill, draw=#1Line},
  aetxt/.style    ={inner sep=0pt, font=\aefont\aeBody, text=aeInk},
  aevar/.style    ={inner sep=0pt, font=\aefont\aeVar, text=aeInk},
  aehead/.style   ={inner sep=0pt, anchor=north west, font=\aefont\aeHead, text=aeInk},
  aetitle/.style  ={inner sep=0pt, anchor=north west, font=\aefont\aeTitle, text=#1},
  aenote/.style   ={inner sep=0pt, font=\aefont\aeNote, text=aeMuted},
  aearrow/.style  ={-{Stealth[length=5pt,width=4pt]}, line width=0.8pt, draw=aeArrow},
  aeline/.style   ={line width=0.8pt, draw=aeArrow},
  aeint/.style    ={aearrow, dash pattern=on 3.2pt off 2.2pt},
  aepill/.style   ={rounded corners=2pt, line width=0.45pt, inner xsep=2.4pt,
                    inner ysep=1.3pt, font=\aefont\fontsize{6.3}{7}\selectfont,
                    text height=1.4ex, text depth=0.3ex},
  aemode/.style   ={circle, minimum size=0.62cm, inner sep=0pt, line width=0.6pt,
                    draw=aeBlueLine, fill=white, text=aeBlueText,
                    font=\aefont\aeTiny\bfseries},
  aeedge/.style   ={-{Stealth[length=3.6pt,width=3pt]}, line width=0.75pt},
  aedim/.style    ={{Stealth[length=3.6pt,width=3pt]}-{Stealth[length=3.6pt,width=3pt]},
                    line width=0.7pt},
  aeport/.style   ={circle, inner sep=0pt, minimum size=3.4pt, fill=aeArrow},
}
\def\aecar#1#2{%
  \begin{scope}[shift={#1}]
    \path[#2] (0,0.16) -- (0,0.47) .. controls (0,0.57) .. (0.12,0.59)
      -- (0.46,0.62) -- (0.74,0.93) -- (1.34,0.93) -- (1.66,0.62)
      -- (1.9,0.56) .. controls (2,0.54) .. (2,0.44) -- (2,0.16) -- cycle;
  \end{scope}}
\def\aeX#1{\fpeval{18.8-0.5*(#1)}}
\path[aeregion] (0,3.95) rectangle (20,7.25);
\node[aehead] at (0.35,7.03) {Environment};
\fill[aeClrFill]            (0.35,5.15) rectangle (\aeX{15},6.62);
\fill[aeWrnFill]            (\aeX{15},5.15) rectangle (\aeX{10},6.62);
\fill[aeUrgFill]            (\aeX{10},5.15) rectangle (19.65,6.62);
\draw[aeWrnLine, line width=0.5pt, dash pattern=on 2pt off 1.5pt]
      (\aeX{15},5.15) -- (\aeX{15},6.62);
\draw[aeUrgLine!70, line width=0.5pt, dash pattern=on 2pt off 1.5pt]
      (\aeX{10},5.15) -- (\aeX{10},6.62);
\node[aepill, fill=white, draw=aeClrLine, text=aeMintText, anchor=north east]
      at (\aeX{15}-0.15,6.5) {clear};
\node[aepill, fill=white, draw=aeWrnLine, text=aeWrnLine!50!black, anchor=north]
      at ({(\aeX{15}+\aeX{10})/2},6.5) {$\aesf{WRN}$};
\node[aepill, fill=white, draw=aeUrgLine, text=aeUrgLine!80!black, anchor=north]
      at ({(\aeX{10}+18.8)/2},6.5) {$\aesf{URG}$};
\fill[aeGrayLine!55] (0.35,4.95) rectangle (19.65,5.15);
\foreach \s in {0.6,1.6,...,19.2}
  \fill[white] (\s,5.035) rectangle (\s+0.5,5.065);
\draw[aeMuted, line width=0.5pt] (\aeX{30},4.72) -- (18.8,4.72);
\foreach \m/\l in {30/30,20/20,15/15,10/10,0/0}
  { \draw[aeMuted, line width=0.5pt] (\aeX{\m},4.72) -- (\aeX{\m},4.62);
    \node[aenote, anchor=north] at (\aeX{\m},4.56) {\l}; }
\node[aenote, anchor=west] at (18.95,4.6) {m};
\draw[aeInitLine, line width=0.8pt, decorate,
      decoration={brace, mirror, amplitude=3pt}]
      (\aeX{30},4.27) -- (\aeX{20},4.27);
\node[aepill, fill=aeInitFill, draw=aeInitLine, text=aeInitLine!50!black, anchor=north]
      at ({(\aeX{30}+\aeX{20})/2},4.2) {init};
\fill[aeBlueLine, opacity=0.16] (\aeX{22},5.42) -- (18.65,5.95) -- (18.65,5.2) -- cycle;
\def\aeCarX{\fpeval{\aeX{22}-2}}%
\def\aeGhostX{\fpeval{\aeX{19.8}-2}}%
\aecar{(\aeGhostX,5.15)}{draw=aeBlueText!55, line width=0.6pt, dash pattern=on 2pt off 1.5pt, fill=white, fill opacity=0.55}
\foreach \w in {0.45,1.55}
  \draw[aeBlueText!55, line width=0.6pt, dash pattern=on 2pt off 1.5pt]
      (\aeGhostX+\w,5.32) circle (0.16);
\aecar{(\aeCarX,5.15)}{fill=aeBlueText}
\fill[white!88!aeBlueText] (\aeCarX+0.64,5.79) -- (\aeCarX+0.82,6.02) -- (\aeCarX+1.01,6.02)
      -- (\aeCarX+1.01,5.79) -- cycle
      (\aeCarX+1.09,5.79) -- (\aeCarX+1.09,6.02) -- (\aeCarX+1.3,6.02)
      -- (\aeCarX+1.52,5.79) -- cycle;
\fill[aeInitFill] (\aeCarX+1.88,5.58) rectangle (\aeCarX+1.99,5.66);
\foreach \w in {0.45,1.55}
  { \fill[aeInk] (\aeCarX+\w,5.32) circle (0.17);
    \fill[aeGrayFill] (\aeCarX+\w,5.32) circle (0.07); }
\draw[aearrow, draw=aeBlueText] (\aeCarX+0.35,6.35) -- (\aeCarX+1.55,6.35)
      node[midway, above=2.5pt, aevar, text=aeBlueText] {$v$};
\begin{scope}[shift={(19.08,5.15)}, line cap=round, aeInk, line width=2.3pt]
  \fill[aeInk] (0.02,0.92) circle (0.11);
  \draw (0.02,0.77) -- (-0.02,0.42);
  \draw (-0.02,0.42) -- (-0.15,0.03) (-0.02,0.42) -- (0.12,0.03);
  \draw (0.02,0.7) -- (-0.14,0.48) (0.02,0.7) -- (0.16,0.5);
\end{scope}
\draw[aedim, draw=aeBlueText] (\aeX{22},5.42) -- (18.8,5.42)
      node[pos=0.6, above=3pt, aevar, text=aeBlueText] {$d$};
\draw[aedim, draw=aeMintText] (\aeX{19.8},5.95) -- (18.8,5.95)
      node[pos=0.6, above=3pt, aevar, text=aeMintText] {$z$};
\node[aeport] (pu) at (0.52,3.95) {};
\node[aeport] (pd) at (1.02,3.95) {};
\path[aeregion, fill=white] (0,0) rectangle (16.45,3.62);
\node[aehead, anchor=north west] at (1.55,3.47)
      {Automated emergency braking};
\node[aetxt, anchor=north east, font=\aeTitle] at (16.25,3.45)
      {$\mathcal C_S\parallel\mathcal C_V\parallel\mathcal C_D$};
\def\aeCard#1#2#3#4#5{
  \path[aecard=#3] (#1,0.35) coordinate (#3SW) rectangle (#2,2.68) coordinate (#3NE);
  \node[aetitle=#3Text, anchor=base west] (#3T) at (#1+0.17,2.3) {#4};
  \node[aetxt, anchor=base east, text=#3Text, font=\aeTitle] (#3C) at (#2-0.17,2.3) {#5};}
\aeCard{1.55}{5.1}{aeBlue}{Safety supervisor}{$\mathcal C_S$}
\aeCard{7.6}{11.15}{aeLav}{Vehicle response}{$\mathcal C_V$}
\aeCard{12.85}{16.25}{aePeach}{Clearance}{$\mathcal C_D$}
\begin{scope}[local bounding box=aeBlueB]
\node[aemode] (mov) at (2.25,1.32) {MOV};
\node[aemode] (dcl) at (3.33,1.8) {DCL};
\node[aemode] (stp) at (4.41,1.32) {STP};
\draw[aeedge, aeWrnLine]  (mov) to[bend left=18] (dcl);
\draw[aeedge, aeUrgLine]  (dcl) to[bend left=18] (stp);
\draw[aeedge, aeUrgLine]  (mov) to[bend right=12] (stp);
\draw[aeedge, aeClrLine]  (stp.south) to[out=-120, in=-60, looseness=0.75] (mov.south);
\end{scope}
\begin{scope}[shift={(7.895,0.6)}, local bounding box=aeLavB]
  \draw[aeLavLine, line width=0.45pt] (0,0) -- (2.85,0);
  \draw[aeLavLine, line width=0.45pt] (0,0) -- (0,1.38);
  \fill[aeLavLine, opacity=0.28]
    (0,1.22) -- (0.95,1.22) -- (0.95,0.74) -- (1.85,0.74) -- (1.85,0.2)
    -- (2.75,0.2) -- (2.75,0) -- (1.85,0) -- (1.85,0.42) -- (0.95,0.42)
    -- (0.95,0.9) -- (0,0.9) -- cycle;
  \draw[aeLavText, line width=0.6pt, dash pattern=on 2pt off 1.4pt]
    (0,1.06) -- (0.95,1.06) -- (0.95,0.58) -- (1.85,0.58) -- (1.85,0.04) -- (2.75,0.04);
  \draw[aeLavText, line width=1pt]
    (0,1.0) .. controls (0.5,1.12) and (0.8,1.0) .. (1.05,0.82)
    .. controls (1.25,0.66) and (1.6,0.52) .. (1.92,0.36)
    .. controls (2.15,0.2) and (2.4,0.06) .. (2.75,0.06);
  \node[aetxt, text=aeLavText, anchor=west] at (1.0,1.2) {$u\pm\epsilon$};
  \node[aetxt, text=aeLavText, anchor=west] at (2.8,0.1) {$v$};
\end{scope}
\node[aetxt, text=aePeachText, font=\aeBody] (aePeachF) at (14.55,1.67)
      {$\aeN z \;=\; d-\Delta t\,\aeN v$};
\begin{scope}[shift={(13.375,0.73)}, local bounding box=aePeachB]
  \fill[aePeachLine] (0,0.12) rectangle (0.55,0.36);
  \fill[aePeachLine] (0.12,0.36) rectangle (0.42,0.5);
  \fill[aeInk] (0.13,0.11) circle (0.07) (0.42,0.11) circle (0.07);
  \draw[aedim, draw=aeMintText] (0.62,0.27) -- (2.25,0.27)
        node[midway, above=1.5pt, aetxt, text=aeMintText] {$z$};
  \draw[aeInk, line width=1.4pt, line cap=round] (2.36,0.04) -- (2.36,0.5);
\end{scope}
\draw[aearrow] (pd.center) -- (1.02,2.1) -- (1.55,2.1);
\draw[aearrow] (pu.center) -- (0.52,1.1) -- (1.55,1.1);
\node[aevar, anchor=west, text=aeBlueText] at (1.1,2.5) {$d$};
\node[aevar, rotate=90, anchor=south, font=\aefont\fontsize{9}{11}\selectfont] at (0.44,2.0)
      {$\aesf{URG},\aesf{WRN}$};
\fill[aeArrow] (1.02,2.98) circle (1.7pt);
\draw[aearrow] (1.02,2.98) -- (14.55,2.98) -- (14.55,2.68);
\draw[aeint] (5.1,1.5) -- (7.6,1.5);
\node[aevar, anchor=south, align=center] (aeLblS) at (6.35,1.62)
      {$u$\\[1pt]{\fontsize{9}{11}\selectfont$\aesf{MOV},\aesf{DCL},\aesf{STP}$}};
\draw[aeint] (11.15,1.5) -- (12.85,1.5);
\node[aevar, anchor=south] at (12.0,1.62) {$v$};
\foreach \c in {6.35,12.0}
  \node[aepill, fill=aePeachFill, draw=aePeachLine, text=aePeachText, anchor=north]
       at (\c,1.32) {internal};
\draw[aearrow, line width=1.05pt] (16.25,1.5) -- (17.15,1.5)
      node[midway, above=2pt, aevar, text=aeMintText] {$z$};
\path[aecard=aeMint] (17.15,0.35) coordinate (aeMintSW) rectangle (20,2.68) coordinate (aeMintNE);
\node[aetitle=aeMintText, align=center, anchor=north] (aeMintT) at (18.575,2.52) {Externally\\visible};
\node[aetxt, anchor=center] (aeMintB) at (18.575,1.12) {%
  \begin{tabular}{@{}r@{\hspace{4pt}}l@{}}
    {\aeNote\color{aeMuted}in}  & $d,\aesf{URG},\aesf{WRN}$\\[2pt]
    {\aeNote\color{aeMuted}out} & $z$
  \end{tabular}};
\end{tikzpicture}%
\endgroup

%% file: appendices.tex
\emergencystretch=2em 
\section{GR(1) Contract-Operation Algorithms}
\label{app:gr1-contract-operations}

This appendix gives the GR(1) realizations of the contract-operation algorithms
implemented in Pacti~\cite{pacti}. The underlying assume--guarantee algebra is
the one introduced in Section~\ref{sec:assume-guarantee-contracts}; CAVE
eliminates the irrelevant variables of an interconnection while
retaining a bounded-next GR(1) representation.

Throughout this appendix, let
\[
\mathcal C=(I,O,A,G),
\qquad
\mathcal C'=(I',O',A',G')
\]
denote GR(1) contracts, where $I$ and $O$ are the input and output interfaces
and $A$ and $G$ are the assumptions and guarantees, respectively. Although we
write $A$ and $G$ as complete formulas below, the implementation represents
each of them through separate \textsc{Init}, \textsc{Inv}, and \textsc{Liv}
buckets.

\subsection{CAVE primitives}
\label{app:cave-primitives}

We first recall the semantic properties of the two CAVE primitives used by the
contract operations. Let $\Gamma$ be a GR(1) context, $\phi$ a GR(1) target,
and $\bar p$ the variables to eliminate. Algorithm~\ref{alg:gr1-refinement}
returns a $\bar p$-free GR(1) formula
$\Psi^{\mathrm{ref}}=\Call{RefineCAVE}{\Gamma,\phi,\bar p}$ with
$\Gamma\land\Psi^{\mathrm{ref}}\Rightarrow\phi$; since $\Psi^{\mathrm{ref}}$ is
$\bar p$-free, this is equivalent to
$\Psi^{\mathrm{ref}}\Rightarrow\forall\bar p\,(\Gamma\Rightarrow\phi)$
(Theorem~\ref{thm:refinement-soundness}). Dually,
Algorithm~\ref{alg:gr1-relaxation} returns a $\bar p$-free GR(1) formula
$\Psi^{\mathrm{rel}}=\Call{RelaxCAVE}{\Gamma,\phi,\bar p}$ with
$\Gamma\land\phi\Rightarrow\Psi^{\mathrm{rel}}$, equivalently
$\exists\bar p\,(\Gamma\land\phi)\Rightarrow\Psi^{\mathrm{rel}}$
(Theorem~\ref{thm:relaxation-soundness}). Refinement CAVE thus yields a GR(1)
strengthening of the exact universal elimination, and relaxation CAVE a GR(1)
weakening of the exact existential elimination.

These implications, rather than equality with the quantified formulas, are the
properties required below. Both procedures operate on the structured GR(1)
representation: initial and invariant obligations are projected under the
corresponding finite shifted invariant contexts, and liveness obligations are
treated core by core. The calls in the algorithms below therefore denote the
complete structured CAVE procedures, not monolithic LTL quantifier elimination.

\subsection{GR(1) Refinement}
\label{app:gr1-refinement-check}

For contracts
$\mathcal C=(I,O,A,G)$ and
$\mathcal C'=(I,O,A',G')$ over the same interface, refinement is defined by
\begin{equation}
\mathcal C\leq\mathcal C'
\quad\Longleftrightarrow\quad
(A'\Rightarrow A)
\land
(A'\land G\Rightarrow G').
\label{eq:app-exact-refinement}
\end{equation}
These are semantic infinite-trace implications. The implementation verifies
them using a sound finite procedure over the structured bounded-next
GR(1) representation.

\begin{algorithm}[H]
\caption{GR(1) Contract Refinement Verification}
\label{alg:app-gr1-refinement-check}
\begin{algorithmic}[1]
\Input{$\mathcal C=(I,O,A,G)$,
       $\mathcal C'=(I,O,A',G')$}
\Output{\textsc{True} if $\mathcal C\leq\mathcal C'$ is certified}

\State $b_A\gets\Call{GR1Implies}{A',A}$
\State $b_G\gets
\Call{GR1Implies}{G\land A',\,G'\land A'}$
\State \Return $b_A\land b_G$

\vspace{1mm}
\Function{GR1Implies}{$S,T$}
    \State $I_S\gets\bigwedge S_{\mathrm{init}}$,
           $I_T\gets\bigwedge T_{\mathrm{init}}$
    \State $b_{\mathrm{init}}\gets
           \Call{Implies}{I_S,I_T}$
    \If{$\neg b_{\mathrm{init}}$}
        \For{$h=0,\ldots,h_{\max}$}
            \State $K_h\gets
            \displaystyle\bigwedge_{k=0}^{h}
            \Next^k\!\left(\bigwedge S_{\mathrm{inv}}\right)$
            \If{$\Call{Implies}{I_S\land K_h,I_T}$}
                \State $b_{\mathrm{init}}\gets\textsc{True}$
                \State \textbf{break}
            \EndIf
        \EndFor
    \EndIf
    \If{$\neg b_{\mathrm{init}}$}
        \State \Return \textsc{False}
    \EndIf

    \ForAll{$t\in T_{\mathrm{inv}}$}
        \State $b\gets\textsc{False}$
        \For{$h=0,\ldots,h_{\max}$}
            \State $K_h\gets
            \displaystyle\bigwedge_{k=0}^{h}
            \Next^k\!\left(\bigwedge S_{\mathrm{inv}}\right)$
            \If{$\Call{Implies}{K_h,t}$}
                \State $b\gets\textsc{True}$
                \State \textbf{break}
            \EndIf
        \EndFor
        \If{$\neg b$}
            \State \Return \textsc{False}
        \EndIf
    \EndFor

    \ForAll{$t\in T_{\mathrm{liv}}$}
        \State $b\gets
        \displaystyle\bigvee_{s\in S_{\mathrm{liv}}}
        \Call{Implies}{s,t}$
        \If{$\neg b$}
            \For{$h=0,\ldots,h_{\max}$}
                \State $K_h\gets
                \displaystyle\bigwedge_{k=0}^{h}
                \Next^k\!\left(\bigwedge S_{\mathrm{inv}}\right)$
                \If{$\Call{Implies}{K_h,t}$ \textbf{or}
                    $\displaystyle\bigvee_{s\in S_{\mathrm{liv}}}
                    \Call{Implies}{K_h\land s,t}$}
                    \State $b\gets\textsc{True}$
                    \State \textbf{break}
                \EndIf
            \EndFor
        \EndIf
        \If{$\neg b$}
            \State \Return \textsc{False}
        \EndIf
    \EndFor

    \State \Return \textsc{True}
\EndFunction
\end{algorithmic}
\end{algorithm}

For a source term list, let $I_S$ and $S_{\mathrm{inv}}$ denote the
conjunctions of its \textsc{Init} and \textsc{Inv} cores, respectively, and
let $\{S_{\mathrm{liv},j}\}_j$ denote its \textsc{Liv} cores. We write
$S_{\mathrm{inv}}^{[0:h]}:=\bigwedge_{k=0}^{h}\Next^k S_{\mathrm{inv}}$.
For each target invariant core $T_{\mathrm{inv},i}$, the verifier searches for
\begin{equation}
S_{\mathrm{inv}}^{[0:h]}
\Rightarrow T_{\mathrm{inv},i}.
\label{eq:app-refinement-inv}
\end{equation}
For the target initial conjunction $I_T$, it first checks
$I_S\Rightarrow I_T$ and, if necessary, searches for
\begin{equation}
I_S\land S_{\mathrm{inv}}^{[0:h]}
\Rightarrow I_T.
\label{eq:app-refinement-init}
\end{equation}
Finally, each target liveness core $T_{\mathrm{liv},i}$ is discharged by one of
three routes: a single source recurrence with
$S_{\mathrm{liv},j}\Rightarrow T_{\mathrm{liv},i}$, the shifted source
invariant alone, or the shifted source invariant together with one recurring
source liveness core,
\begin{align}
S_{\mathrm{inv}}^{[0:h]}&\Rightarrow T_{\mathrm{liv},i},
\label{eq:app-refinement-liv}\\
S_{\mathrm{inv}}^{[0:h]}\wedge S_{\mathrm{liv},j}
&\Rightarrow T_{\mathrm{liv},i}.
\label{eq:app-refinement-liv-supported}
\end{align}
The verifier searches $0\leq h\leq h_{\max}$, with $h_{\max}=2$ in the
experiments. Algorithm~\ref{alg:app-gr1-refinement-check} summarizes the
resulting structured verifier.

Here $\Call{Implies}{\phi,\psi}$ returns \textsc{True} only when the SMT
solver establishes $\phi\land\neg\psi$ as unsatisfiable.

\begin{proposition}[Soundness of GR(1) refinement verification]
\label{prop:app-refinement-soundness}
If Algorithm~\ref{alg:app-gr1-refinement-check} returns \textsc{True}, then
\[
\boxed{
\mathcal C\leq\mathcal C'.
}
\]
The converse need not hold.
\end{proposition}

\begin{proof}
Each successful structured check establishes the corresponding temporal
implication. By Lemma~\ref{lem:shifted-invariant}, $\Always S_{\mathrm{inv}}$
entails $\Always S_{\mathrm{inv}}^{[0:h]}$ for every finite $h$, so
\eqref{eq:app-refinement-inv} establishes each target invariant and
\eqref{eq:app-refinement-init} the target initial condition. For liveness,
$S_{\mathrm{liv},j}\Rightarrow T_{\mathrm{liv},i}$ gives
$\Always\Eventually S_{\mathrm{liv},j}\Rightarrow
\Always\Eventually T_{\mathrm{liv},i}$ by Lemma~\ref{lem:liveness-prop};
\eqref{eq:app-refinement-liv} gives $\Always T_{\mathrm{liv},i}$ and hence
$\Always\Eventually T_{\mathrm{liv},i}$; and under
\eqref{eq:app-refinement-liv-supported}, $\Always S_{\mathrm{inv}}$ supplies
$S_{\mathrm{inv}}^{[0:h]}$ at every position while
$\Always\Eventually S_{\mathrm{liv},j}$ supplies infinitely many positions at
which the support holds, so $T_{\mathrm{liv},i}$ holds at each of them.
Hence a successful call $\Call{GR1Implies}{S,T}$ establishes $S\Rightarrow T$.

Algorithm~\ref{alg:app-gr1-refinement-check} therefore establishes
$A'\Rightarrow A$ and $A'\land G\Rightarrow A'\land G'$, the latter implying
$A'\land G\Rightarrow G'$. These are the two obligations in
\eqref{eq:app-exact-refinement}. \qedbox
\end{proof}

\paragraph{Proof-search limitations.}
The bound $h_{\max}$ limits the finite certificates considered by the verifier;
it is a proof-search bound, not a semantic horizon. Even when a finite certificate of the form
\eqref{eq:app-refinement-inv} exists, the required $h$ can be arbitrarily
large relative to the $\Next$-depth of the individual cores. Moreover, some
valid infinite-trace implications admit no certificate of this finite-shift
form for any $h$. Consequently, increasing $h_{\max}$ can establish additional
refinements, but a returned \textsc{False} means only that refinement was not
established by this procedure.

\paragraph{Exact refinement.}
A complete verifier would decide the temporal implications in
\eqref{eq:app-exact-refinement} over infinite traces, rather than searching
only for finite shifted-invariant certificates. Such a procedure would require
a complete temporal reasoning backend capable of accounting for unbounded
invariant propagation and general interactions between invariant and recurrence
obligations. This is orthogonal to the finite SMT-based CAVE machinery developed
here.

\subsection{GR(1) Composition}
\label{app:gr1-composition}

Ignoring interface elimination, the exact contract-algebra composition is
\begin{equation}
\mathcal C\parallel\mathcal C'
=
(B,H),
\end{equation}
where
\begin{align}
B
&=
(A\land A')
\lor(A\land\neg G)
\lor(A'\land\neg G').
\\
H
&=
(G\lor\neg A)
\land
(G'\lor\neg A').
\end{align}

Algorithm~\ref{alg:app-gr1-composition} gives the GR(1) realization used by
the implementation. For clarity, the algorithm shows the case in which all
interconnection variables are irrelevant. The implementation additionally permits
selected interconnection variables to be retained explicitly. 

\begin{algorithm}[t]
\caption{GR(1) Contract Composition}
\label{alg:app-gr1-composition}
\begin{algorithmic}[1]
\Input{GR(1) contracts $\mathcal C=(I,O,A,G)$,
       $\mathcal C'=(I',O',A',G')$}
\Output{GR(1) relaxation
        $\widehat{\mathcal C}^{\,c}=(I^c,O^c,A^c,G^c)$ satisfying
        $\mathcal C\parallel\mathcal C'
        \leq\widehat{\mathcal C}^{\,c}$}
\Require{$O\cap O'=\emptyset$}

\State $\bar p\gets(O\cap I')\cup(I\cap O')$
\State $I^c\gets(I\cup I')\setminus\bar p$
\State $O^c\gets(O\cup O')\setminus\bar p$
\State $F_A\gets\bar p\cup O^c$

\State $I_A\gets\Vars(A)$;
       $I_{A'}\gets\Vars(A')$
       \Comment{assumption-constrained inputs}
\State $\mathsf{Cycle}\gets
       (O'\cap I\neq\emptyset)\land(O\cap I'\neq\emptyset)$
\Require{$\neg\Bigl(
         \mathsf{Cycle}\land
         \bigl[(O'\cap I_A\neq\emptyset)
         \lor(O\cap I_{A'}\neq\emptyset)\bigr]
         \Bigr)$}

\If{$O\cap I'\neq\emptyset$ and $O'\cap I=\emptyset$}
    \State $\bar A'\gets
      \Call{RefineCAVE}{A\land G,A',F_A}$
    \Require{$\Vars(\bar A')\cap F_A=\emptyset$}
    \State $A^c\gets A\land\bar A'$
\ElsIf{$O'\cap I\neq\emptyset$ and $O\cap I'=\emptyset$}
    \State $\bar A\gets
      \Call{RefineCAVE}{A'\land G',A,F_A}$
    \Require{$\Vars(\bar A)\cap F_A=\emptyset$}
    \State $A^c\gets A'\land\bar A$
\Else
    \State $A^c\gets A\land A'$
\EndIf

\State $\bar G\gets
  \Call{RelaxCAVE}{G',G,\bar p}$
\State $\bar G'\gets
  \Call{RelaxCAVE}{G,G',\bar p}$
\State $G_0\gets\bar G\land\bar G'$
\State $G_{\mathrm{pre}}\gets
  \Call{RelaxComp}{A^c,G_0,\bar p}$
\State $G^c\gets
  G_{\mathrm{pre}}\setminus
  \{\tau\in G_{\mathrm{pre}}:
  \Vars(\tau)\cap\bar p\neq\emptyset\}$

\State \Return
$\Call{ReduceGR1}{(I^c,O^c,A^c,G^c)}$
\end{algorithmic}
\end{algorithm}

Here $\Call{RelaxComp}{A^c,G_0,\bar p}$ invokes relaxation CAVE whenever
elimination remains necessary. If the target and context are already
$\bar p$-free, it retains the projected guarantee directly. The implementation
also preserves the time-zero consequence of an invariant when needed by the
structured \textsc{Init}/\textsc{Inv} representation. Any residual guarantee
term containing a variable in $\bar p$ is then discarded, further relaxing the
result. The operation $\Call{ReduceGR1}{}$ denotes structure-preserving
simplification and redundancy removal; it does not change the semantic
direction established by CAVE.

As in the underlying Pacti composition algorithm~\cite{pacti}, the
composability condition permits bidirectional feedback unless an output of
either contract feeds an input of the other that is constrained by its
assumption. Specifically, $I_A=\Vars(A)\subseteq I$ and
$I_{A'}=\Vars(A')\subseteq I'$ denote the respective assumption-constrained
inputs, and $\mathsf{Cycle}$ indicates that the two contracts are connected in
both directions. The second composability condition excludes precisely those
bidirectional cycles for which $O'\cap I_A\neq\emptyset$ or
$O\cap I_{A'}\neq\emptyset$.

The following proposition records the relation between
Algorithm~\ref{alg:app-gr1-composition} and the exact algebraic composition.

\begin{proposition}[Soundness of GR(1) composition]
\label{prop:app-composition-soundness}
Let $\widehat{\mathcal C}^{\,c}$ be returned by
Algorithm~\ref{alg:app-gr1-composition}. Then
\[
\boxed{
\mathcal C\parallel\mathcal C'
\leq
\widehat{\mathcal C}^{\,c}.
}
\]
Thus the computed GR(1) composition is a relaxation of the exact algebraic
composition.
\end{proposition}

\begin{proof}
Consider first the branch in which $\mathcal C$ supplies an input to
$\mathcal C'$. Refinement CAVE gives
\[
A\land G\land\bar A'\Rightarrow A'.
\]
Since $A^c=A\land\bar A'$, it follows that
\begin{equation}
A^c\land G\Rightarrow A'.
\label{eq:app-comp-directional}
\end{equation}
Moreover, $A^c\Rightarrow B$: if $G$ is false, then
$A\land\neg G$ establishes $B$; if $G$ is true,
\eqref{eq:app-comp-directional} gives $A'$, so $A\land A'$ establishes $B$.
The symmetric branch is identical, while $A^c=A\land A'$ implies $B$
immediately in the remaining case. Hence
\begin{equation}
A^c\Rightarrow B.
\label{eq:app-comp-assumption}
\end{equation}

Under $A^c$, the saturated exact guarantee simplifies to
\begin{equation}
A^c\land H
\;\Longleftrightarrow\;
A^c\land G\land G'.
\label{eq:app-comp-saturation}
\end{equation}
Indeed, $A^c\Rightarrow A$, and whenever $A^c\land G$ holds,
\eqref{eq:app-comp-directional} gives $A'$.

The two relaxation calls satisfy
\[
G\land G'\Rightarrow\bar G,
\qquad
G\land G'\Rightarrow\bar G',
\]
and therefore
\[
G\land G'\Rightarrow G_0.
\]
The final relaxation step gives
\[
A^c\land G_0\Rightarrow G^c.
\]
Together with \eqref{eq:app-comp-saturation},
\[
A^c\land H\Rightarrow G^c,
\]
or equivalently
\[
H\lor\neg A^c
\Rightarrow
G^c\lor\neg A^c.
\]
Combining this implication with
\eqref{eq:app-comp-assumption} gives the claimed contract refinement. \qedbox
\end{proof}

The relaxation need not be an equivalence. In particular, separate structured
eliminations may discard correlations that would be retained by unrestricted
temporal quantifier elimination.

\subsection{GR(1) Quotient}
\label{app:gr1-quotient}

For a dividend $\mathcal C$ and divisor $\mathcal C'$, the exact algebraic
quotient is
\begin{equation}
\mathcal Q_*=\mathcal C/\mathcal C'=(A_*,G_*),
\end{equation}
with
\begin{align}
A_*
&=
A\land(\neg A'\lor G'),
\\
G_*
&=
(A'\land G)\lor\neg A\lor(A'\land\neg G').
\end{align}
The guarantee can equivalently be written as
\begin{equation}
\boxed{
G_*
\equiv
A\Rightarrow
\bigl(A'\land(G'\Rightarrow G)\bigr).
}
\label{eq:app-quotient-factorization}
\end{equation}
This factorization motivates the two refinement-CAVE passes in
Algorithm~\ref{alg:app-gr1-quotient}.

\begin{algorithm}[t]
\caption{GR(1) Contract Quotient}
\label{alg:app-gr1-quotient}
\begin{algorithmic}[1]
\Input{Dividend GR(1) contract $\mathcal C=(I,O,A,G)$, divisor GR(1) contract $\mathcal C'=(I',O',A',G')$}
\Output{GR(1) refinement
        $\mathcal Q=(I^q,O^q,A^q,G^q)$ satisfying $\mathcal Q \leq \mathcal C/\mathcal C'$}

\Require $(O\setminus O')\cap I'=\emptyset$
\State $O^q\gets(O\setminus O')\cup(I'\setminus I)$
\State $I^q\gets(I\setminus I')\cup(O'\setminus O)$
\State $\bar p\gets(O\cap O')\cup(I\cap I')$
\State $F_A\gets\bar p\cup O^q$

\If{$\Call{GR1Implies}{A,A'}$}
    \State $A^q\gets
      \Call{RelaxCAVE}{G',A,F_A}$
\Else
    \State $A^q\gets
      \Call{RelaxCAVE}{\top,A,F_A}$
\EndIf

\State $G_1\gets
  \Call{RefineCAVE}{G',G,\bar p}$
\State $G_2\gets G_1\land A'$
\State $G^q\gets
  \Call{RefineCAVE}{A,G_2,\bar p}$
\Require $G^q$ contains no variable in $\bar p$
\State \Return
$\Call{ReduceGR1}{(I^q,O^q,A^q,G^q)}$
\end{algorithmic}
\end{algorithm}

The implementation treats an unsuccessful intermediate refinement elimination
conservatively: the pre-elimination target is retained, and a quotient is
returned only if no eliminated variable remains in the final guarantee. This
does not alter the implication directions used below.

\begin{proposition}[Soundness of GR(1) quotient]
\label{prop:app-quotient-soundness}
Let $\mathcal Q$ be returned by
Algorithm~\ref{alg:app-gr1-quotient}, and let $\mathcal Q_*=\mathcal C/\mathcal C'$
denote the exact quotient. Then
\[
\boxed{
\mathcal Q\leq\mathcal Q_*.
}
\]
Thus the computed GR(1) quotient is a refinement of the exact quotient.
\end{proposition}

\begin{proof}
First consider the guarantee. The first refinement-CAVE call gives
\[
G_1
\Rightarrow
\forall\bar p\,(G'\Rightarrow G).
\]
Hence
\[
G_2=G_1\land A'
\Rightarrow
A'\land\forall\bar p\,(G'\Rightarrow G).
\]
The second refinement call gives
\[
A\land G^q\Rightarrow G_2.
\]
Using \eqref{eq:app-quotient-factorization}, this implies
\[
G^q
\Rightarrow
\forall\bar p\,G_*.
\]

For the assumption, suppose first that the sound refinement checker establishes
$A\Rightarrow A'$. Then
\[
A_*
=
A\land(\neg A'\lor G')
\equiv
A\land G',
\]
and relaxation CAVE yields
\[
\exists F_A\,(A\land G')
\Rightarrow A^q.
\]
Thus the existential elimination of $A_*$ implies $A^q$.

If $A\Rightarrow A'$ is not established, the algorithm instead projects $A$
under the trivial context. Since
\[
A_*=A\land(\neg A'\lor G')\Rightarrow A,
\]
we again obtain
\[
\exists F_A\,A_*\Rightarrow A^q.
\]
Since $A_*\Rightarrow\exists F_A\,A_*$ and $\forall\bar p\,G_*\Rightarrow G_*$,
we obtain $A_*\Rightarrow A^q$ and $A_*\land G^q\Rightarrow G_*$. These are
precisely the two obligations for $\mathcal Q\leq\mathcal Q_*$. \qedbox
\end{proof}

The refinement test used in the assumption branch is sound. Consequently, a
positive result justifies the first branch, whereas failure to establish the
implication merely selects the more conservative fallback. Completeness of the
refinement checker is not required for Proposition~\ref{prop:app-quotient-soundness}.

\begin{corollary}[Missing-component soundness]
\label{cor:app-missing-component}
Let $\mathcal Q$ be returned by Algorithm~\ref{alg:app-gr1-quotient}.
Then
\[
    \mathcal C' \parallel \mathcal Q \le \mathcal C,
\]
where $\parallel$ denotes the exact composition of the contract algebra.
\end{corollary}

\begin{proof}
By Proposition~\ref{prop:app-quotient-soundness},
$\mathcal Q\le\mathcal Q_*=\mathcal C/\mathcal C'$.
Since exact composition is monotone and the exact quotient satisfies
$\mathcal C'\parallel\mathcal Q_*\le\mathcal C$,
\[
    \mathcal C'\parallel\mathcal Q
    \le
    \mathcal C'\parallel\mathcal Q_*
    \le
    \mathcal C.
\]
\end{proof}

Thus the computed quotient is a sound missing-component specification, as in
the post-processing principle of Pacti~\cite{pacti}.

\subsection{GR(1) Merging}
\label{app:gr1-merging}

Merging requires no variable elimination. The implementation first lifts both
contracts to the compatible union interface and then merges their structured
GR(1) obligations.

\begin{algorithm}[t]
\caption{GR(1) Contract Merging}
\label{alg:app-gr1-merging}
\begin{algorithmic}[1]
\Input{GR(1) contracts $\mathcal C=(I,O,A,G)$,
       $\mathcal C'=(I',O',A',G')$}
\Output{Merged GR(1) contract $\mathcal C^m$}
\Require shared variable sorts agree
\Require
$(I\cup I')\cap(O\cup O')=\emptyset$
\State $I^m\gets I\cup I'$
\State $O^m\gets O\cup O'$
\State $A^m\gets\Call{MergeBuckets}{A,A'}$
\State $G^m\gets\Call{MergeBuckets}{G,G'}$
\State \Return $(I^m,O^m,A^m,G^m)$
\end{algorithmic}
\end{algorithm}

Here $\Call{MergeBuckets}{}$ conjoins the \textsc{Init} and \textsc{Inv}
obligations and retains all liveness cores as separate recurrence obligations.
In particular,
\[
\Always\Eventually\phi
\land
\Always\Eventually\psi
\]
is retained as two liveness obligations and is not replaced by
$\Always\Eventually(\phi\land\psi)$.

After lifting both contracts to the compatible union interface,
Algorithm~\ref{alg:app-gr1-merging} is exact: it computes
\[
A^m=A\land A',
\qquad
G^m=G\land G'.
\]
No refinement or relaxation CAVE call is required.

\section{Setting for Proofs}
\label{app:setting}
We use the temporal semantics of Section~\ref{sec:preliminaries}, extended with
quantification over typed variables. Let $V$ be a fixed set of variables, where
each $x\in V$ has domain $D_x$ (Boolean propositions are the case
$D_x=\{\bot,\top\}$). A \emph{trace} $\sigma=\sigma_0\sigma_1\cdots$ assigns
every $x\in V$ a value $\sigma_i(x)\in D_x$ at each time $i\in\mathbb N$, and
atomic predicates are evaluated under the background-theory interpretation. A
\emph{temporal valuation} of a variable $p$ is a function
$P:\mathbb N\to D_p$, and $\sigma[p\mapsto P]$ denotes the trace that agrees
with $\sigma$ on every variable other than $p$ and satisfies
$\sigma[p\mapsto P]_i(p)=P(i)$. Quantification ranges over temporal valuations:
\[
  (\sigma,i)\models\forall p.\,\psi
  \iff
  (\sigma[p\mapsto P],i)\models\psi\ \text{for every }P:\mathbb N\to D_p,
\]
and $(\sigma,i)\models\exists p.\,\psi$ iff this holds for some $P$.

\section{Auxiliary Lemmas}
\label{app:auxiliary-lemmas}

\subsection{Temporal Facts}

\begin{lemma}[Propagation of Liveness under Invariant Implication]
\label{lem:liveness-prop}
For any LTL formulas $\varphi$ and $\psi$,
\[
  \Always(\varphi \rightarrow \psi)
  \;\Longrightarrow\;
  \big(\Always \Eventually \varphi \rightarrow \Always \Eventually \psi\big).
\]
\end{lemma}

\begin{proof}
Let $\sigma\models\Always(\varphi\rightarrow\psi)\wedge\Always\Eventually\varphi$
and fix $i\in\mathbb N$. Some $j\ge i$ satisfies $(\sigma,j)\models\varphi$, and
$(\sigma,j)\models\varphi\rightarrow\psi$ then gives $(\sigma,j)\models\psi$.
Since $i$ was arbitrary, $\sigma\models\Always\Eventually\psi$. \qedbox
\end{proof}

\begin{lemma}[One-way commutation of $\exists p$ and $\Always$]
\label{lem:exists-always}
For every temporal formula $\varphi(p)$,
\[
  \exists p.\,\Always\,\varphi(p)
  \;\Rightarrow\;
  \Always\,\exists p.\,\varphi(p).
\]
\end{lemma}

\begin{proof}
Let $\sigma\models\exists p.\,\Always\,\varphi(p)$, witnessed by a temporal
valuation $P$, so that $(\sigma[p\mapsto P],i)\models\varphi(p)$ for every
$i\in\mathbb N$. The same $P$ witnesses $(\sigma,i)\models\exists p.\,\varphi(p)$
at each $i$, hence $\sigma\models\Always\,\exists p.\,\varphi(p)$. \qedbox
\end{proof}

\begin{lemma}[Removal of Redundant Liveness Terms]
\label{lem:redundant-liveness}
If $\alpha\rightarrow\beta$ is valid, then
$\Always\Eventually\alpha\Rightarrow\Always\Eventually\beta$, and consequently
$\Always\Eventually\alpha\wedge\Always\Eventually\beta
\equiv\Always\Eventually\alpha$.
\end{lemma}

\begin{proof}
Validity of $\alpha\rightarrow\beta$ gives $\Always(\alpha\rightarrow\beta)$, so
the implication follows from Lemma~\ref{lem:liveness-prop}. The equivalence
holds because $\Always\Eventually\alpha$ implies the second conjunct, while the
conjunction trivially implies $\Always\Eventually\alpha$. \qedbox
\end{proof}

\subsection{Bounded-Next Elimination}
\label{app:bounded-next}

\begin{lemma}[Shifted Consequences of an Invariant]
\label{lem:shifted-invariant}
For any LTL formula \(\rho\) and any \(h\in\mathbb N\),
\[
  \Always \rho
  \;\Longrightarrow\;
  \Always\left(\bigwedge_{k=0}^{h}\Next^k\rho\right).
\]
\end{lemma}

\begin{proof}
Assume \(\sigma\models \Always\rho\) and fix a time \(t\). For every
\(k\in\{0,\dots,h\}\) we have \((\sigma,t+k)\models \rho\), that is,
\((\sigma,t)\models \Next^k\rho\). Hence
\((\sigma,t)\models \bigwedge_{k=0}^{h}\Next^k\rho\) for every \(t\). \qedbox
\end{proof}

\begin{lemma}[Bounded-Next Invariant Guard]
\label{lem:bounded-next-invariant-guard}
Let $B$, $h_B$, $D_B$, and $\Gamma_{\mathrm{inv}}^{[0:h_B]}$ be as in
Section~\ref{sec:finite-window-cave}. If $\sigma\models\Always\gamma_{\mathrm{inv}}$,
then for every time $t$,
\[
  v_{\sigma,t,D_B}
  \models
  \big(\Gamma_{\mathrm{inv}}^{[0:h_B]}\big)^\flat.
\]
\end{lemma}

\begin{proof}
By Lemma~\ref{lem:shifted-invariant}, $\Gamma_{\mathrm{inv}}^{[0:h_B]}$ holds at
every time. Each conjunct $\Next^k\gamma_{\mathrm{inv}}$ with $k\le h_B$ has
next-depth $k+d(\gamma_{\mathrm{inv}})\le D_B$, so
Lemma~\ref{lem:flattening-correctness} gives the claim. \qedbox
\end{proof}

\begin{lemma}[Fixed vs.\ Time-Varying Support]
\label{lem:fixed-liveness-support}
For a fixed target liveness core $\phi_i$, let $\psi^{ref}_{i,j}$, $j\in J$, be as in
Algorithm~\ref{alg:gr1-refinement}. Then
\[
  \Bigl(\bigvee_{j\in J}\Always\psi^{ref}_{i,j}\Bigr)
  \wedge \Always\gamma_{\mathrm{inv}}
  \wedge \bigwedge_{j\in J}\Always\Eventually\gamma_{\mathrm{liv},j}
  \;\Longrightarrow\;
  \Always\Eventually\phi_i .
\]
In contrast, the time-varying condition
\[
  \Always\Bigl(\bigvee_{j\in J}\psi^{ref}_{i,j}\Bigr)
  \wedge \Always\gamma_{\mathrm{inv}}
  \wedge \bigwedge_{j\in J}\Always\Eventually\gamma_{\mathrm{liv},j}
\]
does not, in general, imply $\Always\Eventually\phi_i$, nor
$\Always\Eventually\bigvee_{j}(\gamma_{\mathrm{liv},j}\wedge\psi^{ref}_{i,j})$.
\end{lemma}

\begin{proof}
First claim: some fixed $j^\star\in J$ satisfies $\Always\psi^{ref}_{i,j^\star}$.
Together with $\Always\gamma_{\mathrm{inv}}$ and
$\Always\Eventually\gamma_{\mathrm{liv},j^\star}$, and with
$h_{i,j^\star}=\max\{d(\phi_i),d(\gamma_{\mathrm{liv},j^\star})\}$,
Lemma~\ref{lem:fixed-support-liveness-refinement}
yields $\Always\Eventually\phi_i$.

Second claim: let $\bar p=\emptyset$, $\gamma_{\mathrm{inv}}=\top$, $\phi_i=\bot$,
$J=\{1,2\}$, $\gamma_{\mathrm{liv},1}=a$, $\gamma_{\mathrm{liv},2}=\neg a$ for a
Boolean variable $a$. Then $\psi^{ref}_{i,j}=(\gamma_{\mathrm{liv},j}\Rightarrow\bot)$,
that is, $\psi^{ref}_{i,1}=\neg a$ and $\psi^{ref}_{i,2}=a$.
On the trace where $a$ holds exactly at even times, both
$\Always\Eventually\gamma_{\mathrm{liv},j}$ hold and
$\Always(\psi^{ref}_{i,1}\vee\psi^{ref}_{i,2})\equiv\Always\top$ holds, yet
$\Always\Eventually\phi_i=\Always\Eventually\bot$ fails, and
$\gamma_{\mathrm{liv},j}\wedge\psi^{ref}_{i,j}\equiv\bot$ for each $j$. \qedbox
\end{proof}

\begin{lemma}[Bounded-Next Existential Relaxation]
\label{lem:bounded-next-existential-relaxation}
Let $B=\{\chi\}$, so that $h_B=d(\chi)$ and $D_B=h_B+d(\gamma_{\mathrm{inv}})$ as in
Section~\ref{sec:finite-window-cave}, and define
\[
  \psi(\bar q_0,\dots,\bar q_{D_B})
  :=
  \exists\bar p_0,\dots,\bar p_{D_B}.
  \left[
    \big(\Gamma_{\mathrm{inv}}^{[0:h_B]}\big)^\flat
    \wedge
    \chi^\flat
  \right].
\]
Then
\[
  \exists\bar p.\bigl(\Always\gamma_{\mathrm{inv}}(\bar p,\bar q)\wedge\Always\Eventually\chi(\bar p,\bar q)\bigr)
  \;\Longrightarrow\;
  \Always\Eventually\psi(\bar q_0,\dots,\bar q_{D_B}).
\]
\end{lemma}

\begin{proof}
Let $P$ be a temporal valuation of $\bar p$ such that
$\sigma':=\sigma[\bar p\mapsto P]$ satisfies
$\Always\gamma_{\mathrm{inv}}\wedge\Always\Eventually\chi$, and let
$t_0\in\mathbb N$. Some $t\ge t_0$ has $(\sigma',t)\models\chi$. By
Lemma~\ref{lem:bounded-next-invariant-guard},
$v_{\sigma',t,D_B}\models(\Gamma_{\mathrm{inv}}^{[0:h_B]})^\flat$, and by
Lemma~\ref{lem:flattening-correctness}, $v_{\sigma',t,D_B}\models\chi^\flat$.
The $\bar p$-copies of this valuation therefore witness the existential
quantifier in $\psi$, and since $v_{\sigma',t,D_B}$ agrees with
$v_{\sigma,t,D_B}$ on the copies of $\bar q$, we obtain $(\sigma,t)\models\psi$.
As $t_0$ was arbitrary, $\sigma\models\Always\Eventually\psi$. \qedbox
\end{proof}

\begin{lemma}[Optimality of Universal Elimination]
\label{lem:universal-projection-optimal}
Let \(F(\bar p,\bar q)\) be a propositional or SMT formula and define
\(\psi(\bar q):=\forall \bar p.\,F(\bar p,\bar q)\). Then
\begin{enumerate}
  \item \(\psi\Rightarrow F\) for every valuation of \(\bar p\);
  \item for every \(\bar p\)-free formula \(\eta(\bar q)\) with
  \(\eta(\bar q)\Rightarrow F(\bar p,\bar q)\) for all \(\bar p\), we have
  \(\eta\Rightarrow\psi\).
\end{enumerate}
Hence \(\psi\) is the weakest \(\bar p\)-free condition sufficient for \(F\)
uniformly over the eliminated variables.
\end{lemma}

\begin{proof}
Claim~1 is the semantics of \(\forall\bar p\). For Claim~2, suppose
\(\eta(\bar q)\) holds. By assumption \(F(\bar p,\bar q)\) holds for every
valuation of \(\bar p\), that is, \(\forall\bar p.\,F(\bar p,\bar q)\), which is
\(\psi(\bar q)\). \qedbox
\end{proof}

\begin{lemma}[Optimality of Existential Elimination]
\label{lem:existential-projection-optimal}
Let \(F(\bar p,\bar q)\) be a propositional or SMT formula and define
\(\psi(\bar q):=\exists\bar p.\,F(\bar p,\bar q)\). Then
\begin{enumerate}
  \item \(F\Rightarrow\psi\);
  \item for every \(\bar p\)-free formula \(\eta(\bar q)\) with
  \(F\Rightarrow\eta(\bar q)\), we have \(\psi\Rightarrow\eta\).
\end{enumerate}
Hence \(\psi\) is the strongest \(\bar p\)-free consequence of \(F\).
\end{lemma}

\begin{proof}
For Claim~1, a valuation of \((\bar p,\bar q)\) satisfying \(F\) witnesses
\(\exists\bar p.\,F\). For Claim~2, suppose \(\psi(\bar q)\) holds. Then some
valuation of \(\bar p\) satisfies \(F(\bar p,\bar q)\), and the assumed
implication gives \(\eta(\bar q)\). Since \(\eta\) is \(\bar p\)-free, this
conclusion does not depend on the witness. \qedbox
\end{proof}

\section{Soundness Proofs}
\label{app:soundness}

\subsection{Soundness of Context-Aided Variable Elimination by Refinement}
\label{app:soundness-ref}

Recall $\Gamma$ and $\phi$ from Section~\ref{sec:gr1-cave}; the exact objective
is $\Psi^{ref}_{exact}:=\forall\bar p\,(\Gamma\Rightarrow\phi)$.
Algorithm~\ref{alg:gr1-refinement} returns
\[
\Psi^{ref}=\psi^{ref}_{\mathrm{init}}
\wedge\Always\psi^{ref}_{\mathrm{inv}}
\wedge\bigwedge_{i\in I_0}\Always\Eventually\psi^{ref}_{i,0}
\wedge\bigwedge_{i\in I_{\mathrm{sup}}}\Always\psi^{ref}_{i,j_i}.
\]
The initial and invariant implications follow exactly as in
Lemmas~\ref{lem:bounded-next-invariant-guard} and
\ref{lem:flattening-correctness}: the shifted invariant supplies the finite
window, and universal elimination makes the implication hold for every typed
valuation of the eliminated indexed variables. The liveness targets are handled
branch by branch: Lemmas~\ref{lem:synthetic-liveness-refinement} and
\ref{lem:fixed-support-liveness-refinement} show that each candidate produced by
\textsc{SelectSupport} is individually sound.

\begin{lemma}[Synthetic liveness refinement]
\label{lem:synthetic-liveness-refinement}
For $h_{i,0}=d(\phi_i)$ and
\[
\psi^{ref}_{i,0}:=\forall\bar p.\bigl(
(\Gamma_{\mathrm{inv}}^{[0:h_{i,0}]})^\flat
\Rightarrow\phi_i^\flat\bigr),
\]
\[
\Always\gamma_{\mathrm{inv}}\wedge
\Always\Eventually\psi^{ref}_{i,0}
\Rightarrow\Always\Eventually\phi_i.
\]
\end{lemma}
\begin{proof}
Fix a time $t_0$. Recurrence of $\psi^{ref}_{i,0}$ gives $t\ge t_0$ at which it
holds, and $\Always\gamma_{\mathrm{inv}}$ supplies
$\Gamma_{\mathrm{inv}}^{[0:h_{i,0}]}$ at the same $t$
(Lemma~\ref{lem:bounded-next-invariant-guard}). By exact flattening
(Lemma~\ref{lem:flattening-correctness}), the actual typed values on this finite
window form a valuation of the indexed variables. The universal elimination
holds for every valuation of the eliminated indexed variables, hence for these
actual values, and yields $\phi_i$ at $t$. Since such a $t$ exists after every
$t_0$, $\Always\Eventually\phi_i$ holds. \qedbox
\end{proof}

\begin{lemma}[Fixed-support liveness refinement]
\label{lem:fixed-support-liveness-refinement}
For $j\in J$, let
$h_{i,j}=\max\{d(\phi_i),d(\gamma_{\mathrm{liv},j})\}$ and
\[
\psi^{ref}_{i,j}:=\forall\bar p.\bigl(
(\Gamma_{\mathrm{inv}}^{[0:h_{i,j}]})^\flat
\wedge\gamma_{\mathrm{liv},j}^\flat
\Rightarrow\phi_i^\flat\bigr).
\]
Then
\[
\Always\gamma_{\mathrm{inv}}\wedge
\Always\Eventually\gamma_{\mathrm{liv},j}\wedge
\Always\psi^{ref}_{i,j}
\Rightarrow\Always\Eventually\phi_i.
\]
\end{lemma}
\begin{proof}
After every $t_0$, recurrence supplies a time $t\ge t_0$ satisfying
$\gamma_{\mathrm{liv},j}$. The invariant supplies the shifted finite window
$\Gamma_{\mathrm{inv}}^{[0:h_{i,j}]}$ at $t$
(Lemma~\ref{lem:bounded-next-invariant-guard}), while
$\Always\psi^{ref}_{i,j}$ supplies the projected implication at that same
position. Exact flattening and universal elimination, over the typed domains of
all eliminated indexed variables, therefore yield $\phi_i$ at $t$, so
$\phi_i$ recurs. In contrast,
$\Always\Eventually\gamma_{\mathrm{liv},j}\wedge
\Always\Eventually\psi^{ref}_{i,j}$ is insufficient in general, because the two
recurrent formulas may hold at disjoint sets of positions. For instance, with
$\bar p=\emptyset$, $\gamma_{\mathrm{inv}}=\top$, $\phi_i=\bot$, and
$\gamma_{\mathrm{liv},j}=a$, we have $\psi^{ref}_{i,j}=\neg a$, and both
$\Always\Eventually a$ and $\Always\Eventually\neg a$ hold on a trace
alternating $a$ while $\Always\Eventually\bot$ fails. \qedbox
\end{proof}

\begin{theorem}[Soundness of Refinement]
\label{thm:refinement-soundness}
The formula returned by Algorithm~\ref{alg:gr1-refinement} satisfies
\[
\Psi^{ref}(\bar q)\Rightarrow
\forall\bar p\,(\Gamma(\bar p,\bar q)\Rightarrow\phi(\bar p,\bar q)).
\]
\end{theorem}
\begin{proof}
The initial and invariant finite-window arguments establish their corresponding
target obligations for every typed temporal valuation of $\bar p$.
Lemma~\ref{lem:synthetic-liveness-refinement} handles every $i\in I_0$, and
Lemma~\ref{lem:fixed-support-liveness-refinement} handles every
$i\in I_{\mathrm{sup}}$, because $\Gamma$ supplies the invariant and every
genuine support recurrence. Each candidate is individually sound, so the argument
uses only the soundness of the candidate that $\operatorname{Select}$ returns and
does not depend on how it is chosen. Conjoining these obligations proves the
claim. \qedbox
\end{proof}

\subsection{Soundness of Context-Aided Variable Elimination by Relaxation}
\label{app:soundness-rel}

The exact relaxation objective is
$\Psi^{rel}_{exact}:=\exists\bar p.\,(\Gamma\wedge\phi)$, and
Algorithm~\ref{alg:gr1-relaxation} returns
$\Psi^{rel}=\psi^{rel}_{\mathrm{init}}
\wedge\Always\psi^{rel}_{\mathrm{inv}}
\wedge\bigwedge_{\rho\in\mathcal P^\star}\Always\Eventually\rho$.
We prove $\Psi^{rel}_{exact}\Rightarrow\Psi^{rel}$ component by component,
using the cores $\psi^{rel}_{\mathrm{init}}$, $\psi^{rel}_{\mathrm{inv}}$,
$\psi^\gamma_j$, and $\psi^\phi_i$ and the set $\mathcal P^\star$ of
Algorithm~\ref{alg:gr1-relaxation}.

\paragraph{Initial relaxation.}
A valuation of $\bar p$ witnessing $\Psi^{rel}_{exact}$ satisfies
$\gamma_{\mathrm{init}}\wedge\Always\gamma_{\mathrm{inv}}\wedge\phi_{\mathrm{init}}$.
By Lemmas~\ref{lem:bounded-next-invariant-guard}
and~\ref{lem:flattening-correctness}, the indexed values it induces at time $0$
satisfy the conjunction inside the existential quantifier of
$\psi^{rel}_{\mathrm{init}}$, so
$\Psi^{rel}_{exact}\Rightarrow\psi^{rel}_{\mathrm{init}}$.

\paragraph{Invariant relaxation.}
Any valuation witnessing $\Psi^{rel}_{exact}$ satisfies
$\Always\gamma_{\mathrm{inv}}\wedge\Always\phi_{\mathrm{inv}}$. At each time,
Lemmas~\ref{lem:bounded-next-invariant-guard}
and~\ref{lem:flattening-correctness} show that the indexed values induced by
that same valuation satisfy the formula inside the existential quantifier of
$\psi^{rel}_{\mathrm{inv}}$, so
$\Psi^{rel}_{exact}\Rightarrow\Always\psi^{rel}_{\mathrm{inv}}$. This is the
one-sided direction $\exists\bar p.\,\Always F\Rightarrow\Always\exists\bar p.\,F$
of Lemma~\ref{lem:exists-always} at the level of finite windows; the converse
equivalence is not asserted.

\paragraph{Liveness relaxation.}
Let $\mathcal P=\{\psi^\gamma_j:j\in J\}\cup\{\psi^\phi_i:i\in I\}$ be the
projected liveness cores before reduction. A valuation witnessing
$\Psi^{rel}_{exact}$ satisfies $\Always\gamma_{\mathrm{inv}}$ together with
$\Always\Eventually\gamma_{\mathrm{liv},j}$ for every $j\in J$ and
$\Always\Eventually\phi_i$ for every $i\in I$. Lemma~\ref{lem:bounded-next-existential-relaxation},
applied to each $\gamma_{\mathrm{liv},j}$ and each $\phi_i$, gives
$\Psi^{rel}_{exact}\Rightarrow\bigwedge_{\rho\in\mathcal P}\Always\Eventually\rho$.

\begin{lemma}[Exactness of \textsc{SelectLiveness} reduction]
\label{lem:select-liveness-exact-reduction}
Let $\mathcal P$ be the projected liveness cores before reduction and
$\mathcal P^\star$ the set returned by \textsc{SelectLiveness}. Then
\[
\bigwedge_{\rho\in\mathcal P}\Always\Eventually\rho
\equiv
\bigwedge_{\rho\in\mathcal P^\star}\Always\Eventually\rho.
\]
\end{lemma}
\begin{proof}
Removing a core equivalent to $\top$ (as $\Always\Eventually\top$ is valid) or
a duplicate of an equivalent core is immediate. If $\rho_b$ is removed because
$\rho_a\Rightarrow\rho_b$, then $\Always\Eventually\rho_a\Rightarrow\Always\Eventually\rho_b$
by Lemma~\ref{lem:redundant-liveness}, so the conjunction is unchanged.
Applying this argument to every removal proves the claim. The implication test
is not complete for temporal redundancy: for example,
$\Always\Eventually p\equiv\Always\Eventually\Next p$ although
$p\not\Rightarrow\Next p$. This affects only simplification, not the
equivalence of the returned conjunction. \qedbox
\end{proof}

By Lemma~\ref{lem:select-liveness-exact-reduction}, the liveness consequence
established above holds with $\mathcal P^\star$ in place of $\mathcal P$.

\begin{theorem}[Soundness of Relaxation]
\label{thm:relaxation-soundness}
The formula returned by the relaxation CAVE algorithm,
\[
  \Psi^{rel}
  :=
  \psi^{rel}_{\mathrm{init}}
  \wedge
  \Always\psi^{rel}_{\mathrm{inv}}
  \wedge
  \bigwedge_{\rho\in\mathcal P^\star}
  \Always\Eventually\rho
\]
satisfies
\[
  \exists\bar p.\,
  \bigl(
    \Gamma(\bar p,\bar q)
    \wedge
    \phi(\bar p,\bar q)
  \bigr)
  \Rightarrow
  \Psi^{rel}(\bar q).
\]
\end{theorem}

\begin{proof}
The initial, invariant, and liveness implications were established above, the
latter with $\mathcal P^\star$ by
Lemma~\ref{lem:select-liveness-exact-reduction}. Conjoining them gives
$\exists\bar p.\,(\Gamma\wedge\phi)\Rightarrow\Psi^{rel}$. \qedbox
\end{proof}

\section{Optimality Proofs}
\label{app:optimality}

The following results use the admissible finite-subproblem classes of
Definitions~\ref{def:admissible-gr1-cave} and~\ref{def:gr1-optimality}. They do
not assert unrestricted temporal quantifier elimination.

\begin{lemma}[Refinement finite-subproblem optimality]
\label{lem:gr1-optimal-ref-subproblems}
The initial elimination $\psi^{ref}_{\mathrm{init}}$ and invariant elimination
$\Always\psi^{ref}_{\mathrm{inv}}$ are weakest within their corresponding
fixed finite refinement subproblems. For each liveness target $i$,
$\Always\Eventually\psi^{ref}_{i,0}$ is weakest in the synthetic recurrent
branch, and for every fixed $j\in J$,
$\Always\psi^{ref}_{i,j}$ is weakest in the persistent branch using
$\gamma_{\mathrm{liv},j}$.
\end{lemma}
\begin{proof}
Fix a subproblem or certificate branch $b$ of
Definition~\ref{def:admissible-gr1-cave}. Let $F^{temp}_b$ be its finite
obligation (the implication in the refinement table), let $F_b$ be its
flattening, and let $\psi_b=\forall\bar p.\,F_b$ be the core that
Algorithm~\ref{alg:gr1-refinement} computes for $b$. For any admissible
$\bar p$-free core $\eta$ of $b$, admissibility gives $\eta\models F^{temp}_b$.
Flattening is exact (Lemma~\ref{lem:flattening-correctness}), so
$\eta^\flat\Rightarrow F_b$ for every valuation of the eliminated indexed
variables. Since $\eta^\flat$ is $\bar p$-free,
Lemma~\ref{lem:universal-projection-optimal} gives
$\eta^\flat\Rightarrow\forall\bar p.\,F_b$, and exactness of flattening
transfers this back to $\eta\Rightarrow\psi_b$. This yields the claim for
$\psi^{ref}_{\mathrm{init}}$ and for the cores of the INV, synthetic, and
fixed-support branches. For the branches whose certificate carries a temporal
wrapper, the order is preserved under it: $\eta\Rightarrow\psi_b$ gives
$\Always\eta\Rightarrow\Always\psi_b$ for the INV and fixed-support branches, and
$\Always\Eventually\eta\Rightarrow\Always\Eventually\psi_b$ for the synthetic
branch by Lemma~\ref{lem:liveness-prop}. \qedbox
\end{proof}

\paragraph{$\bar p$-free targets.}
If $\phi_i$ is $\bar p$-free, then $\phi_i\Rightarrow\psi^{ref}_{i,0}$, because
$\phi_i$ implies $(\Gamma_{\mathrm{inv}}^{[0:h_{i,0}]})^\flat\Rightarrow\phi_i^\flat$
for every eliminated-variable valuation. Hence
$\Always\Eventually\phi_i\Rightarrow\Always\Eventually\psi^{ref}_{i,0}$, so no
separate retention case is needed. No claim is made that finite core
implication computes exact temporal dominance between different liveness
branches; bounded-next windows can overlap across successive positions.

\begin{lemma}[Relaxation finite-subproblem optimality]
\label{lem:gr1-optimal-rel-subproblems}
For every fixed finite relaxation subproblem with flattened formula $F$, the
elimination $\exists\bar p.F$ is its strongest $\bar p$-free consequence.
Consequently, $\psi^{rel}_{\mathrm{init}}$ is strongest for INIT,
$\Always\psi^{rel}_{\mathrm{inv}}$ is strongest for the admissible INV branch,
and each $\Always\Eventually\rho$ generated before
\textsc{SelectLiveness} is strongest for its corresponding liveness
subproblem.
\end{lemma}
\begin{proof}
Let $F^{temp}$ be the finite obligation of the subproblem (the entailment in
the relaxation table of Definition~\ref{def:admissible-gr1-cave}) and $F$ its
flattening, so that Algorithm~\ref{alg:gr1-relaxation} sets the corresponding
core to $\exists\bar p.F$. For any admissible $\bar p$-free consequence
$\eta$, admissibility gives $F^{temp}\models\eta$, and exact flattening
(Lemma~\ref{lem:flattening-correctness}) gives $F\Rightarrow\eta^\flat$. Then
Lemma~\ref{lem:existential-projection-optimal} gives
$\exists\bar p.F\Rightarrow\eta^\flat$, which transfers back to
$\psi^{rel}_{\mathrm{init}}\Rightarrow\eta$ for INIT,
$\psi^{rel}_{\mathrm{inv}}\Rightarrow\eta$ for INV,
$\psi^\gamma_j\Rightarrow\eta$ for context LIV, and
$\psi^\phi_i\Rightarrow\eta$ for target LIV. Lifting under $\Always$ for INV and
under $\Always\Eventually$ for LIV (Lemma~\ref{lem:liveness-prop}) preserves
the order. \qedbox
\end{proof}

\begin{lemma}[Optimality of the complete decomposed relaxation]
\label{lem:gr1-optimal-complete-relaxation}
Within the decomposed finite-subproblem relaxation class of
Definition~\ref{def:admissible-gr1-cave}, the formula $\Psi^{rel}$ returned by
Algorithm~\ref{alg:gr1-relaxation} is the strongest admissible complete
relaxation consequence.
\end{lemma}
\begin{proof}
Let $E=\bigwedge_s E_s$ be any complete admissible relaxation, with
independently admissible INIT, INV, context-LIV, and target-LIV terms
$E_s$, and let $R_s$ be the corresponding CAVE results before liveness
reduction. By Lemma~\ref{lem:gr1-optimal-rel-subproblems}, $R_s\Rightarrow E_s$
for every $s$, so $\bigwedge_s R_s\Rightarrow E$. By
Lemma~\ref{lem:select-liveness-exact-reduction}, \textsc{SelectLiveness} does
not change this conjunction, hence $\Psi^{rel}\Rightarrow E$. This proves
strongest-consequence optimality within the class; it does not cover joint
consequences that use dependencies between distinct subproblems, which
Definition~\ref{def:admissible-gr1-cave} excludes. \qedbox
\end{proof}

There is no analogous complete-refinement optimality conclusion: different
liveness certificate branches can be incomparable, and independent branch
selection need not yield an undominated final conjunction. Every candidate is
individually sound, and no optimality is claimed across mixed $\Always$ and
$\Always\Eventually$ branches, so
Theorem~\ref{thm:finite-window-gr1-optimality} states refinement optimality only
subproblem-wise.

\section{Computational Evaluation}
\label{app:scalability}
\renewcommand{\topfraction}{0.95}\renewcommand{\bottomfraction}{0.95}\renewcommand{\textfraction}{0.05}\renewcommand{\floatpagefraction}{0.75}

This appendix reports all measurements behind the empirical statements in the
paper. It can be read independently of the proofs. We first fix the
terminology and the measurement protocol (Sections~\ref{app:eval-terms}
and~\ref{app:eval-protocol}) and describe the benchmark families
(Section~\ref{app:eval-families}). The experiments are numbered E1 to E8 and,
for the stress test against automata, E6b. Each experiment answers one question.
\begin{itemize}
  \item E1 and E2 ask how the cost of one CAVE call depends on the number of
  eliminated variables, the depth of the finite window, and the number of liveness
  conjuncts (Section~\ref{app:eval-scaling}).
  \item E3 compares CAVE with eliminating the same variables after unrolling the
  specifications over a finite horizon (Section~\ref{app:eval-unrolling}).
  \item E4 tests, against an independent explicit-state oracle, whether the
  returned formulas are sound in practice and how often they are exact
  (Section~\ref{app:eval-precision}).
  \item E5 measures chains of compositions (Section~\ref{app:eval-chains}).
  \item E6 and E6b compare CAVE with exact elimination through B\"uchi automata
  (Section~\ref{app:eval-automata}).
  \item E7 isolates the effect of individual implementation choices
  (Section~\ref{app:eval-ablations}).
  \item E8 measures the composition and quotient computations on the automated
  emergency-braking (AEB) example (Section~\ref{app:scalability-aeb}).
\end{itemize}
Every experiment except E4 measures running time, and E4 makes no runtime claim.
All numbers describe the implementation over the tested ranges. They are not
asymptotic bounds, and the synthetic families are built to isolate individual
cost drivers, not to represent GR(1) specifications in general.

\paragraph{Summary of findings.}
\begin{itemize}
  \item The cost of one CAVE call is governed by the number $n_p(D_B+1)$ of
  eliminated indexed variables. On Boolean families it grows steeply and quantifier
  elimination (QE) dominates, so that eight irrelevant variables exceed the time limit.
  On linear real arithmetic it stays below $0.8\,\mathrm{s}$ over the tested range
  (E1, E2).
  \item The cost grows linearly with the number of liveness conjuncts, at about
  $9$ to $10\,\mathrm{ms}$ per liveness candidate in the substitution family (E1, E2).
  \item On the AEB composition, a finite-horizon unrolling baseline that ignores
  liveness is faster than CAVE for horizons $T\le 2$, slower from $T=3$, and
  exceeds the time limit at $T=6$ when it eliminates the production interface
  (E3).
  \item On $16000$ random classifications against an explicit-state oracle on
  bounded lasso words, no unsound result was observed. The refinement result was
  exact on $69.4\%$ and the relaxation result on $90.6\%$ of the instances (E4).
  \item Chains of up to eight stages are composed in $2.8\,\mathrm{s}$ or less
  (E5).
  \item Against exact Spot-based elimination, CAVE is slower when many variables are
  irrelevant or the window is deep, and faster when the number of liveness conjuncts
  grows. The CAVE results were sound but not exact on every instance that could be
  checked (E6, E6b).
  \item Computing the AEB contracts $\mathcal C_{SVD}$ and $\mathcal Q$ takes $0.93\,\mathrm{s}$, of which $0.48\,\mathrm{s}$ is spent inside CAVE (E8).
\end{itemize}

\subsection{Terminology and Measured Quantities}
\label{app:eval-terms}

A GR(1) specification is a conjunction of \emph{cores}, namely initial cores,
invariant cores (under $\Always$), and liveness cores (under
$\Always\Eventually$), as in Section~\ref{sec:preliminaries}.
A CAVE call eliminates a set $\bar p$ of \emph{irrelevant} variables from a
\emph{target} $\varphi$ in the presence of a \emph{context} $\Gamma$. The
\emph{refinement} call returns a $\bar p$-free GR(1) formula that implies
$\forall\bar p.(\Gamma\Rightarrow\varphi)$, and the \emph{relaxation} call returns a
$\bar p$-free GR(1) formula implied by $\exists\bar p.(\Gamma\wedge\varphi)$
(Section~\ref{sec:gr1-cave-bounds}). The following quantities appear throughout.
\begin{itemize}
  \item $n_p=|\bar p|$ is the number of irrelevant variables, and $n_I$ and $n_J$ are
  the numbers of liveness cores of the target and of the context.
  \item $D$ is the nesting depth of $\Next$ in the generated invariants, and
  $D_B$ is the depth of the finite window that CAVE actually uses
  (Section~\ref{sec:finite-window-cave}). A call flattens the problem over
  $D_B+1$ time points, so each irrelevant variable becomes $D_B+1$ indexed variables.
  We call $q:=|\bar p|(D_B+1)$ the \emph{quantified dimension}.
  \item Quantifier elimination (QE) is performed by Z3. Unless stated otherwise
  the tactic is \texttt{qe2}, and the \emph{QE share} is the fraction of the
  runtime spent inside the QE tactic.
  \item A \emph{liveness candidate} is one liveness core that a call must project.
  As in Section~\ref{sec:gr1-cave-bounds}, refinement builds one candidate per branch
  for each target liveness core, namely one for the synthetic branch and one for each
  support branch, which gives $n_I(1+n_J)$ candidates. Relaxation projects every context
  and target liveness core once, which gives $n_I+n_J$ candidates.
  \item \emph{Simplify} is the post-processing step of Pacti that removes redundant
  cores from a result. It is on by default.
  \item The size of a result is measured by its number of \emph{cores} and by the
  number of nodes of its Z3 term DAG (\emph{DAG nodes}).
  \item A \emph{Büchi automaton} accepts an infinite word if it visits an
  accepting state infinitely often, and every LTL formula can be translated to one.
  Existential quantification of a variable projects the variable out of the edge
  labels. Universal quantification needs complementation, which is exponential in
  the worst case~\cite{sistla1987complementation}.
\end{itemize}

\subsection{Measurement Protocol}
\label{app:eval-protocol}

\paragraph{Environment.}
All runs were made on one ARM64 Mac (Apple M5 Pro, 48\,GB of memory) under
macOS~26.7.1 with Python~3.11.16 and Z3~5.1.0. The comparisons with automata use
the Spot library~\cite{duret2022spot}, version~2.16 (\url{https://spot.lre.epita.fr}). The source
tree was clean, and its revision, the package versions, and the digests of the
benchmark scripts are recorded together with the raw per-run data.

\paragraph{Cold and warm calls.}
Every repetition runs in a freshly spawned worker process. The worker builds the
instance and then executes the measured operation twice. The first call (\emph{cold})
includes one-time costs of the process, and the second call (\emph{warm}) is the one
we report. Construction of the instance and all diagnostics, such as result sizes and
equivalence checks, lie outside the timed region. For every configuration with a warm
time of at least $0.05\,\mathrm{s}$ in E1 to E3, E5, E7, and E8, the cold and warm
medians differ by less than $7\%$, so the choice does not affect any conclusion.

\paragraph{Repetitions, limits, and statistics.}
We report medians over successful repetitions, with $10$ repetitions in E1, E3, E6,
E7, and E8, $5$ in E2 and E6b, and $3$ in E5. Each worker has a hard limit that
covers imports, preparation, both calls, and checks. It is $600\,\mathrm{s}$ in E1 to E3, E6, E7, and E8, and $300\,\mathrm{s}$ in E5 and E6b. A configuration that
exceeds the limit is a \emph{timeout} (T/O in the tables). Because the limit covers
both calls, a timeout means that the two calls together need more than the limit,
which for similar cold and warm times is more than about half of the limit per call.
Within one sweep the values ascend, and once all repetitions at a value (the first
three in E6b) time out, larger values are skipped, which is shown as ``--''.
Timeouts are censored observations and are never averaged into medians. The order of
the sweeps is shuffled with a fixed seed. E4 has no repetitions and no timing.

\subsection{Benchmark Families}
\label{app:eval-families}

Experiments E1 to E3 and E5 to E8 use the following instances. In all families the
parameters not being varied are fixed to $n_p=n_I=n_J=2$ and $D=1$, and every sweep
varies one parameter at a time.

\paragraph{Substitution family (E1).}
The irrelevant variables $p_0,\ldots,p_{n_p-1}$ are Boolean, and each is tied to a retained
Boolean $a_k$ over the window. Writing $P:=\bigwedge_{k=0}^{n_p-1}p_k$, the context is
\[
  \Gamma \;=\; \Always\!\Bigl(\bigwedge_{k=0}^{n_p-1}\bigwedge_{r=0}^{D}
  (\Next^{r}p_k\leftrightarrow\Next^{r}a_k)\Bigr)\;\wedge\;
  \bigwedge_{j=0}^{n_J-1}\Always\Eventually(P\leftrightarrow s_j),
\]
and the target is $\varphi=\bigwedge_{i=0}^{n_I-1}\Always\Eventually(P\leftrightarrow t_i)$,
with retained Booleans $s_j$ and $t_i$. Elimination is by substitution, so the family
isolates the cost of QE and of candidate construction. The sweeps are
$n_p\in\{1,2,4,6,8\}$, $D\in\{0,\ldots,5\}$, and $n_I,n_J\in\{1,2,4,6,8\}$.

\paragraph{Boolean-bound family (E2, E6, E6b, E7).}
Here each irrelevant variable $p_k$ is only bounded by retained Booleans $a_k$ and $b_k$.
For $k=0,\ldots,n_p-1$ and $r=0,\ldots,D$ the context has the invariant core
\[
  \Always\,\Next^{r}\bigl(((a_k\wedge\Next\neg a_k)\rightarrow p_k)\wedge
  (p_k\rightarrow(a_k\vee b_k))\bigr).
\]
The context liveness cores are $\Always\Eventually(\bigwedge_{k\in C_j}p_k\leftrightarrow s_j)$ and
the target liveness cores are $\Always\Eventually(\bigwedge_{k\in T_i}p_k\leftrightarrow t_i)$.
Here $T_i=\{(i+o)\bmod n_p : 0\le o<w_T\}$ and
$C_j=\{(j+1+o)\bmod n_p : 0\le o<w_C\}$ with $w_T=\max(1,\lfloor n_p/2\rfloor)$ and
$w_C=\min(n_p,w_T+1)$, so that distinct cores constrain different subsets of
$\bar p$. The target has no invariant. Since the invariant
contains one $\Next$, the window depth is $D_B=D+1$.

\paragraph{Linear-real family (E2).}
For $k=0,\ldots,n_p-1$ the real variables $u_k,d_k,v_k$ satisfy the context invariant
$\Always\,\Next^r\bigl(d_k/10\le u_k\le d_k/5\;\wedge\;
u_k-\tfrac12\le\Next v_k\le u_k+\tfrac12\bigr)$ for $r=0,\ldots,D$. Context liveness
core $j$ is $\Always\Eventually(u_{j \bmod n_p}\le(3+j)/2\wedge ms_j)$ and target liveness core $i$
is $\Always\Eventually(u_{i \bmod n_p}\ge(i\bmod 3)/2\wedge mt_i)$ with retained Booleans
$ms_j$ and $mt_i$. The target invariant is $\Always(\Next v_k\le 5+k)$ for each $k$.
Relaxation eliminates the $u_k$, and refinement eliminates the $u_k$ and the $v_k$, so
that the quantified dimension of refinement is twice as large.

\paragraph{Remark on the depth sweeps.}
In all three families the copies of an invariant shifted by $r\ge1$ are implied by
the copy with $r=0$ under $\Always$. The depth sweeps therefore measure how the
implementation reacts to the syntactic depth of the window, and not a change of the
logical content. We did not add a preprocessing step that removes such copies.

\paragraph{AEB and chain instances.}
E3, E7, and E8 use the AEB contracts of Section~\ref{sec:aeb-running-example}, and E5
uses the chain of Section~\ref{app:eval-chains}.

\subsection{Scaling with the Problem Parameters (E1 and E2)}
\label{app:eval-scaling}

Table~\ref{tab:e12-sweeps} lists the warm medians of all one-at-a-time sweeps for the
three families and both directions, and Figure~\ref{fig:e12-runtime} plots them.
Table~\ref{tab:e2-grid} gives the joint grid over $n_p$ and $D$ for the Boolean-bound
family.

\begin{table}[!tbp]
\centering
\caption{Warm median runtime in seconds of one CAVE call when one parameter is
varied and the others are fixed to $n_p=n_I=n_J=2$ and $D=1$. Ref.\ is refinement
and Rel.\ is relaxation. T/O marks a timeout.}
\label{tab:e12-sweeps}
\small
\setlength{\tabcolsep}{5pt}
\begin{tabular}{@{}llrrrrrr@{}}
\toprule
 & & \multicolumn{2}{c}{Substitution} & \multicolumn{2}{c}{Boolean-bound} & \multicolumn{2}{c}{Linear real} \\
\cmidrule(lr){3-4}\cmidrule(lr){5-6}\cmidrule(lr){7-8}
Parameter & Value & Ref. & Rel. & Ref. & Rel. & Ref. & Rel. \\
\midrule
$n_p$ & 1 & 0.0257 & 0.0342 & 0.0449 & 0.0419 & 0.06 & 0.0922 \\
 & 2 & 0.0573 & 0.0542 & 0.151 & 0.0669 & 0.0946 & 0.155 \\
 & 4 & 0.705 & 0.565 & 2.22 & 0.463 & 0.171 & 0.318 \\
 & 6 & 26.9 & 38.7 & 71.4 & 17.6 & 0.258 & 0.496 \\
 & 8 & T/O & T/O & T/O & T/O & 0.358 & 0.708 \\
\midrule
$D$ & 0 & 0.0263 & 0.0382 & 0.0622 & 0.0499 & 0.0635 & 0.093 \\
 & 1 & 0.0562 & 0.0541 & 0.152 & 0.0664 & 0.0946 & 0.149 \\
 & 2 & 0.168 & 0.131 & 0.318 & 0.116 & 0.126 & 0.213 \\
 & 3 & 0.703 & 0.571 & 0.75 & 0.28 & 0.158 & 0.279 \\
 & 4 & 3.7 & 3.81 & 2.36 & 0.974 & 0.191 & 0.356 \\
 & 5 & 26.9 & 39.1 & 8.5 & 4.85 & 0.225 & 0.439 \\
\midrule
$n_I$ & 1 & 0.0294 & 0.0399 & 0.0753 & 0.0523 & 0.0543 & 0.118 \\
 & 2 & 0.0575 & 0.0542 & 0.151 & 0.0681 & 0.0949 & 0.149 \\
 & 4 & 0.113 & 0.0903 & 0.294 & 0.108 & 0.171 & 0.218 \\
 & 6 & 0.166 & 0.135 & 0.44 & 0.154 & 0.253 & 0.298 \\
 & 8 & 0.221 & 0.186 & 0.584 & 0.211 & 0.329 & 0.388 \\
\midrule
$n_J$ & 1 & 0.0375 & 0.0396 & 0.0928 & 0.0485 & 0.0653 & 0.122 \\
 & 2 & 0.0568 & 0.0538 & 0.149 & 0.0683 & 0.0927 & 0.152 \\
 & 4 & 0.0947 & 0.0893 & 0.265 & 0.116 & 0.149 & 0.223 \\
 & 6 & 0.132 & 0.133 & 0.375 & 0.171 & 0.204 & 0.307 \\
 & 8 & 0.168 & 0.187 & 0.487 & 0.238 & 0.258 & 0.399 \\
\bottomrule
\end{tabular}
\end{table}

\begin{table}[!tbp]
\centering
\caption{Warm median runtime in seconds for the Boolean-bound family over the grid of
irrelevant variables $n_p$ (rows) and depth $D$ (columns), with $n_I=n_J=2$. T/O marks a
timeout, and ``--'' a configuration skipped after a smaller one timed out.}
\label{tab:e2-grid}
\footnotesize
\setlength{\tabcolsep}{3.5pt}
\begin{tabular}{@{}lrrrrr@{\hspace{1.2em}}rrrrr@{}}
\toprule
 & \multicolumn{5}{c}{Refinement, depth $D$} & \multicolumn{5}{c}{Relaxation, depth $D$} \\
\cmidrule(r){2-6}\cmidrule(l){7-11}
$n_p$ & 0 & 1 & 2 & 3 & 4 & 0 & 1 & 2 & 3 & 4 \\
\midrule
1 & 0.0326 & 0.0448 & 0.0668 & 0.0842 & 0.147 & 0.0377 & 0.0419 & 0.0466 & 0.0553 & 0.0702 \\
2 & 0.0625 & 0.151 & 0.314 & 0.74 & 2.37 & 0.0498 & 0.0667 & 0.116 & 0.28 & 0.976 \\
4 & 0.289 & 2.2 & 23.4 & T/O & -- & 0.0951 & 0.462 & 8.03 & T/O & -- \\
6 & 1.38 & 70.1 & T/O & -- & -- & 0.261 & 17 & T/O & -- & -- \\
\bottomrule
\end{tabular}
\end{table}

\begin{figure}[!tbp]
    \centering
    \begin{subfigure}[t]{0.48\linewidth}
        \centering
        \includegraphics[width=\linewidth]{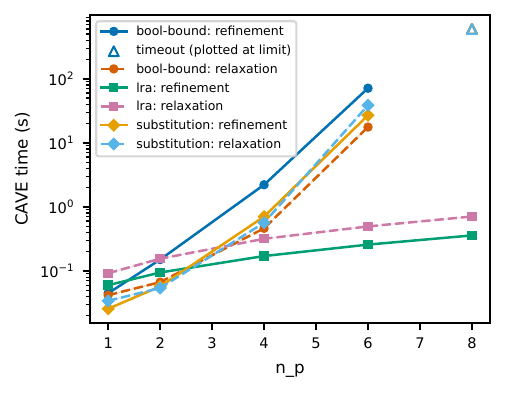}
        \caption{Irrelevant variables $n_p$.}
        \label{fig:e12-np}
    \end{subfigure}
    \hfill
    \begin{subfigure}[t]{0.48\linewidth}
        \centering
        \includegraphics[width=\linewidth]{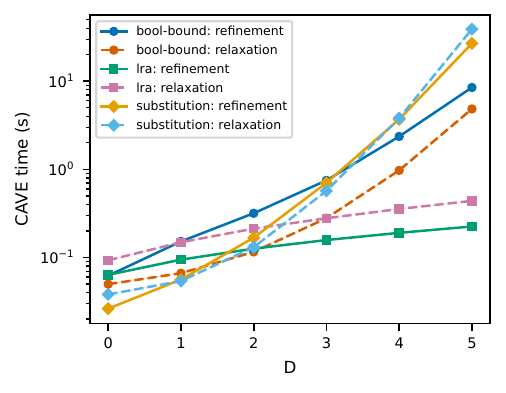}
        \caption{Depth $D$.}
        \label{fig:e12-D}
    \end{subfigure}

    \vspace{0.4em}

    \begin{subfigure}[t]{0.48\linewidth}
        \centering
        \includegraphics[width=\linewidth]{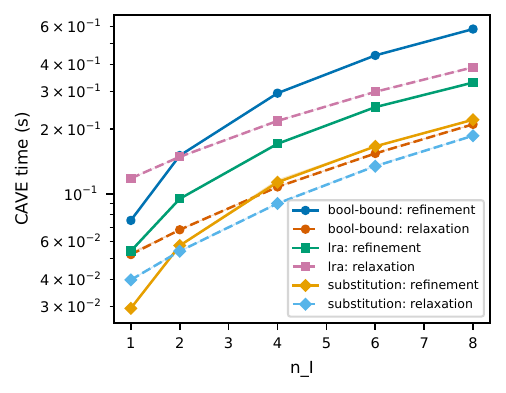}
        \caption{Target liveness cores $n_I$.}
        \label{fig:e12-nI}
    \end{subfigure}
    \hfill
    \begin{subfigure}[t]{0.48\linewidth}
        \centering
        \includegraphics[width=\linewidth]{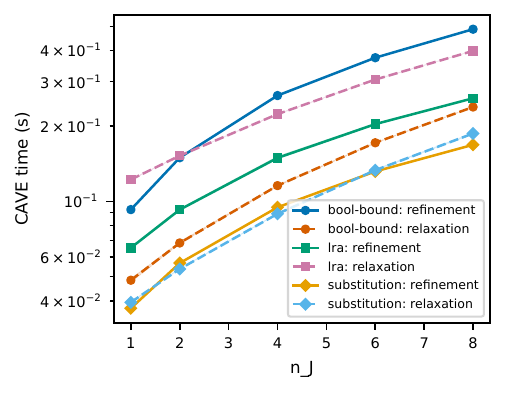}
        \caption{Context liveness cores $n_J$.}
        \label{fig:e12-nJ}
    \end{subfigure}
    \caption{Warm median CAVE time for the three families in both directions under
    one-at-a-time variation of each parameter, on logarithmic vertical axes.
    Triangles mark timeouts and are drawn at the limit.}
    \label{fig:e12-runtime}
\end{figure}

\paragraph{Irrelevant variables and depth.}
On the Boolean families the runtime grows steeply with $n_p$ and with $D$. In the
substitution family refinement takes $0.026\,\mathrm{s}$ for $n_p=1$,
$0.057\,\mathrm{s}$ for $n_p=2$, $0.70\,\mathrm{s}$ for $n_p=4$, and
$26.9\,\mathrm{s}$ for $n_p=6$, and for $n_p=8$ all ten repetitions time out.
Relaxation behaves alike ($0.034$, $0.054$, $0.565$, and $38.7\,\mathrm{s}$).
The depth sweep shows the same pattern, from $0.026\,\mathrm{s}$ at $D=0$ to
$26.9\,\mathrm{s}$ at $D=5$ for refinement and from $0.038\,\mathrm{s}$ to
$39.1\,\mathrm{s}$ for relaxation. The Boolean-bound family is more expensive
(refinement $2.2\,\mathrm{s}$ at $n_p=4$ and $71.4\,\mathrm{s}$ at $n_p=6$), and
Table~\ref{tab:e2-grid} shows where it ends. Refinement finishes up to
$(n_p,D)=(4,2)$ in $23.4\,\mathrm{s}$ and $(6,1)$ in $70.1\,\mathrm{s}$, while
$(4,3)$ and $(6,2)$ time out. Relaxation finishes $(4,2)$ in $8.0\,\mathrm{s}$ and
$(6,1)$ in $17.0\,\mathrm{s}$ and times out at the same two points.

\paragraph{The role of quantifier elimination.}
In the substitution family the QE share rises with the quantified dimension, from
$0.19$ at $n_p=1$ to $0.78$ at $n_p=6$ for refinement and from $0.18$ to above $0.99$
for relaxation, and it follows the same trend along the depth sweep
(Figure~\ref{fig:e12-qe}(a)). For the largest instances that completed, QE
therefore accounts for most of the time, and the growth is that of the QE problem.
Fitting $\log(\text{time})$ linearly against $q$ gives slopes between $0.38$ and
$0.71$ per unit of $q$, with $R^2$ between $0.91$ and $0.98$ (the fits for the
joint grid are drawn in Figure~\ref{fig:e12-qe}(b)). We report the fits only as a
description. The factor by which the runtime grows per additional quantified variable
is not constant. In the refinement sweep over $n_p$ it is $1.5$, $1.9$, and $2.5$ for
the steps from $q=2$ to $4$, $4$ to $8$, and $8$ to $12$, and it is $1.3$, $1.8$, and
$2.9$ for relaxation. We therefore make no claim about the base of an exponential.

\begin{figure}[!tbp]
    \centering
    \begin{subfigure}[t]{0.48\linewidth}
        \centering
        \includegraphics[width=\linewidth]{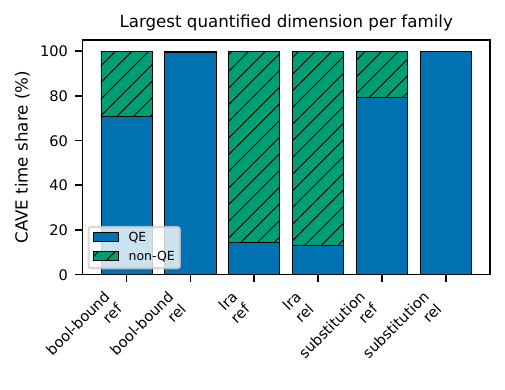}
        \caption{Share of the time spent in QE at the largest quantified dimension that
        completed, for each family and direction.}
        \label{fig:e12-qeshare}
    \end{subfigure}
    \hfill
    \begin{subfigure}[t]{0.48\linewidth}
        \centering
        \includegraphics[width=\linewidth]{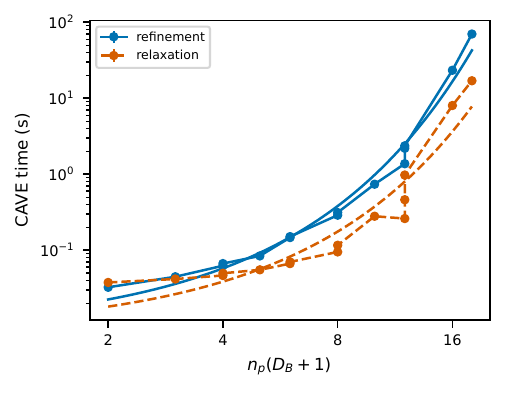}
        \caption{Warm CAVE time against $q=n_p(D_B+1)$ on the joint grid of the
        Boolean-bound family, with log-linear fits.}
        \label{fig:e12-q}
    \end{subfigure}
    \caption{Cost as a function of the quantified dimension.}
    \label{fig:e12-qe}
\end{figure}

\paragraph{Liveness cores.}
The runtime grows linearly with $n_I$ and $n_J$ and stays small. In the substitution
family refinement takes $0.029\,\mathrm{s}$ for $n_I=1$ and $0.221\,\mathrm{s}$ for
$n_I=8$, and relaxation takes $0.040\,\mathrm{s}$ and $0.186\,\mathrm{s}$. The
refinement time is proportional to the number $n_I(1+n_J)$ of liveness candidates
at $9.2$ to $9.8\,\mathrm{ms}$ per candidate (linear fits in $n_I$ and in $n_J$ have
$R^2=0.9999$), and the relaxation time grows by $0.021\,\mathrm{s}$ per added
liveness core of either kind. The Boolean-bound and linear-real families show the same
linear trend, for instance $0.075$ to $0.584\,\mathrm{s}$ and $0.054$ to
$0.329\,\mathrm{s}$ for refinement when $n_I$ goes from $1$ to $8$.

\paragraph{Linear real arithmetic.}
The linear-real family is benign over the tested range. The largest runtime is
$0.708\,\mathrm{s}$, at $n_p=8$ for relaxation, and no configuration times out.
Refinement grows approximately linearly from $0.060\,\mathrm{s}$ at $n_p=1$ to
$0.358\,\mathrm{s}$ at $n_p=8$ and from $0.064\,\mathrm{s}$ at $D=0$ to
$0.225\,\mathrm{s}$ at $D=5$, and QE accounts for $10$ to $14\%$ of the time.
The steep growth of the Boolean families is thus a property of QE over Boolean
variables, which Z3 handles by case splitting, and not of the temporal construction.

\subsection{Comparison with Finite-Horizon Unrolling (E3)}
\label{app:eval-unrolling}

\paragraph{Setting.}
The baseline eliminates the same variables without CAVE. It instantiates the initial
and invariant cores of the supervisor $\mathcal C_S$ and the vehicle $\mathcal C_V$ at
$T$ consecutive time steps, which introduces variables at $T+1$ time points, and
it calls Z3 QE twice on the resulting formulas, once for the universally quantified
implication from the assumptions and guarantees of the supervisor to the assumptions of
the vehicle (the refinement problem) and once for the existentially quantified
conjunction of the two guarantees (the relaxation problem). It is not an equivalent
computation. It drops all liveness cores and
constrains only $T$ steps, so its result is not a GR(1) contract, and it favors the
baseline because no liveness is processed. We use it to show how a monolithic QE call
scales, not to claim a semantic equivalence. Two variants differ in the guarantee
projection. The \emph{requested} variant eliminates all five irrelevant variables
$\{u,\MOV,\DCL,\STP,v\}$ from the guarantee. The \emph{interface-matched} variant
eliminates only $\{u,\MOV,\DCL,\STP\}$, because the production composition keeps
$v$ on its interface, and it is the primary comparison. The reference is the
production computation of $\mathcal C_S\parallel\mathcal C_V$ with CAVE, which does
not depend on $T$ and processes the liveness cores.

\begin{table}[!tbp]
\centering
\caption{Warm median runtime and result size of the unrolled baseline for horizon $T$
and of CAVE for the AEB composition $\mathcal C_S\parallel\mathcal C_V$. Size is the
number of Z3 DAG nodes. T/O is a timeout and ``--'' a skipped configuration.}
\label{tab:e3}
\small
\setlength{\tabcolsep}{6pt}
\begin{tabular}{@{}lrrrr@{}}
\toprule
 & \multicolumn{2}{c}{Interface-matched} & \multicolumn{2}{c}{Requested} \\
\cmidrule(lr){2-3}\cmidrule(lr){4-5}
Horizon $T$ & Time [s] & Size & Time [s] & Size \\
\midrule
1 & 0.00424 & 30 & 0.00411 & 16 \\
2 & 0.0305 & 130 & 0.0127 & 79 \\
3 & 0.371 & 217 & 0.0513 & 116 \\
4 & 15.1 & 313 & 0.241 & 170 \\
6 & T/O & -- & 10.1 & 649 \\
8 & -- & -- & T/O & -- \\
12 & -- & -- & -- & -- \\
16 & -- & -- & -- & -- \\
\midrule
\multicolumn{5}{@{}l}{CAVE (independent of $T$, liveness included)} \\
true & 0.308 & 138 & 5 \\
false & 0.124 & 140 & 17 \\
\bottomrule
\end{tabular}
\end{table}

\begin{figure}[!tbp]
    \centering
    \begin{subfigure}[t]{0.48\linewidth}
        \centering
        \includegraphics[width=\linewidth]{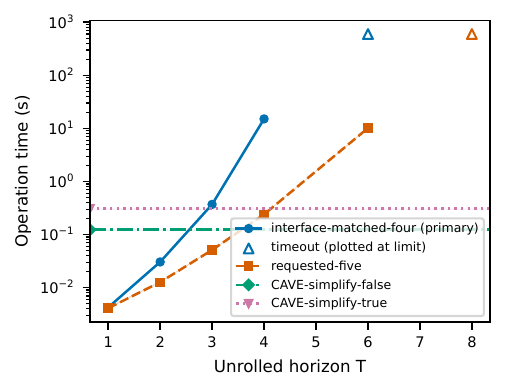}
        \caption{Runtime against the horizon $T$.}
        \label{fig:e3-time}
    \end{subfigure}
    \hfill
    \begin{subfigure}[t]{0.48\linewidth}
        \centering
        \includegraphics[width=\linewidth]{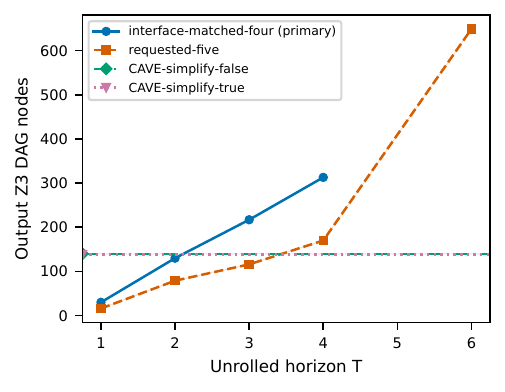}
        \caption{Result size against the horizon $T$.}
        \label{fig:e3-size}
    \end{subfigure}
    \caption{Unrolled baseline for both guarantee projections against the CAVE
    references with simplification on and off. Triangles mark timeouts.}
    \label{fig:e3}
\end{figure}

\paragraph{Results.}
Table~\ref{tab:e3} and Figure~\ref{fig:e3} show the following. CAVE takes
$0.308\,\mathrm{s}$ with the default simplification and $0.124\,\mathrm{s}$ without it,
independently of $T$. The interface-matched baseline takes $0.0042$, $0.031$,
$0.37$, and $15.1\,\mathrm{s}$ for $T=1,2,3,4$ and times out at $T=6$. It is faster
than CAVE for $T\le2$ and slower from $T=3$, both against the default
($0.37$ against $0.31\,\mathrm{s}$) and against the variant without simplification.
The requested variant is cheaper, with $0.24\,\mathrm{s}$ at $T=4$ and
$10.1\,\mathrm{s}$ at $T=6$, and times out at $T=8$. It is faster than the default
CAVE configuration up to $T=4$ and slower at $T=6$ ($T=5$ was not measured), and it is
slower than CAVE without simplification from $T=4$. QE accounts for $92$ to
$100\%$ of the baseline time. The size of the unrolled result grows with the horizon
(from $30$ to $313$ DAG nodes for $T=1,\ldots,4$ in the interface-matched variant
and to $649$ nodes at $T=6$ in the requested one), whereas the CAVE result has
$138$ nodes ($5$ cores) with simplification and $140$ nodes ($17$ cores) without.
The unrolled results cover only the first $T$ steps and no liveness, so they are not a
substitute for the CAVE result even where they are cheaper.

\subsection{Soundness and Precision on Random Instances (E4)}
\label{app:eval-precision}

\paragraph{Setting.}
Timing runs cannot establish soundness, so E4 compares the implementation with an
independent explicit-state oracle. An instance consists of a context and a target over
two relevant Boolean variables $q_1,q_2$ and $n_{\mathrm{int}}\in\{1,2\}$ irrelevant Boolean variables.
The context and the target each have an invariant core, an initial core with probability
$1/2$, and zero to two liveness cores. Every core is a random Boolean formula of nesting
depth at most three over $\neg,\wedge,\vee,\rightarrow,\leftrightarrow$ and the constants
$\True$ and $\False$, in which a variable may occur directly under one $\Next$. There are two sources of instances.
The \emph{uniform} source draws contexts without restriction. The \emph{nonblocking}
source keeps only contexts that pass a quick check for nonblocking behavior, explained
below. Each source and each value of $n_{\mathrm{int}}$ contributes $2000$ instances, which gives
$8000$ instances, and each is classified for refinement and for relaxation.

The oracle works on \emph{lasso words}, which are ultimately periodic words $u\,v^{\omega}$
over the four valuations of $q_1,q_2$ with $|u|\le2$ and $1\le|v|\le3$. There are
$1764$ such words. For one word it decides exactly, by exploring the product of
the word with all valuations of the irrelevant variables and analyzing its strongly connected components
under the liveness requirements, whether $\forall\bar p.(\Gamma\Rightarrow\varphi)$
holds (refinement) or $\exists\bar p.(\Gamma\wedge\varphi)$ holds (relaxation). We
compare the CAVE result with this answer on all $1764$ words and classify the
instance as follows.
\begin{itemize}
  \item \emph{Unsound} if for some word the refinement result holds while the exact
  answer fails, or the exact relaxation holds while the relaxation result fails.
  \item \emph{Exact} if the result agrees with the oracle on every word.
  \item \emph{Conservative} otherwise, that is, sound on all words but not exact.
\end{itemize}
A context is \emph{nonblocking} if for every word some run of the irrelevant variables satisfies it. On
a word without such a run the refinement answer is vacuously true. An instance is
\emph{trivial} if the CAVE result and the exact answer are the same constant on all
words, and \emph{nontrivial} otherwise. The nonblocking source samples contexts that are
feasible on a subset of the words ($|u|\le1$ and $|v|\le2$), and one of its $4000$
instances turned out to be blocking on the full set of words. All draws are
deterministic functions of a fixed seed.
Since only $1764$ words are examined, ``exact'' means exact on this suite, and the
absence of unsound observations is evidence and not a proof.

\begin{table}[!tbp]
\centering
\caption{Classification of CAVE results on $8000$ random instances per direction.
Here $n$ is the number of instances in the subset. The column ``Exact [\%]'' gives the
percentage with its Wilson $95\%$ interval. The unsound column gives the count and, in
parentheses, the Wilson upper bound in percent on the unsound fraction. Instances
that are neither exact nor unsound are conservative.}
\label{tab:e4}
\footnotesize
\setlength{\tabcolsep}{4pt}
\begin{tabular}{@{}lrrrrr@{}}
\toprule
Subset & $n$ & Exact & Exact [\%] & Conserv. & Unsound \\
\midrule
\multicolumn{6}{@{}l}{\emph{Refinement}} \\
All instances & 8000 & 5549 & 69.4 [68.3, 70.4] & 2451 & 0 ($\le$0.05) \\
Uniform source & 4000 & 1681 & 42.0 [40.5, 43.6] & 2319 & 0 ($\le$0.10) \\
Nonblocking source & 4000 & 3868 & 96.7 [96.1, 97.2] & 132 & 0 ($\le$0.10) \\
Context nonblocking & 4890 & 4737 & 96.9 [96.3, 97.3] & 153 & 0 ($\le$0.08) \\
Context blocking & 3110 & 812 & 26.1 [24.6, 27.7] & 2298 & 0 ($\le$0.12) \\
Nontrivial & 3788 & 1337 & 35.3 [33.8, 36.8] & 2451 & 0 ($\le$0.10) \\
Nontrivial, nonblocking & 1482 & 1329 & 89.7 [88.0, 91.1] & 153 & 0 ($\le$0.26) \\
Nontrivial, blocking & 2306 & 8 & 0.3 [0.2, 0.7] & 2298 & 0 ($\le$0.17) \\
\midrule
\multicolumn{6}{@{}l}{\emph{Relaxation}} \\
All instances & 8000 & 7249 & 90.6 [90.0, 91.2] & 751 & 0 ($\le$0.05) \\
Uniform source & 4000 & 3735 & 93.4 [92.6, 94.1] & 265 & 0 ($\le$0.10) \\
Nonblocking source & 4000 & 3514 & 87.8 [86.8, 88.8] & 486 & 0 ($\le$0.10) \\
Context nonblocking & 4890 & 4317 & 88.3 [87.4, 89.2] & 573 & 0 ($\le$0.08) \\
Context blocking & 3110 & 2932 & 94.3 [93.4, 95.0] & 178 & 0 ($\le$0.12) \\
Nontrivial & 3510 & 2759 & 78.6 [77.2, 79.9] & 751 & 0 ($\le$0.11) \\
Nontrivial, nonblocking & 2430 & 1857 & 76.4 [74.7, 78.1] & 573 & 0 ($\le$0.16) \\
Nontrivial, blocking & 1080 & 902 & 83.5 [81.2, 85.6] & 178 & 0 ($\le$0.35) \\
\bottomrule
\end{tabular}
\end{table}

\begin{figure}[!tbp]
    \centering
    \begin{subfigure}[t]{0.48\linewidth}
        \centering
        \includegraphics[width=\linewidth]{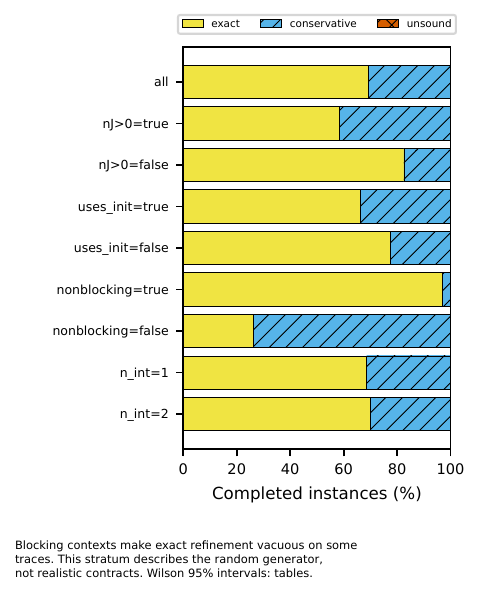}
        \caption{Refinement, all instances.}
        \label{fig:e4-ref-all}
    \end{subfigure}
    \hfill
    \begin{subfigure}[t]{0.48\linewidth}
        \centering
        \includegraphics[width=\linewidth]{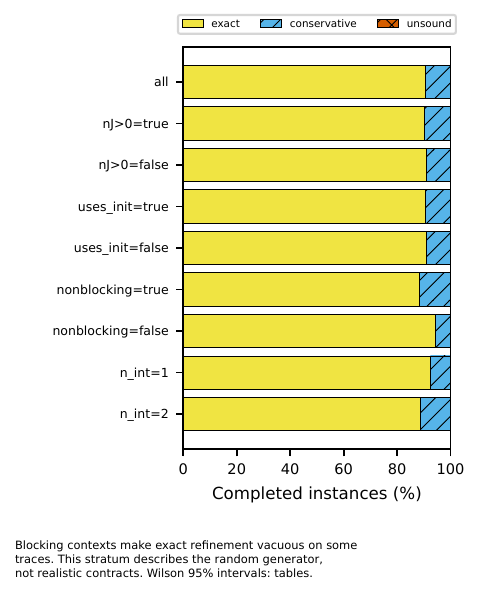}
        \caption{Relaxation, all instances.}
        \label{fig:e4-rel-all}
    \end{subfigure}

    \vspace{0.4em}

    \begin{subfigure}[t]{0.48\linewidth}
        \centering
        \includegraphics[width=\linewidth]{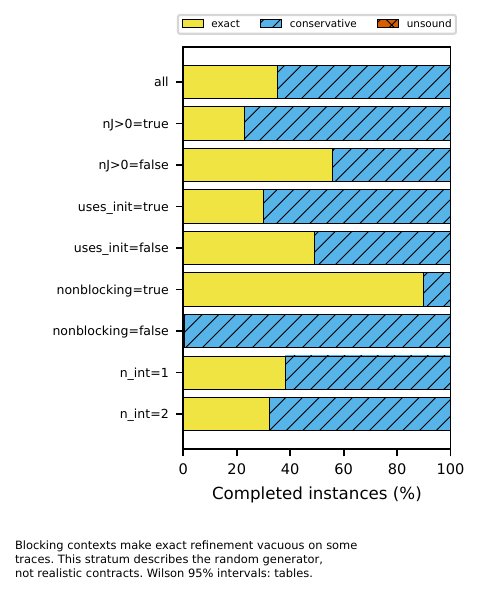}
        \caption{Refinement, nontrivial instances.}
        \label{fig:e4-ref-nt}
    \end{subfigure}
    \hfill
    \begin{subfigure}[t]{0.48\linewidth}
        \centering
        \includegraphics[width=\linewidth]{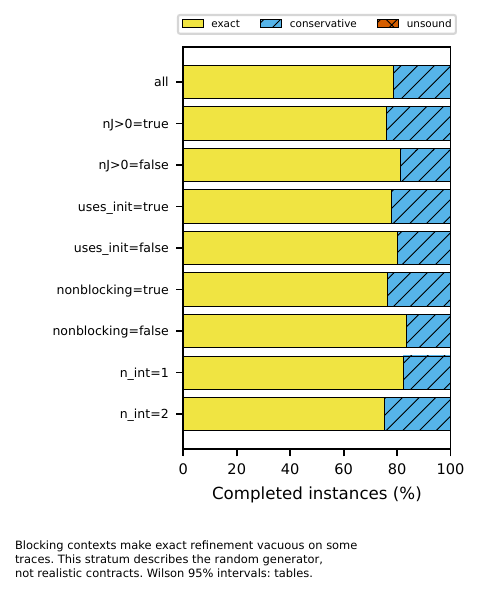}
        \caption{Relaxation, nontrivial instances.}
        \label{fig:e4-rel-nt}
    \end{subfigure}
    \caption{Share of exact, conservative, and unsound results per subset of the
    instances, with Wilson $95\%$ intervals.}
    \label{fig:e4}
\end{figure}

\paragraph{Results.}
No unsound result was observed in any of the $16000$ classifications, and none of the
$1764$ words per instance produced a violation. The Wilson upper bound on the unsound
fraction is $0.05\%$ for each direction. Refinement was exact on $5549$ of $8000$
instances ($69.4\%$) and conservative on the rest. The precision depends strongly on the
context. It is $96.9\%$ for nonblocking contexts and $26.1\%$ for blocking contexts, and
among the nontrivial instances it is $89.7\%$ for nonblocking contexts and $0.3\%$
($8$ of $2306$) for blocking ones. Blocking contexts make the exact answer vacuously
true on some words, and the CAVE result does not exploit this vacuity. Relaxation was
exact on $7249$ instances ($90.6\%$) and on $78.6\%$ of the nontrivial ones, and it does not improve with nonblocking contexts
($88.3\%$ against $94.3\%$ for blocking ones). When a result is conservative, it differs from
the exact answer on a median fraction of $0.49$ of the words for refinement and $0.17$
for relaxation. The two values of $n_{\mathrm{int}}$ give similar numbers
($68.6\%$ and $70.2\%$ exact for refinement, $92.4\%$ and $88.8\%$ for relaxation).
The instances are random and small and do not represent engineering specifications, so
the percentages show where conservatism arises and not how often it occurs in practice.

\subsection{Chains of Compositions (E5)}
\label{app:eval-chains}

\paragraph{Setting.}
Stage $i$ is a contract with the real input $y_{i-1}$ and the outputs $y_i$ (real) and
$m_i$ (Boolean). Its assumptions are $0\le y_{i-1}\le100$ initially and invariantly and
$\Always\Eventually(y_{i-1}\le50)$. Its guarantees are $y_i=y_{i-1}$ initially, the
invariant $y_{i-1}-\tfrac12\le\Next y_i\le y_{i-1}+\tfrac12\wedge(m_i\rightarrow\Next y_i\ge1)$,
and $\Always\Eventually(y_i\le50)$ and $\Always\Eventually m_i$. Composing consecutive
stages eliminates the connecting variable. We measure a \emph{left fold}
$(((\mathcal C_1\parallel\mathcal C_2)\parallel\mathcal C_3)\cdots)$ for $N\le8$ and a
\emph{balanced} composition of contiguous halves for $N\in\{4,8\}$. The time is the sum
of the durations of the composition steps. As a diagnostic, a bounded check
confirms that the safety part of the result is satisfiable and that the composed
assumption restricts $y_0$ to $[\tfrac{N-1}{2},\,100-\tfrac{N-1}{2}]$, which is the range
in which a worst-case drift of $\tfrac12$ per stage keeps every intermediate signal in
$[0,100]$. Both checks succeed for every $N$ (they are not realizability checks).

\begin{table}[!tbp]
\centering
\caption{Warm median runtime of the chain composition, the part spent inside CAVE,
the QE share, and the size of the composed contract.}
\label{tab:e5}
\small
\setlength{\tabcolsep}{6pt}
\begin{tabular}{@{}lrrrrrr@{}}
\toprule
Order & Stages $N$ & Time [s] & CAVE [s] & QE share [\%] & Cores & DAG nodes \\
\midrule
left-fold & 2 & 0.143 & 0.0794 & 7.8 & 7 & 36 \\
left-fold & 3 & 0.367 & 0.21 & 7.3 & 8 & 46 \\
left-fold & 4 & 0.652 & 0.367 & 6.7 & 9 & 56 \\
left-fold & 5 & 1.03 & 0.571 & 6.2 & 10 & 66 \\
left-fold & 6 & 1.49 & 0.808 & 5.8 & 11 & 76 \\
left-fold & 8 & 2.77 & 1.46 & 5.0 & 13 & 96 \\
\midrule
balanced & 4 & 0.572 & 0.321 & 6.9 & 9 & 56 \\
balanced & 8 & 1.79 & 0.961 & 5.8 & 13 & 106 \\
\bottomrule
\end{tabular}
\end{table}

\begin{figure}[!tbp]
    \centering
    \begin{subfigure}[t]{0.48\linewidth}
        \centering
        \includegraphics[width=\linewidth]{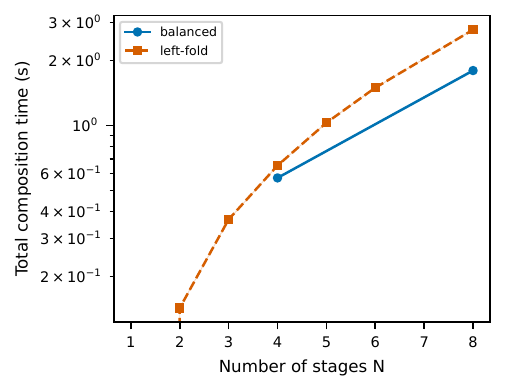}
        \caption{Total composition time against $N$.}
        \label{fig:e5-seconds}
    \end{subfigure}
    \hfill
    \begin{subfigure}[t]{0.48\linewidth}
        \centering
        \includegraphics[width=\linewidth]{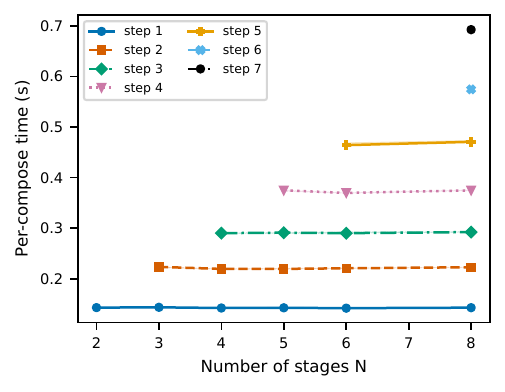}
        \caption{Time of each composition step in the left folds.}
        \label{fig:e5-steps}
    \end{subfigure}
    \caption{Chains of compositions.}
    \label{fig:e5}
\end{figure}

\paragraph{Results.}
Eight stages are composed in $2.77\,\mathrm{s}$ by a left fold, and
$N=2,3,4,5,6$ take $0.14$, $0.37$, $0.65$, $1.03$, and $1.49\,\mathrm{s}$. The cost of
one step depends on the size of the accumulated contract and not on $N$. In the left
fold of eight stages the seven steps take $0.14$, $0.22$, $0.29$, $0.38$, $0.47$,
$0.57$, and $0.69\,\mathrm{s}$, and from two stages on the contract grows by one core and by $10$ DAG
nodes per stage. The total time is therefore superlinear in $N$ but moderate. A balanced
composition reduces the time to $0.57\,\mathrm{s}$ for $N=4$ ($12\%$ less) and
$1.79\,\mathrm{s}$ for $N=8$ ($35\%$ less), because it composes smaller contracts,
at the price of a somewhat larger result for $N=8$ ($106$ against $96$ DAG nodes at the
same $13$ cores). CAVE accounts for $53$ to $57\%$ of the time and QE only for $5.0$ to
$7.8\%$, so the remaining time is spent in the contract operations around the
elimination, which we did not break down further.

\subsection{Comparison with Exact Automata-Based Elimination (E6 and E6b)}
\label{app:eval-automata}

\paragraph{The automata pipeline.}
Temporal quantification can be computed exactly with automata~\cite{qptl}. For
refinement the exact result is $\forall\bar p.(\Gamma\Rightarrow\varphi)\equiv
\neg\exists\bar p.(\Gamma\wedge\neg\varphi)$. Our baseline translates
$\Gamma\wedge\neg\varphi$ to a nondeterministic B\"uchi automaton with Spot, projects
the irrelevant variables out of the edge labels, which realizes $\exists\bar p$, and
complements the result, which realizes the negation. The reported time is that of
these three steps. Afterwards, outside the timed region and only for small instances,
we compare the CAVE result with the exact automaton. The CAVE result is \emph{sound}
if the product of its automaton with the projected automaton has an empty language, that is,
if it implies the exact result, and \emph{exact} if the languages are equal.
The output of the pipeline is an automaton. It is not a GR(1) formula with initial,
invariant, and liveness parts, so it could not be used as a contract in Pacti without a
further translation, which we do not attempt. Spot is the only automata-based tool in
the comparison, and we use one pipeline without tuning. Other complementation
constructions or deterministic automata could change the constants.

\paragraph{E6, small instances.}
E6 runs the Boolean-bound family with $n_I=n_J=1$, $n_p\in\{1,2,3\}$, and $D\in\{0,1\}$.
Table~\ref{tab:e6} shows that Spot is faster on every one of these small
instances, by a factor of $4.4$ to $16.8$ (the automata have at most $16$ states, and the
complements at most $9$). The CAVE result is sound in all $60$ runs and is not exact in
any of them.

\begin{table}[!tbp]
\centering
\caption{E6, warm median runtime of CAVE and of the Spot pipeline on the
Boolean-bound family with $n_I=n_J=1$. States are those of the projected automaton
and of its complement. Sound and exact refer to the CAVE result.}
\label{tab:e6}
\small
\setlength{\tabcolsep}{6pt}
\begin{tabular}{@{}rrrrrrll@{}}
\toprule
$n_p$ & $D$ & CAVE [s] & Spot [s] & States & Compl.\ states & Sound & Exact \\
\midrule
1 & 0 & 0.0121 & 0.00121 & 4 & 3 & yes & no \\
1 & 1 & 0.0165 & 0.00146 & 4 & 3 & yes & no \\
2 & 0 & 0.0222 & 0.00207 & 8 & 6 & yes & no \\
2 & 1 & 0.0469 & 0.00278 & 8 & 6 & yes & no \\
3 & 0 & 0.0428 & 0.00975 & 16 & 9 & yes & no \\
3 & 1 & 0.168 & 0.0128 & 16 & 9 & yes & no \\
\bottomrule
\end{tabular}
\end{table}

\paragraph{E6b, a stress family for the automata pipeline.}
The first comparison is tiny and favors Spot. We therefore designed E6b to probe the regime in
which complementation is expensive, and we disclose that this family was
designed after we had seen the E6 results. It is the Boolean-bound family with
base parameters $n_p=2$, $D=1$, $n_I=n_J=1$, and five sweeps that vary
$n_I$, $n_J$, $k$ (where $n_I=n_J=k$), $n_p\in\{1,\ldots,8\}$, and $D\in\{0,\ldots,6\}$ one
at a time. The liveness counts take the values $1,2,3,4,5,6,8$ (and $k\in\{1,2,3,4,6,8\}$).
Each point is measured five times with a limit of $300\,\mathrm{s}$ for the worker, and a
point is skipped after three timeouts. Ten points with small automata are additionally checked
for soundness and exactness. The run produced $360$ rows, none of which is an error.

\begin{table}[!tbp]
\centering
\caption{E6b, warm median runtime of CAVE and of the Spot pipeline when the number of
liveness conjuncts grows (base $n_p=2$, $D=1$). Compl.\ states and sets are those of the
complemented automaton. T/O is a timeout of the $300\,\mathrm{s}$ worker, and ``--'' a
skipped point.}
\label{tab:e6b-liveness}
\small
\setlength{\tabcolsep}{6pt}
\begin{tabular}{@{}llrrrr@{}}
\toprule
Parameter & Value & CAVE [s] & Spot [s] & Compl.\ states & Compl.\ sets \\
\midrule
$n_I$ & 1 & 0.0458 & 0.0027 & 6 & 2 \\
 & 2 & 0.0943 & 0.00645 & 49 & 4 \\
 & 3 & 0.135 & 0.0375 & 508 & 6 \\
 & 4 & 0.188 & 1.12 & 6754 & 8 \\
 & 5 & 0.231 & 62.3 & 122254 & 10 \\
 & 6 & 0.283 & T/O & -- & -- \\
 & 8 & 0.358 & -- & -- & -- \\
\midrule
$n_J$ & 1 & 0.0471 & 0.00275 & 6 & 2 \\
 & 2 & 0.076 & 0.00589 & 19 & 2 \\
 & 3 & 0.107 & 0.0182 & 150 & 2 \\
 & 4 & 0.131 & 0.356 & 1891 & 2 \\
 & 5 & 0.175 & 14.7 & 31066 & 2 \\
 & 6 & 0.195 & T/O & -- & -- \\
 & 8 & 0.248 & -- & -- & -- \\
\midrule
$k$ & 1 & 0.0471 & 0.00267 & 6 & 2 \\
 & 2 & 0.158 & 0.0489 & 513 & 4 \\
 & 3 & 0.32 & T/O & -- & -- \\
 & 4 & 0.542 & -- & -- & -- \\
 & 6 & 1.14 & -- & -- & -- \\
 & 8 & 2 & -- & -- & -- \\
\bottomrule
\end{tabular}
\end{table}

\begin{table}[!tbp]
\centering
\caption{E6b, warm median runtime of CAVE and of the Spot pipeline when the number of
irrelevant variables or the depth grows (base $n_I=n_J=1$). States are those of the
projected automaton before complementation, and compl.\ states those of its complement.}
\label{tab:e6b-hidden}
\small
\setlength{\tabcolsep}{6pt}
\begin{tabular}{@{}llrrrr@{}}
\toprule
Parameter & Value & CAVE [s] & Spot [s] & States & Compl.\ states \\
\midrule
$n_p$ & 1 & 0.017 & 0.00122 & 4 & 3 \\
 & 2 & 0.047 & 0.00269 & 8 & 6 \\
 & 3 & 0.176 & 0.0122 & 16 & 9 \\
 & 4 & 0.706 & 0.0919 & 32 & 17 \\
 & 5 & 2.52 & 0.77 & 64 & 17 \\
 & 6 & 20.3 & 0.9 & 128 & 17 \\
 & 7 & 137 & 7.84 & 256 & 17 \\
 & 8 & T/O & T/O & -- & -- \\
\midrule
$D$ & 0 & 0.0224 & 0.002 & 8 & 6 \\
 & 1 & 0.0471 & 0.00266 & 8 & 6 \\
 & 2 & 0.103 & 0.0036 & 8 & 6 \\
 & 3 & 0.237 & 0.00477 & 8 & 6 \\
 & 4 & 0.742 & 0.00593 & 8 & 6 \\
 & 5 & 2.68 & 0.0072 & 8 & 6 \\
 & 6 & 11.8 & 0.00875 & 8 & 6 \\
\bottomrule
\end{tabular}
\end{table}

\begin{figure}[!tbp]
    \centering
    \begin{subfigure}[t]{0.48\linewidth}
        \centering
        \includegraphics[width=\linewidth]{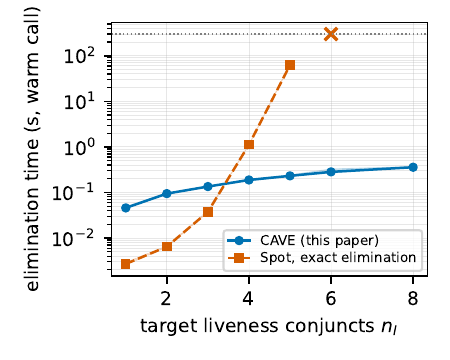}
        \caption{Target liveness $n_I$.}
        \label{fig:e6b-nI}
    \end{subfigure}
    \hfill
    \begin{subfigure}[t]{0.48\linewidth}
        \centering
        \includegraphics[width=\linewidth]{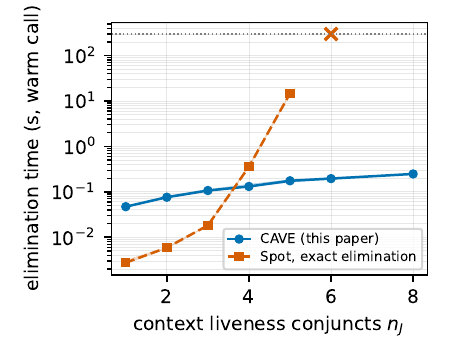}
        \caption{Context liveness $n_J$.}
        \label{fig:e6b-nJ}
    \end{subfigure}

    \vspace{0.4em}

    \begin{subfigure}[t]{0.48\linewidth}
        \centering
        \includegraphics[width=\linewidth]{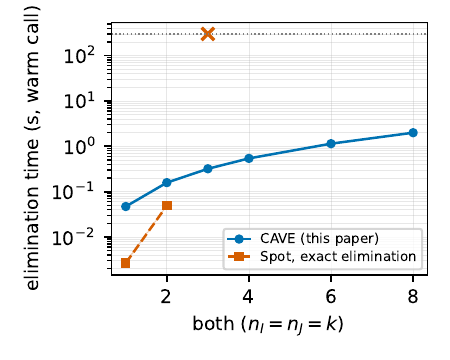}
        \caption{Both, $n_I=n_J=k$.}
        \label{fig:e6b-k}
    \end{subfigure}
    \hfill
    \begin{subfigure}[t]{0.48\linewidth}
        \centering
        \includegraphics[width=\linewidth]{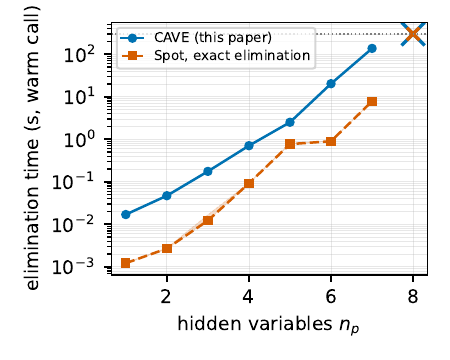}
        \caption{Irrelevant variables $n_p$.}
        \label{fig:e6b-np}
    \end{subfigure}

    \vspace{0.4em}

    \begin{subfigure}[t]{0.48\linewidth}
        \centering
        \includegraphics[width=\linewidth]{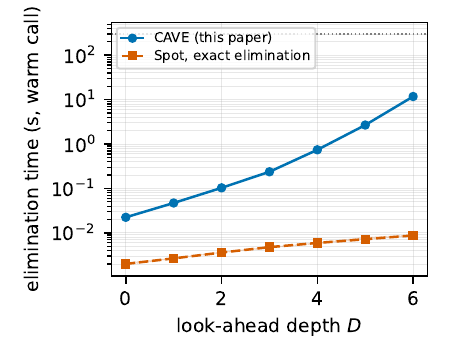}
        \caption{Depth $D$.}
        \label{fig:e6b-D}
    \end{subfigure}
    \caption{E6b, warm median time of CAVE and of exact elimination with Spot, with
    the minimum to maximum range over the repetitions. Crosses mark points at which the first
    repetitions reached the $300\,\mathrm{s}$ limit.}
    \label{fig:e6b}
\end{figure}

\paragraph{E6b results for liveness conjuncts.}
The cost of the pipeline is governed by complementation (Table~\ref{tab:e6b-liveness}).
Translation and projection stay below $0.1\,\mathrm{s}$ for $n_I\le5$, but the complement
has $6$, $49$, $508$, $6754$, and $122254$ states for $n_I=1,\ldots,5$, so the number of
states grows by a factor that rises from $8$ to $18$ per added conjunct, and its number of
acceptance sets is $2n_I$. At $n_I=5$ complementation takes almost all of the $62.3\,\mathrm{s}$,
and $n_I=6$ exceeds the limit. CAVE grows linearly from $0.046$ to $0.358\,\mathrm{s}$ over
$n_I=1,\ldots,8$ with $n_I$ cores and $52$ to $300$ DAG nodes in the result. The curves cross
between $n_I=3$ ($0.038\,\mathrm{s}$ for Spot against $0.135\,\mathrm{s}$ for CAVE) and
$n_I=4$ ($1.12$ against $0.19\,\mathrm{s}$), and at $n_I=5$ CAVE is $269$ times faster.
Context liveness behaves alike. The complement grows from $6$ to $31066$ states while its
acceptance sets stay at two, Spot needs $0.356\,\mathrm{s}$ at $n_J=4$ and $14.7\,\mathrm{s}$
at $n_J=5$ and times out at $6$, whereas CAVE needs $0.248\,\mathrm{s}$ at $n_J=8$. When both
counts grow together, Spot is faster at $k=2$ ($0.049$ against $0.158\,\mathrm{s}$) and
times out at $k=3$, where CAVE needs $0.32\,\mathrm{s}$, and CAVE reaches $k=8$ in
$2.0\,\mathrm{s}$.

\paragraph{E6b results for irrelevant variables and depth.}
The picture reverses (Table~\ref{tab:e6b-hidden}). Spot handles $n_p=7$ irrelevant
variables in $7.8\,\mathrm{s}$, mostly in the translation ($7.1\,\mathrm{s}$) and with a
complement that stays at $17$ states, while CAVE needs $137\,\mathrm{s}$, which is $17.5$
times as long, and $22.6$ times as long at $n_p=6$ ($20.3$ against $0.90\,\mathrm{s}$).
The automaton has $2^{n_p+1}$ states, and for $n_p=8$ both approaches exceed the
limit. In the three repetitions of Spot at $n_p=8$ the translation took
$44$ to $52\,\mathrm{s}$, and in one of them the cold call finished after $295\,\mathrm{s}$, of which
$248\,\mathrm{s}$ were complementation, although complementation took $0.14\,\mathrm{s}$ at
$n_p=7$. We do not explain this jump. The CAVE point at $n_p=7$ is close to the limit as
well, since the worker needed $268$ to $279\,\mathrm{s}$ of the $300\,\mathrm{s}$ for its two calls. In CAVE the
result grows with $n_p$ (from $52$ to $3132$ DAG nodes between $n_p=2$ and $6$) and
the QE share rises from $0.26$ to $0.70$. Along the depth sweep CAVE grows from
$0.022\,\mathrm{s}$ at $D=0$ to $11.8\,\mathrm{s}$ at $D=6$, while the automaton is
unchanged ($8$ states) and Spot needs between $0.002$ and $0.009\,\mathrm{s}$.
Since the shifted invariants are redundant (Section~\ref{app:eval-families}), the
automaton is unaffected, whereas the cost of CAVE grows with the window.

\paragraph{Soundness and exactness of the CAVE result.}
Table~\ref{tab:e6b-checks} lists the ten checked points, which include
$n_I$ and $n_J$ up to $4$ and $n_p$ up to $6$. The CAVE result is sound at every point and
equivalent to the exact result at none of them. Together with E6 the comparison thus covers
$15$ distinct instances, and on each of them the CAVE result implies the exact result
without being equivalent to it. We did not determine whether the exact results can be expressed as GR(1)
formulas, which is why we do not attribute the gap to the restriction to GR(1). The points with the
largest Spot cost could not be checked.

\begin{table}[!tbp]
\centering
\caption{E6b, soundness and exactness checks of the CAVE result against the exact
Spot result. States are those of the projected automaton and of its complement.}
\label{tab:e6b-checks}
\small
\setlength{\tabcolsep}{4pt}
\begin{tabular}{@{}rrrrrrrrll@{}}
\toprule
$n_I$ & $n_J$ & $n_p$ & $D$ & CAVE [s] & Spot [s] & States & Compl.\ states & Sound & Exact \\
\midrule
1 & 1 & 2 & 1 & 0.0472 & 0.00313 & 8 & 6 & yes & no \\
2 & 1 & 2 & 1 & 0.0948 & 0.00683 & 12 & 49 & yes & no \\
1 & 2 & 2 & 1 & 0.0754 & 0.00603 & 12 & 19 & yes & no \\
2 & 2 & 2 & 1 & 0.154 & 0.0516 & 20 & 513 & yes & no \\
3 & 1 & 2 & 1 & 0.139 & 0.0417 & 16 & 508 & yes & no \\
1 & 3 & 2 & 1 & 0.103 & 0.0184 & 16 & 150 & yes & no \\
4 & 1 & 2 & 1 & 0.188 & 1.14 & 20 & 6802 & yes & no \\
1 & 4 & 2 & 1 & 0.13 & 0.347 & 20 & 1891 & yes & no \\
1 & 1 & 4 & 1 & 0.703 & 0.099 & 32 & 17 & yes & no \\
1 & 1 & 6 & 1 & 19.6 & 0.907 & 128 & 17 & yes & no \\
\bottomrule
\end{tabular}
\end{table}

\paragraph{Conclusion of the comparison.}
Neither method dominates. The automata pipeline returns the exact result and is
faster when many variables are irrelevant or the window is deep. Its cost is
dominated by complementation, which grows by orders of magnitude with each added
liveness conjunct, and CAVE is faster in that regime. The family of E6b is adversarial
for the automata pipeline by construction and should not be read as a typical workload.

\subsection{Implementation Ablations (E7)}
\label{app:eval-ablations}

E7 changes one implementation choice at a time in refinement CAVE and measures the
effect on the AEB composition $\mathcal C_S\parallel\mathcal C_V$ and on the Boolean-bound
instance with $n_p=n_I=n_J=2$ and $D=1$. The variants are the \texttt{qe} and
\texttt{qe\_rec} QE tactics of Z3 in place of the default \texttt{qe2}, no simplification of the
result, and construction of the synthetic-branch candidates only, which omits all
support-branch candidates. We check whether the variant returns an equivalent result by comparing it with the
production result bucket by bucket with Z3 and a limit of $500\,\mathrm{ms}$ per query. This is a
sufficient condition, so ``inconclusive'' means only that no certificate was found and
not that the results differ. On the AEB composition it is also inconclusive when the
production result is compared with itself, so inconclusive carries no information
there.

\begin{table}[!tbp]
\centering
\caption{E7, warm median runtime, result size, and equivalence check against the
production result for each variant of refinement CAVE.}
\label{tab:e7}
\footnotesize
\setlength{\tabcolsep}{4pt}
\begin{tabular}{@{}llrrrl@{}}
\toprule
Instance & Variant & Time [s] & Cores & DAG nodes & Equivalence \\
\midrule
AEB composition & production (\texttt{qe2}) & 0.305 & 5 & 138 & inconclusive \\
 & \texttt{qe} & 0.486 & 5 & 303 & equivalent \\
 & \texttt{qe\_rec} & 0.301 & 5 & 136 & equivalent \\
 & simplify off & 0.124 & 17 & 140 & inconclusive \\
 & synthetic only & 0.173 & 4 & 66 & inconclusive \\
\midrule
Boolean family & production (\texttt{qe2}) & 0.213 & 2 & 31 & equivalent \\
 & \texttt{qe} & 0.124 & 2 & 27 & equivalent \\
 & \texttt{qe\_rec} & 0.108 & 2 & 23 & equivalent \\
 & simplify off & 0.15 & 2 & 94 & equivalent \\
 & synthetic only & 0.104 & 2 & 32 & equivalent \\
\bottomrule
\end{tabular}
\end{table}

\begin{figure}[!tbp]
    \centering
    \includegraphics[width=0.5\linewidth]{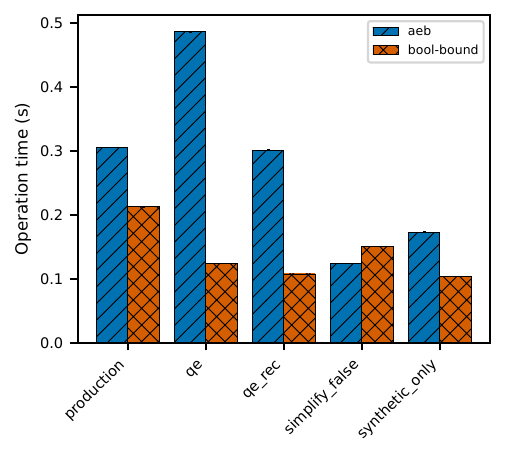}
    \caption{E7, warm median runtime of the variants on the AEB composition and on the
    Boolean-bound instance.}
    \label{fig:e7}
\end{figure}

\paragraph{Results.}
The default tactic \texttt{qe2} is not uniformly the fastest. On the AEB composition
\texttt{qe} is $1.6$ times slower ($0.486$ against $0.305\,\mathrm{s}$) with an equivalent
result, and \texttt{qe\_rec} takes $0.301\,\mathrm{s}$ with an equivalent result. On the
Boolean-bound instance \texttt{qe\_rec} is about twice as fast as the default
($0.108$ against $0.213\,\mathrm{s}$) and \texttt{qe} $1.7$ times, both with equivalent results.
Disabling simplification makes the AEB composition $2.5$ times faster ($0.124\,\mathrm{s}$)
but returns $17$ cores instead of $5$, so the time is saved at the price of a larger and
less readable result. Constructing only the synthetic candidate takes $0.173\,\mathrm{s}$ on
the AEB composition and returns $4$ cores instead of $5$, so some support-based core is
missing, and it takes $0.104\,\mathrm{s}$ on the Boolean-bound instance, where the result is
equivalent and the support candidates do not contribute.

\subsection{AEB Performance (E8)}
\label{app:scalability-aeb}

E8 measures the AEB computation with the same instrumented implementation.    The timed region consists of computing $\mathcal C_{SVD}$, by composing $\mathcal C_S$ and $\mathcal C_V$ and then composing the result with $\mathcal C_D$, and of computing the quotient $\mathcal Q=\mathcal C_T/\mathcal C_V$, both with the production operations. These are the two contracts of Table~\ref{tab:aeb-operation-results}. Construction of the input contracts and bookkeeping lie outside the timed region. Both operations use the default simplification. The cumulative CAVE time and QE time
are the times spent in these parts of the operations, and counters are
taken from the instrumentation. Each of the $10$ repetitions runs in a fresh worker.

\begin{table}[!tbp]
\centering
\caption{E8, warm medians over $10$ runs for $\mathcal C_{SVD}$, $\mathcal Q$, and
their sum.}
\label{tab:aeb-performance}
\small
\setlength{\tabcolsep}{6pt}
\begin{tabular}{@{}lrrr@{}}
\toprule
Metric & Composition & Quotient & Combined \\
\midrule
Runtime [s] & 0.718 & 0.208 & 0.927 \\
CAVE runtime [s] & 0.408 & 0.0767 & 0.485 \\
QE runtime [s] & 0.251 & 0.00221 & 0.254 \\
\midrule
CAVE invocations & 8 & 3 & 11 \\
Refinement / relaxation & 2/6 & 2/1 & 4/7 \\
\texttt{Solver.check()} calls & 189 & 117 & 306 \\
\texttt{Optimize.check()} calls & 42 & 28 & 70 \\
QE applications & 20 & 3 & 23 \\
Liveness targets & 1 & 3 & 4 \\
Candidates & 3 & 5 & 8 \\
Max.\ quantified atoms & 11 & 3 & 11 \\
Max.\ quantified indexed variables & 19 & 10 & 19 \\
\midrule
Output cores & 5 & 7 & 12 \\
Output DAG nodes & 124 & 133 & 257 \\
\bottomrule
\end{tabular}
\end{table}

The composition and the quotient take $0.718\,\mathrm{s}$ and $0.208\,\mathrm{s}$, which is $0.927\,\mathrm{s}$ for both (Table~\ref{tab:aeb-performance}). The time spent inside CAVE is $0.485\,\mathrm{s}$, which is $52\%$ of the total, and QE takes $0.254\,\mathrm{s}$, which is $27\%$ of the total and $52\%$ of the CAVE time. The computation invokes CAVE $11$ times ($4$ refinements and $7$ relaxations) and performs $306$ solver checks, $70$ optimization checks, and $23$ QE applications. The largest QE problem has $19$ indexed variables. The results have $5$ and $7$ cores and $124$ and $133$ DAG nodes. The worker needs about $156\,\mathrm{MiB}$ of resident memory at peak, which includes the loaded libraries. The first step of the composition, $\mathcal C_S\parallel\mathcal C_V$, is also timed in E3 and E7 ($0.308\,\mathrm{s}$ and $0.305\,\mathrm{s}$), which leaves roughly $0.4\,\mathrm{s}$ of the composition time to the step with $\mathcal C_D$.


\subsection{Limitations of the Evaluation}
\label{app:eval-limits}

All measurements come from one machine, one Python implementation, and one version of
Z3 and of Spot, and absolute times depend on them. The synthetic families isolate
single cost drivers, and their depth sweeps add redundant copies of the same invariant, so they
show sensitivity to the syntactic window depth. Timeouts are censored, and a timeout
here means that the cold and warm calls together needed more than the limit. E4 uses random small
instances, bounded lasso words, and an oracle that is itself software, so it is evidence
for soundness and not a proof. The automata comparison uses a single Spot pipeline, and its
stress family was designed after the first comparison. No other tool known to us eliminates irrelevant variables when computing the composition or quotient of contracts whose assumptions and guarantees are temporal formulas, so E3 and E6 compare CAVE with the closest alternative constructions and not with competing tools. Finally, in E8 about half of the time is spent outside CAVE, in the contract operations that surround it, and we did not profile these further.
